\documentclass[12pt]{article}
\usepackage{amsmath}
\usepackage{graphicx}
\usepackage{enumerate}
\usepackage{natbib}
\usepackage{url} % not crucial - just used below for the URL 
\usepackage{amsthm}
\usepackage{tikz}
\usepackage{setspace}

\usepackage[margin=0.85in]{geometry}

\usepackage{graphicx}
\usepackage{caption}
\usepackage{subcaption}
\usepackage{float}
\usepackage{booktabs}
\usepackage{makecell}
\usepackage{diagbox}
\usepackage{tabularx,booktabs}
\usepackage{amsfonts,amssymb}
\usepackage{mathtools}
\usepackage{bm}
\usepackage{bbm}         % \mathbbm{1}
\usepackage{mleftright}  % optional

\usepackage{enumitem}
\usepackage{setspace}
\usepackage{titlesec}

\usepackage[algoruled,lined]{algorithm2e}

\usepackage{xr}
\usepackage{hyperref} % load late
\hypersetup{
  colorlinks = true,
  urlcolor   = black,
  linkcolor  = black,
  citecolor  = black
}
\usepackage{pdfpages}
\usepackage{verbatim}
\usepackage{lastpage}
\usepackage{xcolor}
\usepackage{soul}
\usepackage{endnotes}
\usepackage{authblk}

\def \cF{\mathcal{F}}
\def \cP{\mathcal{P}}

\def \cA{\mathcal{A}}

\def \M{\mathcal{M}}

\def \G{\mathcal{G}}

\def \E{\textrm{E}}

\def \1{\mathds{1}}

\def \mR{\mathbb{R}}

\def \X{\mathcal{X}}

\def \MV{\textrm{MV}}

\def \V{\mathcal{V}}

\def \RRMV{\textrm{RRMV}}

\def \Var{\textrm{Var}}

\def \tg{\textrm{tg}}

\def \tr{\textrm{tr}}

\def \Cov {\textrm{Cov}}
\def \diag {\textrm{diag}}

\def\tr{\operatorname{tr}}

\def\Var{\textrm{Var}}

\def\MV{\textrm{MV}}

\def \SR{\textrm{SR}}
\def\E{\mathbb{E}}

\def\mR{\mathbb{R}}

\def\1{\mathds{1}}

\def\cF{\mathcal{F}}

\def\A{\mathbf{A}}

\def\D{\mathbf{D}}

\def\G{\mathbf{G}}

\def\I{\mathbf{I}}

\def\M{\mathbf{M}}
\def\P{\mathbf{P}}
\def\Q{\mathbf{Q}}
\def\R{\mathbf{R}}
\def\S{\mathbf{S}}

\def\V{\mathbf{V}}

\def\X{\mathbf{X}}

\def\Z{\mathbf{Z}}

\def\a{\mathbf{a}}
\def\b{\mathbf{b}}
\def\c{\mathbf{c}}
\def\d{\mathbf{d}}
\def\e{\mathbf{e}}

\def\h{\mathbf{h}}

\def\u{\mathbf{u}}
\def\w{\mathbf{w}}
\def\x{\mathbf{x}}
\def\y{\mathbf{y}}
\def\z{\mathbf{z}}

\def\bs{\boldsymbol}
\def\bmu{{\bs{\mu}}} 
\def\bsigma{{\bs{\Sigma}}}

\def\BFzero{\bs{0}}

\def \bdelta{\bs{\delta}}

\usepackage{natbib}
 \bibpunct[, ]{(}{)}{,}{a}{}{,}%

\newtheorem{theorem}{Theorem}
\newtheorem{lemma}{Lemma} 
\newtheorem{proposition}{Proposition} 

\newtheorem{corollary}{Corollary}

\newtheorem{assumption}{Assumption}

\usepackage{hyperref}

\def\ECEquationsNumberedThrough{%
  \setcounter{equation}{0}%
  \def\theequation{S\arabic{equation}}%
  \def\theHequation{S.\arabic{equation}}%
}

\begin{document}
\title{\bf On Reference-Regulated Multiperiod \\
Mean-Variance Portfolio Optimization \\
in High Dimensions}
    \author[1]{Yutao DENG}
    \author[2]{Jianjun GAO}
    \author[1]{Weichen WANG$^*$}
    \affil[1]{Faculty of Business and Economics, The University of Hong Kong}    
    \affil[2]{School of Information Management and Engineering, Shanghai University of Finance and Economics}
    \date{}

\maketitle

%\bigskip
\begin{abstract}
The multiperiod mean-variance (MV) portfolio optimization serves as a vital expansion of Markowitz's static MV portfolio selection framework. Just like its static counterpart, the multiperiod MV portfolio remains susceptible to estimation errors. We propose a reference-regulated multiperiod mean-variance (RRMV) framework that penalizes deviations from a reference policy. Therefore, this new optimization successfully combines the advantages of dynamic strategies and reference portfolios. A key contribution of this paper is the characterization of the out-of-sample Sharpe ratio under high-dimensional asymptotics with estimation errors in both the mean vector and the covariance matrix. We show how the reference penalty and the investment horizon jointly affect the optimized portfolio performance, and how regularization operates differently from the single-period portfolio optimization. Extensive simulation and real data studies demonstrate that the proposed framework improves the stability and out-of-sample Sharpe ratios of multiperiod policies significantly.
\end{abstract}


\onehalfspacing

\section{Introduction}

The single-period mean-variance (MV) portfolio optimization initiated by \cite{markowitz1952harry} laid the foundation of modern portfolio theory. Owing to its parsimonious structure---requiring only the mean vector and covariance matrix of asset returns---the MV model has been widely adopted in both academic research and investment practice \citep{BRANDT2010-book,Kolm:2014}. A natural and important extension of the classical single-period formulation is its multiperiod (dynamic) generalization, where portfolio decisions are updated over time to incorporate evolving information and thereby potentially enhance long-run performance.

Multiperiod MV problems, however, can be formulated in fundamentally different ways. Some formulations such as \cite{GarleanuPedersen:2013JOF, DeMiguel-JFQA-2015} which optimize a sequence of single-period MV tradeoffs, effectively maximizing a discounted sum of period-by-period mean–variance objectives. In contrast, \cite{LiNg:2000MF} directly extend Markowitz’s original framework by focusing on the mean and variance of cumulative wealth (return). This terminal-return formulation preserves the global risk–return tradeoff of the classical MV criterion and yields several theoretical advantages. In particular, the resulting dynamic policy structure allows the multiperiod MV portfolio to attain an improved MV efficient frontier, thereby generating an additional risk premium relative to static strategies \citep{ChiuZhou2011, vigna2014efficiency, Yao02082016}. Moreover, unlike the static MV portfolio---whose terminal wealth is a linear transformation of risky asset returns---the dynamic strategy of \cite{LiNg:2000MF} fundamentally alters the distributional properties of terminal wealth. 
Under normally distributed returns, terminal wealth follows a quasi-lognormal distribution, reflecting nonlinear compounding through rebalancing. 
More recently, \cite{van2021surprising} demonstrate that this class of dynamic portfolio policies exhibits notable robustness to model misspecification in return dynamics, further highlighting its structural advantages.

Despite these appealing theoretical properties, it is well known that the classical static MV portfolio is highly sensitive to estimation error and may perform poorly out of sample (see, e.g., \citealp{demiguel2009optimal, BRANDT2010-book}). This sensitivity becomes increasingly severe in high-dimensional settings. In contemporary investment environments, portfolios often contain hundreds of assets, whereas the available historical data remain limited. As a result, estimation error in expected returns and covariance matrices can substantially distort portfolio weights and undermine performance \citep{Ao2019, MengCaoWang2025}.
Motivated by this inherent fragility of the MV framework, our objective is to retain the structural advantages of the dynamic multiperiod formulation while mitigating its vulnerability to estimation error. 
Building on the hybrid portfolio literature \citep{kanOptimalPortfolioChoice2007, KanWangZhou-allrisky-2022, KanWang-MS-2024, Lassance-MS-2024}, we develop a reference–regulated multiperiod MV model that penalizes deviations of the dynamic portfolio from a benchmark policy over time. This penalty serves as an economically interpretable form of dynamic regularization, stabilizing portfolio decisions and reducing sensitivity to estimation noise.

Specifically, we consider the following optimization objective:
\begin{align*}
\omega\cdot\Var(X_T) - \E[X_T] + \omega\cdot\sum_{t=0}^{T-1} \E\big[\|\u_t - X_t \w_t \|_{\Q_t}^2 \big],
\end{align*}
where $X_t$ and $\u_t$ denote the investor’s wealth and portfolio holdings at time $t$, respectively, evolving according to a self-financing process; $\w_t$ is a pre-specified reference portfolio weight; and $\Q_t$ is a positive-definite matrix governing the strength and direction of regulation. The quadratic penalty controls deviations of the dynamic allocation from the benchmark in a state-dependent manner. The framework is flexible: any benchmark portfolio with desirable stability properties---such as shrinkage-based or diversified static allocations---can serve as the reference. The formulation is closely related to norm-based constraints and regularization techniques that stabilize portfolio weights under estimation error \citep{demiguel2009generalized}. It is also connected to shrinkage approaches, as the matrices $\Q_t$ and reference portfolios $\w_t$ effectively pull the dynamic allocation toward structured targets, analogous to shrinking estimated means and covariances toward more stable benchmarks. 
In this sense, our framework generalizes portfolio combination methods (e.g., \cite{kanOptimalPortfolioChoice2007,KanWangZhou-allrisky-2022}) to a dynamic, path-dependent setting.

Our contributions are multifold. 
First, we provide a closed-form characterization of the reference-regulated multiperiod MV policy.  
By adapting the embedding technique of \cite{LiNg:2000MF}, we resolve the nonseparability of the variance term in dynamic programming and establish that the optimal policy admits a recursive representation governed by a finite set of market-dependent parameters.

Second, we derive high-dimensional asymptotic characterizations of the out-of-sample Sharpe ratio (SR) for these MV portfolios under estimation risk. 
Different from the analysis in \citet{kanOptimalPortfolioChoice2007,KanWangZhou-allrisky-2022, Lassance-MS-2024}, we study the model in a high-dimensional asymptotic regime where the number of assets $p$ grows proportionally with the sample size $n$, so that $p/n \to c \in (0,\infty)$, as in \citet{Ao2019,MengCaoWang2025}. 
This setting captures large portfolios with limited data and differs fundamentally from classical asymptotics where $p$ is fixed and $n \to \infty$. Using tools from random matrix theory, we characterize the limiting behavior of the Sharpe ratio under this high-dimensional regime. Our analysis separates the effects of estimation error in the mean vector and covariance matrix, thereby quantifying how different sources of statistical uncertainty affect portfolio performance. We further show how reference-based regularization modifies the limiting Sharpe ratio by attenuating the impact of estimation noise, providing a theoretical explanation for the stabilizing role of regularization.

Third, we extend high-dimensional Sharpe ratio analysis from static portfolios to multiperiod policies. 
In the static setting, the regulation matrix $\Q_t$ primarily shrinkage the estimation error in the covariance matrix $\hat\bsigma$, while the reference component $\w_t$ is used to mitigate estimation error in the mean vector $\hat\bmu$.
However, estimation errors in the mean and covariance no longer operate independently in dynamic environments. 
Instead, they interact through feedback trading, jointly influencing the out-of-sample Sharpe ratio. We show that reference-based regulation affects not only its primary error channel but also the secondary one through this intertemporal interaction, highlighting a distinct and previously unexplored role for regularization in dynamic portfolio policies. Using this approach, given mild condition, the Sharpe ratio can be further enhanced when the investment period $T>1$. 

Finally, we evaluate the empirical performance of the proposed framework through simulations and real-data studies. Consistent with the theoretical mechanisms identified above, the reference-regulated approach delivers more stable Sharpe ratios and improved out-of-sample performance relative to unregularized multiperiod strategies and leading static regularization methods, particularly in high-dimensional settings. Overall, our findings provide a unified theoretical and empirical perspective on how regularization can be incorporated into multiperiod mean–variance optimization to mitigate estimation risk while preserving the structural benefits of dynamic portfolio selection.

\subsection{Related Literature}\label{sse_literature}

Our work is closely related to a broad literature that seeks to improve the out-of-sample performance of mean–variance portfolios. 
A fundamental difficulty of the classical MV framework is its sensitivity to estimation error, which often leads to unstable portfolio weights and poor out-of-sample results \citep{ChoraZiema1993,demiguel2009optimal}.
These concerns have motivated a wide range of methods designed to reduce the impact of estimation error on portfolio selection.

One major approach to improving out-of-sample portfolio performance focuses on enhancing the estimation of return moments. 
For covariance estimation, shrinkage methods provide well-conditioned alternatives to the sample covariance matrix. 
The linear shrinkage estimator of \citet{ledoit2004well} and its nonlinear extensions \citep{ledoit2012nonlinear} yield better-conditioned and more accurate covariance estimates than the traditional sample covariance matrix. 
Beyond shrinkage, additional structure can be imposed through factor and sparsity-based models \citep{bickel2008regularized,lam2009sparsistency,cai2011adaptive,fan2013large,fan2016overview,cai2020high}.
For expected returns, Bayesian and equilibrium-based approaches incorporate additional structure beyond the sample mean, including the Black–Litterman model \citep{BlackLitterman1991} and related variants \citep{Lai2011,BertsimasGuptaPaschalidis2012}. 

A second approach stabilizes portfolio performance by regularizing the portfolio decision directly. \citet{Jagannathan_jof2003} show that no-short-sale constraints effectively shrink the estimated covariance matrix through the associated Lagrange multipliers, leading to more stable portfolios in practice. 
In a comprehensive empirical study, \citet{demiguel2009optimal} demonstrate that many optimized portfolios fail to consistently outperform the naive ``$1/N$'' benchmark, highlighting the severity of estimation error in portfolio choice. 
Focusing on minimum-variance portfolios, \citet{demiguel2009generalized} propose a generalized framework based on norm constraints that links constrained optimization to covariance shrinkage. 
\citet{fan2012vast} provide theoretical support for gross-exposure constraints that limit estimation error in high-dimensional portfolios. 
Robust optimization approaches similarly account for parameter uncertainty by protecting against worst-case model misspecification \citep{Goldfarb:2003,BlanchetChenZhou:2022}.

Another related strand of the literature studies how estimation error affects portfolio decisions and out-of-sample performance metrics, such as MV utility and the Sharpe ratio, and how these effects can be corrected. \citet{kanOptimalPortfolioChoice2007} characterize the impact of estimation error on out-of-sample MV utility and propose hybrid portfolio rules that combine different estimators; \citet{KanWangZhou-allrisky-2022} extend this framework to portfolios of risky assets only. Subsequent works develop portfolio combination methods with improved out-of-sample stability \citep{Tu2011,KanWang-MS-2024,Lassance-MS-2024}. Because the sample size and portfolio dimension jointly affect out-of-sample performance, \citet{Lassance2025JFQA} derive results on the optimal portfolio size. The effect of estimation error on the out-of-sample performance becomes even more pronounced in high-dimensional settings, where the number of assets is large relative to the available sample size \citep{demiguel2009optimal,Karoui2010}. \citet{Ao2019} propose a sparse regression representation of MV portfolios that improves out-of-sample performance for large portfolios. 
More recently, \citet{MengCaoWang2025} develop a random matrix theory-based method to consistently estimate and correct the out-of-sample Sharpe ratio in high-dimensional portfolio optimization, enabling reliable parameter tuning under general covariance structures.

While existing approaches focus primarily on single-period settings, an important open question is how estimation risk should be managed in multiperiod portfolio problems, where decisions interact over time.
A foundational theoretical formulation of multiperiod mean–variance (MMV) framework is developed by \cite{LiNg:2000MF}. 
Their embedding approach overcomes the non-separability induced by the variance term and thereby resolves the main difficulty in applying dynamic programming. 
This framework is commonly referred to as a \emph{pre-committed} policy. 
Building on this methodology, a variety of extensions have been investigated, including continuous-time counterparts \citep{Zhou2000}, models with no-shorting constraints \citep{CuiGaoLiLi:2014}, and no-bankruptcy constraints \citep{BieleckiJinZhou}. 
A parallel strand studies time-consistent (equilibrium) mean–variance policies that address the dynamic inconsistency of pre-committed solutions; see, for example, \citealp{basak2010dynamic,bjork2014mean}. 
However, \citet{Vigna2020} and \citet{van2021distribution} show that time-consistent policies typically achieve lower theoretical Sharpe ratios than pre-committed solutions, since time consistency effectively imposes additional constraints on the original MMV problem. 
Moreover, \citet{van2021surprising} demonstrate that pre-committed policies may exhibit greater robustness under model misspecification.

Prior studies have provided a rich understanding of dynamic optimal portfolio behavior, but they have largely examined MMV policies under known model parameters and have offered limited guidance on how to control estimation risk when return moments must be estimated from data. 
This omission is particularly important because estimation errors in MMV settings affect not only initial allocations but also subsequent feedback decisions, potentially amplifying performance deterioration over time. We aim to explicitly quantify the effect of estimation error to the asymptotic out-of-sample Sharpe ratio in the MMV setting. 

\subsection{Paper Organization and Notations} \label{ssse:notation}

The remainder of this paper is organized as follows. 
% Section~\ref{sse_literature} reviews the related literature. 
Section~\ref{se_model_and_problem} introduces the reference-regulated MMV portfolio model and derives the associated optimal policy. 
Section~\ref{se:stat} provides statistical justification for the use of regularization and analyzes the asymptotic properties of the out-of-sample Sharpe ratio. 
Section~\ref{se_numerical_evaluation} presents a series of numerical experiments to evaluate the performance of the RRMV model. Finally, we conclude the paper in Section~\ref{sec:Conclusion}.
 
We end this section by introducing some notations used throughout the paper. All vectors and matrices are denoted by boldface letters. 
The notation $\e$ represents the all-ones vector of appropriate dimension, and $\x^{\top}$ denotes the transpose of the vector $\x$. For a symmetric matrix $\bsigma$, the notation $\bsigma \succ 0$ indicates that $\bsigma$ is positive definite. 
For any $\x \in \mathbb{R}^p$, the weighted $\ell_2$-norm is defined as $\| \x \|_{\Q}^2 = \x^{\top} \Q \x$, where $\Q = \textrm{diag}(q_1, q_2, \ldots, q_p) \succ 0$. 
We also denote $\| \x \|$ as the $\ell_2$ norm and $\| \x \|_1$ as the $\ell_1$ norm. 
The inner product between two vectors $\x, \y \in \mathbb{R}^p$ is denoted as $\langle \x, \y \rangle = \x^{\top}\y$. 
Finally, the notation $\bm{x} \sim \mathcal{N}(\bmu, \bsigma)$ indicates that the random vector $\bm{x}$ follows a multivariate normal distribution with mean $\bmu$ and covariance matrix $\bsigma$.

%****************************************** new section *********************************
\section{Market Models and Regulated Multiperiod MV Portfolio}\label{se_model_and_problem}

We consider a capital market comprising $p$ risky assets and a single risk-free asset, which can be traded at discrete time points $t = 0, 1, \ldots, T$, where $T \geq 1$ denotes the finite investment horizon. Let $r > 0$ denote the deterministic risk-free return during the periods of $t = 0, \ldots, T-1$. The random return vector of the $p$ risky assets in period $t$ is denoted by $\R_t \triangleq \big( R_t^{(1)}, \ldots, R_t^{(p)} \big)^{\top}$, for $t = 0, 1, \ldots, T-1$.\endnote{For the $i$-th risky asset, $S_t^{(i)}$ represents its price at time $t$, and we define the return as $R_t^{(i)} \triangleq S_{t+1}^{(i)} / S_t^{(i)}$.} This notation implies that the return vector $\R_t$ becomes observable at time $t+1$. We assume that the sequence of random returns $\{\R_t\}|_{t=0}^{T-1}$ is statistically independent across time.\endnote{The portfolio policy developed in this paper can be extended to the case of temporally correlated returns by employing filtration-based conditional expectations; see \citealp{FollmerSchied:2004}.} For each $t = 0, \ldots, T-1$, let $\E[\R_t]$ and $\Cov[\R_t] \triangleq \E[\R_t \R_t^{\top}] - \E[\R_t]\E[\R_t]^{\top}$ denote, respectively, the expected return vector and covariance matrix of $\R_t$, both of which are assumed to be constant for the $T$ periods. Moreover, $\Cov[\R_t]$ is positive definite for all $t$.

\subsection{Multiperiod MV portfolio optimization}\label{sse_market_model}

An investor dynamically allocates the wealth across available assets in each period under the self-financing policy. Let $X_t$ denote the realized wealth at time $t$ and the initial wealth is $X_0$. The portfolio allocation at the beginning of period $t$ is represented by $\u_t = \big(u_t^{(1)}, \ldots, u_t^{(p)}\big)^{\top}$, where $u_t^{(i)}$ is the amount invested in the $i$-th risky asset. The remaining wealth, $X_t - \e^{\top} \u_t$, is allocated to the risk-free asset. The wealth then evolves as
\begin{align} \label{def_wealth_dyn_rf}
X_{t+1} = r \cdot \big(X_t - \e^{\top} \u_t \big) + \R_t^{\top} \u_t
= r X_t + \P_t^{\top} \u_t\,,
\end{align}
for $t = 0, \ldots, T-1$, where $\P_t = \big(P_t^{(1)}, \ldots, P_t^{(p)}\big)^{\top} \triangleq \R_t - r \e$ denotes the vector of excess returns over the risk-free rate. 
For notational simplicity, we assume that the distribution of $\R_t$ is time-invariant. In particular, the expected value and covariance matrix of $\P_t$ are defined as
\begin{align*}
\bmu \triangleq \E[\P_t] &= \E[\R_t] - r \e \,, \\
\bsigma \triangleq \Cov[\P_t] &= \Cov[\R_t] \,,
\end{align*}
for $t = 0, \ldots, T-1$, respectively. In the subsequent analysis, we will refer to $\bmu$ and $\bsigma$ as the true expected return and covariance matrix, to distinguish them from empirical estimates obtained from historical data.

In this work, we investigate the multiperiod mean–variance (MMV) portfolio optimization problem, formulated as $\cP_{\MV}(\omega,T)$ (see \citealp{LiNg:2000MF}) : 
\begin{equation*}
    \begin{aligned}
    &\min_{ \{\u_t\}|_{t=0}^{T-1} } ~ \omega \Var[X_T] - \E[X_T] \\
    & \text{subject to}~\{X_t, \u_t\} ~\text{satisfy}~ (\ref{def_wealth_dyn_rf}) \,,
    \end{aligned}   
\end{equation*}
where $\omega > 0$ is the risk aversion parameter controlling the return-risk tradeoff. Problem $\cP_{\MV}(\omega,T)$ optimizes the balance between the mean and variance of terminal wealth and can be regarded as a natural extension of the single-period MV model of \citet{markowitz1952harry}. Under this notation, the single-period case corresponds to $\cP_{\MV}(\omega,1)$. Ideally, if all model parameters are correctly estimated, $\cP_{\MV}(\omega,T)$ yields the globally optimal dynamic policy over $T$ periods that maximizes $\E[X_T] - \omega \Var[X_T]$. Such a policy may outperform both the repeated application of a myopic one-period MV strategy and the static MV policy that treats the entire $T$-period horizon as a single decision period. Indeed, as shown by \citet{ChiuZhou2011} and \citet{vigna2014efficiency}, the MMV framework can generate an additional intertemporal diversification premium.
However, this conclusion should be interpreted with caution. It is well known that the out-of-sample performance of static MV portfolios (i.e., the policy associated with $\cP_{\MV}(\omega,1)$) is highly sensitive to estimation errors \citep{demiguel2009optimal}. This sensitivity persists in the multiperiod setting as well, as demonstrated by the examples presented in Section~\ref{se:stat}.

To mitigate this limitation and inspired by the notion of stable portfolio strategies, we propose a hybrid framework that integrates a stable reference policy into the dynamic MV optimization model. We refer to this formulation as the reference-regulated mean–variance (RRMV) portfolio optimization model $\cP_{\RRMV}(\omega, T)$:
\begin{equation*}
    \begin{aligned}
    &\min_{\u_t} ~\omega \Var(X_T) - \E[X_T] +\omega \sum_{t=0}^{T-1}  
    \E\big[\|\u_t - X_t \w_t \|_{\Q_t}^2 \big] \\
    &\textrm{subject to}
    ~ \{X_t, \u_t\}~\textrm{satisfies (\ref{def_wealth_dyn_rf})}\,,
    \end{aligned}   
\end{equation*}
where the input sequence $\{\w_t \in \mathbb{R}^p\}|_{t=0}^{T-1}$ denotes the reference policy, and $\{\Q_t \succ 0\}|_{t=0}^{T-1}$ are the weighting matrices. Note that the sum of $\w_t$ is not necessarily $1$ as $\w_t$ only contains risky assets.
The additional regulation term $\sum_{t=0}^{T-1}\E\big[\|\u_t - X_t \w_t\|_{\Q_t}^2\big]$ penalizes deviations of the decision policy $\u_t$ from the scaled reference portfolio $X_t \w_t$, while the matrices $\Q_t$ determine the strength and direction of this penalization. One practical choice of $\Q_t$ may be diagonal $\Q_t = \diag(q_{1,t},\ldots, q_{p,t})$ with positive diagonal elements.

The problem $\cP_{\RRMV}(\omega,T)$ has the following dual formulation $\tilde{\cP}_{\RRMV}(X_{\tg},T)$:
\begin{equation*}
    \begin{aligned}
    & \min_{\u_t} \Var[X_T] + \sum_{t=0}^{T-1} \E\big[ \big\|\u_t - X_t \w_t \big\|_{\Q_t}^2 \big] \\
    & \text{subject to} \E[X_T] \geq X_{\tg} ~~\text{and}~ \{X_t, \u_t\}~\text{satisfies}~(\ref{def_wealth_dyn_rf}) \,,
    \end{aligned}   
\end{equation*}
where $X_{\tg} \geq X_0 \prod_{t=0}^{T-1} r_t$ is the target expected terminal wealth. Since the optimal solutions to the problems $\cP_{\RRMV}(\omega, T)$ and $\tilde{\cP}_{\RRMV}(X_{\tg},T)$ are internally related, our primary focus will be on $\cP_{\RRMV}(\omega)$.

In $\cP_{\RRMV}(\omega,T)$, the choice of the reference policy is quite flexible. For instance, if we set $\w_k = \BFzero$ and $\Q_k = \rho \cdot \I$ with $\rho>0$ for all $k=0,\ldots, T-1$, the resulting RRMV objective reduces to
\begin{align*}
\min_{\u_t}
\Var[X_T] + \rho \sum_{t=0}^{T-1} \E[\|\u_t\|^2].
\end{align*}
This formulation can be viewed as an extension of the Tikhonov- or Ridge-regularized MV portfolio studied in \cite{demiguel2009optimal,demiguel2009generalized}. The regularization term $\rho \cdot \E[\| \u_t\|^2]$ is a standard device for alleviating the well-known estimation-error problem associated with the covariance matrix, and is closely related to shrinkage-based covariance estimators in \cite{ledoit_improved_2003}.

Beyond the zero reference portfolio, various alternative choices are possible, including the minimum-variance portfolio, the risk-parity portfolio, the naive ``$1/N$'' allocation (in our notation $1/p$ allocation actually), and portfolios obtained from downside-risk formulations \citep{Rockafellar:2002}. One may also use practically relevant benchmark portfolios as reference policies, for example, target holdings specified by a fund manager. Moreover, the regularization term allows the resulting dynamic policy to inherit certain structural patterns. As an illustration, consider $\w_t = \w$ for all $t$, where $\w=(w_1,\ldots,w_p)^\top \geq \BFzero$ is a sparse portfolio with $w_{i}=0$ for any $i \in I \subset \{1,2,\ldots,p\}$, where $I$ represents a subset of assets to be excluded. To enforce this sparsity pattern dynamically, one may set the weighting matrices $\Q_t=\mathrm{diag}(q_{1,t},\ldots,q_{p,t})$ and assign large positive $q_{ti}$ for all $i \in I$.

\subsection{The strategic advantage of feedback policy}\label{sse:comparison_policies}

Before solving the problem $\mathcal{P}_{\RRMV}(\omega,T)$, we first review and compare different frameworks for solving the MMV problem $\cP_{\MV}(\omega,T)$. When $T>1$, MMV problem $\cP_{\MV}(\omega,T)$ is notoriously difficult because the variance term is not separable in the dynamic programming sense. In particular, the variance does not satisfy the smoothing property: $\Var[X_T \mid \cF_t] \neq \Var[\Var[X_T \mid \cF_s] \mid \cF_t]$ for $s > t$, where $\cF_t$ is the information set (filtration) at time $t$.

To address this difficulty, two types of methods have been developed. The first approach, proposed by \cite{LiNg:2000MF}, derives a feedback (FB) policy by introducing an auxiliary parameter that embeds the original problem $\cP_{\RRMV}(\omega, T)$ into a parameterized optimal control problem solvable via dynamic programming (DP). 
By properly selecting this parameter, the optimal policy of $\cP_{\RRMV}(\omega, T)$ can be recovered as\endnote{The detail discussion of time-consistent policy are given in Appendix~\ref{se:tc_policy}.}. 
\begin{align}
\u^{\text{fb}}_t(X_t)
%=\frac{ \bar{X}_{t+1} - r X_t }{1 + \bmu^{\top} \bsigma^{-1} \bmu} \bsigma^{-1} \bmu,\\
= r \cdot \alpha \cdot \big(\bar{X}_{t} -   X_t\big)  \bsigma^{-1} \bmu,
\quad t = 0, \ldots, T-1, \label{def_mv_u}
\end{align}
where $\alpha \triangleq  1/(1+\bmu^{\top}\bsigma^{-1}\bmu)$ and $\bar{X}_{t} \triangleq  r^{t} X_0 + \big( 2\omega r^{T-t}\alpha^T \big)^{-1} $. The policy $\u^{\text{fb}}_t(X_t)$ is called a feedback (FB) policy because it depends adaptively on the realized wealth $X_t$ at time $t$. Since the allocation weight $\bsigma^{-1}\bmu$ is a constant vector, $\bar{X}_{t}$ can be regarded as a risk-adjusted target. The total investment in risky assets is determined by $(\bar{X}_{t} - X_t)$, that is, by the gap between current wealth and this target. Another noteworthy feature of the feedback policy $\u^{\text{fb}}t(X_t)$ is its dependence on the initial wealth $X_0$. In fact, if we re-solve the truncated problem $\cP_{\MV}(\omega,T)\big|_{t=k}^{T}$ at any time $t=k$ given the realized wealth $X_k$, the resulting portfolio policy will not coincide with the original one solved at $t=0$. This inconsistency gives rise to time-inconsistent behavior, and hence the policy $\u^{\text{fb}}_t(X_t)$ is often referred to in the literature as a pre-committed policy (see \cite{basak2010dynamic}).

The second approach to the problem $\cP_{\MV}(\omega,T)$ enforces time consistency by solving the problem backward in time and fixing the optimal policies in later periods (see \cite{basak2010dynamic,bjork2014mean}). Such a policy can also be interpreted as the subgame-perfect Nash equilibrium of a sequential intrapersonal game \citep{basak2010dynamic}. Using our notation, the time-consistent (TC) policy of $\cP_{\MV}(\omega,T)$ is given by
\begin{align}
\u_t^{\text{tc}}= \frac{1}{2 \omega r^{T-t-1} } \bsigma^{-1}\bmu, \quad t = 0, \ldots, T-1. \label{def_mv_u_tc}
\end{align}
Recall that the single-period (static) MV portfolio policy is $\u_t^{\text{static}} = \frac{1}{2\omega} \bsigma^{-1}\bmu$.\endnote{Under our notation, the single-period MV portfolio is obtained from $\cP_{\MV}(\omega,T=1)$.} Clearly, the TC portfolio $\u_t^{\text{tc}}$ shares the same structural form as the static portfolio $\u_t^{\text{static}}$, except that its risk-aversion parameter $\omega$ is effectively discounted over time.

In this work, we aim to derive the pre-committed feedback policy for the problem $\mathcal{P}_{\RRMV}(\omega,T)$, primarily for the following reasons. First, if the investor adheres to the policy computed at time $0$, the feedback policy ~(\ref{def_mv_u}) serves as a globally optimal portfolio policy that achieves the highest Sharpe ratio. Such optimality has been established in several studies (see, e.g., \cite{Vigna2020,van2021distribution}). Second, the feedback policy \eqref{def_mv_u} induces a nonlinear evolution of the wealth process. Specifically, substituting policies \eqref{def_mv_u} and \eqref{def_mv_u_tc} into the wealth dynamics \eqref{def_wealth_dyn_rf} yields
\begin{align*}
\text{FB:}~~& X_{t+1} = r X_t \cdot \big(1-\alpha \P_t^{\top}\bsigma^{-1}\bmu \big) + r \bar{X}_{t} \alpha \P_t^{\top}\bsigma^{-1} \bmu \,,
\\
\text{TC:}~~& X_{t+1} = r X_t + \frac{1}{2\omega r^{T-t-1} } \P_t^{\top} \bsigma^{-1}\bmu \,.
\end{align*}
The feedback policy causes the risky return $\P_t$ to enter the cumulative wealth $X_T$ in a nonlinear manner due to the feedback mechanism, in sharp contrast to the time-consistent (TC) policy, where $\P_t$ contributes linearly to $X_T$. This feedback effect provides a richer structure for shaping the wealth distribution. For example, \cite{van2021distribution} show that the terminal wealth distribution of $X_T$ under the feedback policy is non-Gaussian, even when $\P_t$ follows a Gaussian distribution. Moreover, \cite{van2021surprising} demonstrate that the feedback policy can exhibit more robust performance compared with the time-consistent policy when the portfolio model is misspecified. Motivated by these observations, we focus on developing the feedback policy for the problem $\mathcal{P}_{\RRMV}(\omega,T)$ and analyzing its statistical properties.

\subsection{Optimal policy of RRMV portfolio optimization}

In this section, we present the optimal policy for the problem $\cP_{\RRMV}(\omega, T)$. Since this problem is not separable in the sense of dynamic programming, we adopt the embedding method developed in \cite{LiNg:2000MF} to derive its feedback policy. Specifically, we characterize the solution of the parameterized problem $\cA(\omega, \lambda)$, which is solvable by dynamic programming, and then recover the solution of $\cP_{\RRMV}(\omega, T)$ by identifying a suitable parameter. The key contribution lies in proving the equivalence between the original problem and the parameterized one under the reference regulation term, and in further deriving an explicit analytical feedback policy, extending the classical embedding approach to the regulated MV setting.

We first solve the following auxiliary problem $\cA(\omega,\lambda)$ by introducing parameter $\lambda \in \mR$, 
\begin{equation*}
    \begin{aligned}
    & \min_{\u_t} ~ \E[ \omega X_T^2 - \lambda X_T] + \omega \sum_{t=0}^{T-1} E[\| \u_t - X_t \w_t\|_{\Q_t}^2 ] \\
    & \textrm{subject to} ~ \{X_t, \u_t\} ~\textrm{satisfies \eqref{def_wealth_dyn_rf}} \,.
    \end{aligned}   
\end{equation*}
Clearly, $\cA(\omega, \lambda)$ is a linear-quadratic type of optimization problem, which is solvable by DP. To solve $\cA(\omega, \lambda)$, we introduce the following variables, which are defined recursively for $k=T-1,\dots,1,0$.
\begin{equation} \label{equation:rrmv_rf_para}
    \begin{aligned}
        \D_k &= a_{k+1} \cdot \big( \bsigma + \bmu \bmu^{\top} \big) + \Q_k \,, \\ 
         a_k &= a_{k+1} r^2  + \w_k^\top \Q_k \w_k - \| r a_{k+1} \bmu - \Q_k \w_k \|_{\D_k^{-1}} \,, \\
         b_k &= r b_{k+1}  - b_{k+1} \cdot \big(r a_{k+1} \bmu -  \Q_k \w_k \big)^{\top} \D_k^{-1} \bmu \,, \\
         c_k &= c_{k+1} - b_{k+1}^2 \bmu^\top \D_k^{-1} \bmu \,,
    \end{aligned}
\end{equation}
where $a_T=b_T=1$, $c_T=0$ and $\bmu_k = \E[\P_k]$. Note that these notations work for time-varying $\E[\P_k \P_k^\top]$ and $\E[\P_k]$, so does the solution for $\cP_{\RRMV}(\omega, T)$ below, although for simplicity we only consider constant $\E[\P_k \P_k^\top] = \bsigma + \bmu \bmu^\top$ and $\bmu_k = \bmu$ for theoretical results.
For any choice of positive semi-definite $\Q_k$'s, we have the following lemma making sure that $\D_k$ is always positive definite. 
\begin{lemma}\label{lemma:param_a_nonneg}
The variables $\{a_k\}|_{k=0}^{T}$ in \eqref{equation:rrmv_rf_para} satisfy $a_k > 0$ for all $k= 0,\ldots, T$.
\end{lemma}
Next, we present the optimal portfolio policy for the auxiliary problem $\cA(\omega,\lambda)$.
\begin{theorem}\label{theorem:opt_aux_prob}
The optimal solution of problem $\cA(\omega,\lambda)$ is, for $t=0, \ldots, T-1$, 
\begin{equation} \label{def_A_opt_u}
    \begin{aligned}
    \u_k^{\cA}(X_k,\lambda) = & \Big(\frac{\lambda b_{k+1}}{2\omega} - r a_{k+1} X_k \Big)\cdot\D_k^{-1}\bmu \\
    & + X_k \D_k^{-1} \Q_k \w_k \,,  
    \end{aligned}
\end{equation}
where $\{a_k\}|_{k=0}^{T-1}$ and $\{b_k\}|_{k=0}^{T-1}$ are defined in \eqref{equation:rrmv_rf_para} and $\E[X_T] = b_0 X_0 - \frac{\lambda}{2 \omega} c_0$.
\end{theorem}

The optimal policy of the original RRMV problem $\cP_{\RRMV}(\omega,T)$ and the auxiliary problem $\cA(\omega, \lambda)$ are connected by the parameter $\lambda$, which enables us to retrieve the policy of the problem  $\cP_{\RRMV}(\omega,T)$ and $\tilde{\cP}_{\RRMV}(X_{\tg},T)$ as follows.   
\begin{theorem} \label{theorem:2.2}
When $\lambda^* = (2\omega b_0 X_0 +1)/(1+ c_0)$ where $b_0,c_0$ is from the recursion in \eqref{equation:rrmv_rf_para}, the policy $\u^{\cA}(X_k, \lambda^*)$ defined in  \eqref{def_A_opt_u} solves $\cP_{\RRMV}(\omega, T)$, yielding for $k=0,\ldots, T-1$, 
\begin{equation} \label{RRMVF_omega_uk}
\begin{aligned}
\u^*_k(X_k)
= & \Bigg( \frac{(1 + 2 \omega b_0 X_0) b_{k+1} }{2\omega (1 + c_0)} - r a_{k+1} X_k \Bigg)
\D_k^{-1} \bmu \\ 
& + X_k \D_k^{-1} \Q_k \w_k \,.
\end{aligned}
\end{equation}

Furthermore, when the weighting parameter $\omega$ is chosen as
\begin{align}
\omega^* = \frac{ c_0}{2\big( b_0 X_0 - X_{\tg} \cdot(1+ c_0) \big)} \,,
\label{def_RRMV_rf_omega}
\end{align}
the policy in \eqref{RRMVF_omega_uk} also solves $\tilde{\cP}_{\RRMV}(X_{\tg}, T)$.
\end{theorem}

For the special case when $T=1$ in the problem $\cP_{\RRMV}(\omega,T)$, the optimal policy is as follows.
\begin{proposition}
For the problem $\cP_{\RRMV}(\omega,1)$, the RRMV optimal portfolio policy is
\begin{align}
\u^*_0 = \big(  \bsigma + \Q_0)^{-1} ( \frac{\bmu}{2\omega}  + X_0 \Q_0 \w_0). \label{RRMVF_omega_uk_T1}
\end{align}
\end{proposition}
Note that this solution adds regularization $\Q_0$ to the covariance matrix $\bsigma$ used in the optimization. Motivated by \cite{demiguel2009generalized}, who considered constraining different types of portfolio norms, here we use the weighted $\ell_2$-norm $\u_0^\top \Q_0 \u_0$. In addition, $\Q_0 \w_0$ serves as a regularization to the mean $\bmu$.

In the following proposition, we derive the ideal or maximum Sharpe ratio of the MMV portfolio if there is no estimation error. 
We compare this with holding a static MV portfolio for $T$ periods, 
and the corresponding terminal wealth will be $X_T = X_0 \cdot \prod_{t=0}^{T-1} (r + \frac{1}{2 \omega} \P_t^\top \bsigma^{-1} \bmu)$. 
For any $T>1$, the MMV Sharpe ratio strictly exceeds that of the static MV strategy, showing the advantage of a feedback policy. 
\begin{proposition}\label{prop:sr_compare}
Suppose $\bmu$ and $\bsigma$ are known. 
Let $\SR_{\max}$ and $\SR_{\mathrm{hold}}$ denote the one-period Sharpe ratio achieved by the MMV policy and holding the static MV portfolio over $T$ periods, where 
\begin{equation} \label{SR_max}
    \begin{aligned}
    \SR_{\max}
    &=
    \frac{1}{\sqrt{T}}\sqrt{(1+\bmu^\top\bsigma^{-1}\bmu)^T-1} \,,
    \\
    \SR_{\mathrm{hold}}
    &= 
    \frac{1}{\sqrt{T}}
    \frac{1-(1-M)^T}
    {\sqrt{\left(1+\frac{1}{2\omega r}M(1-M)\right)^T-1}} \,,
    \end{aligned}
\end{equation}
where $M=\frac{\bmu^\top\bsigma^{-1}\bmu}{\,2\omega r +  \bmu^\top\bsigma^{-1}\bmu\,}$. 
Then, for any $T>1$, we have $\SR_{\max}>\SR_{\mathrm{hold}}$.
\end{proposition}

\section{Statistical Justification of Out-of-Sample Performance}\label{se:stat}
In this section, we study the out-of-sample performance of the RRMV portfolio obtained from solving $\cP_{\RRMV}(\omega, T)$ with estimation error.
We begin with the single-period case in Section~\ref{se: out-of-sample-single-period}, which serves as a benchmark and builds intuition for how regularization affects out-of-sample Sharpe ratio.
We then extend the analyses to the multiperiod setting in Section~\ref{sse_multiperiod_RRMV}, where portfolio decisions are made dynamically over time.
This extension shows how the main insights from the single-period case carry over, while also clarifying how the role of regularization can be different under the feedback policies when $T>1$.
To better understand the underlying economic mechanisms, we analyze estimation errors in $\hat\bmu$ and $\hat\bsigma$ separately in both settings.

\subsection{Out-of-Sample Analysis} \label{sec3.1}
In practical applications, the true expected return $\bmu$ and covariance matrix $\bsigma$ are unknown and need to be estimated. Typically, the portfolio evaluation proceeds in two stages. In the portfolio construction stage, unknown parameters are estimated from historical return samples and portfolio weights are computed using the estimated parameters. In the out-of-sample evaluation stage, the portfolio is held for a new unseen sample of returns and assessed based on the realized mean return, volatility and Sharpe ratio etc.

Let $\big\{ \P_{-1}, \P_{-2}, \ldots, \P_{-n} \big\}$ denote the set of historical excess return vectors prior to current time point at zero, where $n$ is the sample size. The sample mean and sample covariance are given by 
$\hat{\bmu} = \frac{1}{n} \sum_{t=1}^n \P_{-t}$ and $\hat{\bsigma} = \frac{1}{n} \sum_{t=1}^n \big(\P_{-t} - \hat\bmu\big)\big(\P_{-t} - \hat\bmu\big)^\top$. 
In what follows, for any parameter $A$, we use $\hat{A}$ to denote its estimated counterpart. Accordingly, when the true expected return $\bmu$ and covariance matrix $\bsigma$ are replaced by their estimators $\hat{\bmu}$ and $\hat{\bsigma}$ in the portfolio optimization, the parameters $\{\hat{a}_k, 
\hat{b}_k, \hat{c}_k, \hat{\D}_k \}|_{k=0}^{T-1}$ serve as the estimated versions of $\{a_k, b_k, c_k, \D_k\}|_{k=0}^{T-1}$, computed by applying the recursions in \eqref{equation:rrmv_rf_para}. Similarly, the RRMV portfolio policy given in \eqref{RRMVF_omega_uk} based on the estimated parameters $\{\hat{a}_k, \hat{b}_k, \hat{c}_k, \hat{\D}_k\}|_{k=0}^{T-1}$ is denoted by $\hat{\u}^*_k$, $t=0,\ldots, T-1$. 
For notational simplicity, we assume $X_0 = 1$. 
In the out-of-sample stage, the estimated optimal portfolio $\{\hat{\u}^*_k\}_{t=0}^{T-1}$ is implemented with future out-of-sample returns $\{\P_0,\P_1, \ldots, \P_{T-1}\}$.

Different from \cite{kanOptimalPortfolioChoice2007}, we focus on the out-of-sample performance of high-dimensional large portfolios where the portfolio size $p$ diverges. Following \cite{Ao2019}, we assume that the ratio $p/n$ converges to some constant $c>0$. Under this regime, we study the asymptotic behavior of the related statistics as $n \rightarrow \infty$ and $p/n \rightarrow c$. To simplify the notation, we use the expression $Y(n) \xrightarrow{a.s} Y^*$ to indicate that $Y(n)$ converges almost surely to $Y^*$ as $n \rightarrow \infty$ and $p/n \rightarrow c$.
In this section, we use the solution of $\tilde{\cP}_{\RRMV}(X_{\tg}, T)$ to illustrate the effect of estimation errors. Similar results can be obtained with $\cP_{\RRMV}(\omega, T)$.

\begin{assumption}\label{asmp_iid}
All excess returns $\P_t$, $t=-q, \ldots, -1, 0, 1, \ldots, T-1$, are independent and identically distributed (i.i.d.) according to the Gaussian distribution $\mathcal{N}(\bmu,\bsigma)$.
\end{assumption}
Then, it  yields the wealth process,
\begin{align}
    \hat{X}_{t+1} = r \hat{X}_{t} + \P_{t}^{\top} \hat{\u}^*_k(\hat{X}_k),~~k =0 ,\ldots, T-1\,, \label{def_cum_esti_wealth}
\end{align}
where all the parameters in the portfolio policy $\{\hat{\u}^*_k\}|_{k=0}^{T-1}$ are fixed numbers conditioning on historical information. Thus, the performance of the cumulative wealth in \eqref{def_cum_esti_wealth} depends only on the newly arrived returns $\{\P_k \}|_{k=0}^{T-1}$. We may compute the expected value and the variance of the cumulated wealth $\hat{X}_t$ with respect to $\{\P_k \}|_{k=0}^{T-1}$ as follows.

\begin{proposition} \label{prop:eff_frontier}
Implementing the optimal RRMV portfolio  $\{\hat{\u}^*_k\}_{t=0}^{T-1}$, the expected value and variance of the cumulated wealth \eqref{def_cum_esti_wealth} are     
\begin{align*}
\E[\hat{X}_{k+1}]   &= \prod_{i=0}^{k} \hat{\eta}_i + \sum_{i=0}^k(\hat{\lambda}_{i}\prod_{j=i+1}^{k}\hat{\eta}_{j}) \,,
\\
\Var[\hat{X}_{k+1}] &= \sum_{i=0}^{k} \hat\nu_i \prod_{j=i+1}^{k}\hat{{g}}_j \,,
\end{align*}
where $\hat{\lambda}_k = \frac{\hat{b}_0 - X_\tg}{\hat{c}_0} \hat{b}_{k+1} \hat{\bmu}^\top \hat{\D}_k^{-1} \hat\bmu$, $\hat\eta_k = r- r \hat{a}_{k+1} \bmu^\top \hat{\D}_k^{-1}  \hat\bmu + \bmu^\top \hat{\D}_k^{-1} \Q_k\w_k$, $\hat\nu_k = \| \E[\hat X_k]\big( r \hat{a}_{k+1} \hat{\bmu} -\Q_k \w_k \big) - \frac{\hat{b}_0 - X_{\tg} }{ 2\hat{c}_0 } \hat{b}_k \hat{\bmu} \|^2_{ \hat\D_k^{-1} \bsigma \hat\D_k^{-1} }$ and $\hat{g}_k = \big\|  ra_{k+1} \hat{\bmu} - \Q_k \w_k  \big\|^2_{ \hat\D_k^{-1} \bsigma \hat\D_k^{-1} } + \hat{\eta}_k^2$, 
for $k=0,\ldots, T-1$. 
\end{proposition}
Note the difference between $\bmu$ and $\hat\bmu$ in $\hat\eta_k$ and the difference between $\hat\D_k$ and $\bsigma$ in $\|\x\|_{\hat\D_k^{-1} \bsigma \hat\D_k^{-1}}^2 = \x^\top \hat\D_k^{-1} \bsigma \hat\D_k^{-1} \x$.
A direct consequence of Proposition~\ref{prop:eff_frontier} is the one-period Sharpe ratio of $\hat{X}_T$:
\begin{equation} \label{def_sharpe_XT}
    \begin{aligned}
    &\SR^2(\hat{X}_T) 
    = \frac{(\E[\hat{X}_{T}] - X_0 r^T)^2}{T \cdot \Var[\hat{X}_{T}]} 
    \\
    &= \frac{(\prod_{i=0}^{T-1} \hat\eta_i + \sum_{i=0}^{T-1}(\hat\lambda_{i}\prod_{j=i+1}^{T-1}\eta_{j})-r^T)^2}{T \cdot \sum_{i=0}^{T-1} \hat\nu_i \prod_{j=i+1}^{T-1}\hat{g}_j} \,. 
    \end{aligned}
\end{equation}

Consider $\Q_k \propto \bar\Q$ for some fixed $\bar\Q$ with proper scaling and zero reference $\w_k = 0$. Then the out-of-sample Sharpe ratio can be significantly simplified.
\begin{corollary}\label{coro:os_SR_w=0}
If we set $\Q_k =  \rho \hat a_{k+1} \bar\Q$ and $\w_k = \BFzero$ for $k=0,\ldots, T-1$, the out-of-sample Sharpe ratio \eqref{def_sharpe_XT} becomes
\begin{align*}
    \SR^2(\hat{X}_T) = \frac{(\hat \eta^T -r^T)^2}{T ((\hat\Gamma + \hat\eta^2)^T - \hat\eta^{2T})} \,,
\end{align*}
where 
$\hat\eta = r\Big(1-\frac{\bmu^\top(\hat\bsigma + \rho \bar\Q)^{-1}\hat\bmu}{1+\hat\bmu^\top (\hat\bsigma + \rho \bar\Q)^{-1} \hat\bmu} \Big)$ is the same $\hat\eta_k$ in Proposition~\ref{prop:eff_frontier} and $\hat{\Gamma} = r^2 \frac{\hat\bmu^\top (\hat{\bsigma} + \rho \bar\Q)^{-1} \bsigma (\hat{\bsigma} + \rho \bar\Q)^{-1} \hat\bmu}{(1+\hat\bmu^\top (\hat{\bsigma} + \rho \bar\Q)^{-1} \hat\bmu)^2}$. When $\rho=0, \hat\bsigma=\bsigma,\hat\bmu=\bmu$, this SR reduces to the maximum possible Sharpe ratio in \eqref{SR_max}.
\end{corollary}

\subsection{Performance of single-period RRMV portfolio}\label{se: out-of-sample-single-period}

Before analyzing the general multiperiod RRMV portfolio, we first consider the single-period case $T=1$. For notational simplicity, we omit the time index and treat $t=0$ implicitly. 
By substituting the estimated parameters $\hat{\bmu}$ and $\hat{\bsigma}$ into (\ref{RRMVF_omega_uk_T1}), the resulting portfolio is $\u^\ast = (\hat\bsigma + \Q)^{-1} (\frac{\hat\bmu}{2\omega}  + \Q \w)$
with the out-of-sample Sharpe Ratio given by $\SR = \frac{\bmu^\top \u^*}{\sqrt{(\u^*)^\top \bsigma \u^*}}$. 

To disentangle the effects of the estimation errors of $\hat\bsigma$ and $\hat\bmu$, we consider two separate scenarios:  
(i) the true covariance matrix $\bsigma$ is known while the expected return $\bmu$ is estimated by $\hat{\bmu}$;  
(ii) the expected return $\bmu$ is known while the covariance matrix is estimated by $\hat{\bsigma}$.

\subsubsection{Scenario (i): known \texorpdfstring{$\bsigma$}~ and estimated \texorpdfstring{$\bmu$}.} \label{sec3.1.1}
We first provide theoretical understanding on the effect of estimating the expected return $\bmu$ by the historical sample mean, leading to the necessity of adding the reference regularization. 
The estimation error of $\hat\bmu$ follows the following distribution, $\varepsilon = \hat\bmu - \bmu \sim N(0, \bsigma/n)$, where $n$ is the sample size of the $p$-dimension return vector.

For any given matrix $\Q \succ 0$, we set the reference portfolio as $\w = \alpha \bar\w$ with some fixed $\bar\w$. 
Then the regularization can be adjusted by choosing $\alpha$, when $\Q$ is fixed. 
The optimal portfolio for this scenario can be written as 
$\u^\ast = (\bsigma + \Q)^{-1} (\frac{\hat\bmu}{2\omega}  + \alpha \Q \bar\w)$
where we can show that $X_\tg - r = \hat\bmu^\top (\bsigma + \Q)^{-1} (\frac{\hat\bmu}{2\omega}  + \alpha \Q \bar\w)$.
The mapping between $\omega$ and target return $X_\tg$ depends on our estimation $\hat\bmu$. Therefore, the optimization problem $\tilde{\cP}_{\RRMV}(X_{\tg},T)$, given any fixed value of $X_\tg$, the corresponding $\omega$ is affected by the estimation error. 
As $n\to \infty$, it is not hard to show 
\begin{align*}
    \frac{1}{2\omega} \xrightarrow{a.s} \frac{1}{2\omega^\ast} =  \frac{X_\tg - r - \alpha \bmu^\top (\bsigma + \Q)^{-1} \Q \bar\w}{\bmu^\top (\bsigma + \Q)^{-1} \bmu + \tr((\bsigma + \Q)^{-1}  \bsigma / n)} \,.
\end{align*}
So, the optimal policy almost surely converges, that is, $\u^\ast \xrightarrow{a.s} (\bsigma + \Q)^{-1} (\frac{\hat\bmu }{2\omega^\ast} + \alpha \Q \bar\w)$, and we can compute the out-of-sample Sharpe ratio of the asymptotic allocation. 
\begin{theorem}
    \label{theorem:single_mu_SR}
    Under the assumptions of known $\bsigma$ and sample mean $\hat\bmu$, when $T=1$, the Sharpe ratio of terminal wealth almost surely converges. 
    And its limiting value depends on $\alpha$ and $\bar\w$ as follows. 
    \begin{align*}
    \SR_\infty(\Q,\alpha\bar\w)
    = 
    \frac{ \bmu^\top (\bsigma + \Q)^{-1} \d_\alpha }{\sqrt{ \| (\bsigma + \Q)^{-1} \d_\alpha \|_{\bsigma} + \frac{e_{\bmu}} {(2\omega^\ast)^2} }} \,,
    \end{align*}
    where $\d_\alpha = \frac{\bmu}{2\omega^\ast}  + \alpha \Q \bar\w$, $e_{\bmu} = \tr( (\bsigma + \Q)^{-1}  \bsigma (\bsigma + \Q)^{-1} \bsigma / n)$ is the bias caused by estimation error and $\SR_\infty$ is the asymptotic Sharpe ratio when $n\to\infty$.
\end{theorem}
Theorem~\ref{theorem:single_mu_SR} states that we can achieve the maximum Sharpe ratio $\sqrt{\bmu^\top \bsigma^{-1} \bmu}$ by letting $\Q \bar\w = (\bsigma + \Q) \bsigma^{-1} \bmu$ and $\alpha = \frac{X_\tg - r}{\bmu^\top \bsigma^{-1} \bmu}$ such that $\frac{1}{2 \omega^\ast}=0$. 
This provides us theoretical guarantee for improved Sharpe by adding the proper regulation term for $\hat\bmu$. For instance, consider the ideal situation of $\Q = \rho \bsigma$ as we know true $\bsigma$.

\begin{corollary}\label{coro:3.2}
    Using the assumptions in Proposition~\ref{theorem:single_mu_SR}, when $\Q = \rho \bsigma$, we can achieve the maximum asymptotic Sharpe ratio by setting $\alpha^\ast = \frac{X_\tg - r}{\bmu^\top \bsigma^{-1} \bmu}$ and $\bar\w^\ast = (1+\rho) \rho^{-1} \bsigma^{-1} \bmu$ for $\rho > 0$. 
    \begin{align*}
    \SR_\infty(\Q, \alpha^\ast \bar\w^\ast)
    = 
    \sqrt{\bmu^\top \bsigma^{-1} \bmu} \,,
    \end{align*} 
    while the MV portfolio without regularization can only achieve $\SR_\infty(\Q, \BFzero)
    = 
    \frac{\bmu^\top \bsigma^{-1} \bmu}{ \sqrt{\bmu^\top \bsigma^{-1} \bmu + p/n} }$. 
\end{corollary}
In high-dimensional settings, the unregulated mean–variance portfolio ($\alpha=0$) exhibits the well-known attenuation of the out-of-sample Sharpe ratio when $p/n \to c > 0$, which falls strictly below the maximum Sharpe ratio. 
In contrast, under the optimal regulation, the implied risk aversion satisfies $1/\omega^\ast = 0$, so the optimal portfolio
effectively uses the reference allocation $\alpha^\ast \bar\w^\ast$.
This result shows that the reference portfolio acts as an implicit shrinkage device on the estimated mean $\hat\bmu$: the regularization adjusts the effective risk aversion in a way that counterbalances the noise in $\hat\bmu$, pulling the portfolio towards the reference direction.

\subsubsection{Scenario (ii): known \texorpdfstring{$\bmu$}~ and estimated \texorpdfstring{$\bsigma$}.} \label{sec3.1.2}
In this section, we investigate the effect of estimation error in the sample covariance $\hat\bsigma$ and assume full knowledge of the expected return, i.e. $\hat\bmu = \bmu$. 
We set the regularization matrix as $\Q = \rho \bar\Q$ and the portfolio is
$\u^\ast = (\hat\bsigma + \rho \bar\Q)^{-1} (\frac{1}{2\omega} \bmu + \rho \bar\Q \w)$.
The question is then by choosing $\rho > 0$ and a proper $\w$, are we able to achieve a better out-of-sample performance for $\u^\ast$? 

Firstly, let us consider the geometric intuition about the out-of-sample Sharpe ratio.
\begin{align*}
    \SR(\rho\bar\Q, \w) 
    = \frac{\langle \bsigma^{-1}\bmu, \u^* \rangle_{\bsigma}}{\|\u^*\|_{\bsigma}} \,.
\end{align*} 
So the numerator of SR is the inner product between our portfolio $\u^*$ and the mean-variance portfolio with true $\bsigma$, where the inner product is defined as $\langle \a, \b\rangle_{\bsigma} = \a^\top \bsigma \b$. Therefore, SR is a scaled version of the $\cos \theta$ angle between $\u^*$ and $\bsigma^{-1}\bmu$, that is $\text{SR} = \cos \theta \cdot \|\bsigma^{-1}\bmu\|_{\bsigma} = \cos \theta \cdot \sqrt{\bmu^\top\bsigma^{-1}\bmu}$, which is optimized at $\sqrt{\bmu^\top\bsigma^{-1}\bmu}$ when those two directions are in perfect align. Is a positive $\rho$ helpful in dragging $\hat\bsigma^{-1} \bmu$ towards $\bsigma^{-1} \bmu$? Intuitively, there are two effects: firstly if $\w$ itself bears certain information on $\bsigma^{-1} \bmu$, we benefit from forcing the portfolio to be similar to $\w$; secondly, based on the random matrix theory (RMT), the spectrum of $\hat\bsigma$ follows the Marchenko–Pastur distribution whose support is wider than that of $\bsigma$ \citep{marvcenko1967distribution,bai2009enhancement}, hence a certain level of shrinkage is useful to correct this bias in recovering eigenvalues of $\bsigma$. More specifically, we present a useful result from RMT in the appendix (Lemma~\ref{lemma:keylemma}), and obtain the following conclusion applying it.

\begin{theorem}
\label{theorem:SR_converge_sig_single}
Under the assumptions of known $\bmu$ and sample covariance $\hat\bsigma$, when $T=1$, we have the almost sure convergence of the Sharpe ratio $\SR(\rho\bar\Q, \w)$ if $\|\bsigma\|$ is bounded. 
Its limit can be expressed as
\begin{align*}
    \SR_\infty(\rho\bar\Q, \w) = \kappa_\rho^{-1/2} \cdot \frac{\bmu^\top \A_\rho^{-1} \d_\rho}{\sqrt{\d_\rho^\top \A_\rho^{-1} \bsigma \A_\rho^{-1} \d_\rho}}
    \,,
\end{align*}
where $\kappa_\rho = 1 - \frac{\tilde s(\rho)}{(1+s(\rho))^2}$, $\A_\rho = \frac{\bsigma}{1+s(\rho)} + \rho \bar\Q$, $\d_\rho = \frac{\bmu}{2\omega^\ast}  + \rho \bar\Q \w$, 
$s(\rho)$ is the solution of $s(\rho) = \frac{c}{p} \tr(\bsigma \A_\rho^{-1})$, 
$\tilde s(\rho)$ is the solution of $\tilde s(\rho) = - \kappa_\rho \frac{c}{p} \tr(\bsigma^2 \A_\rho^{-2})$ and $\frac{1}{2\omega^\ast} = \frac{X_\tg - r - \rho \bmu^\top \A_\rho)^{-1} \bar\Q \w}{\bmu^\top \A_\rho^{-1} \bmu}$. 
\end{theorem} 

We can treat the asymptotic Sharpe ratio as a product of a scalar factor $\kappa_\rho^{-1/2}$ and the pseudo SR, which is the rest of the term. 
Let us now investigate the effects of $\Q$ and $\w$ on the Sharpe ratio. 
First, we check how $\w$ will affect the Sharpe ratio if we set $\Q \to \BFzero$\endnote{Since $\Q = \BFzero$ entirely cancels the regularization term, it makes more sense to consider $\Q \to \BFzero$ while $\Q\w \to \tilde\w$, e.g. $\Q = \I/n$ and $\w = n\tilde\w$.}. 
To maintain the regularization effect of $\w$ when $\Q \to \BFzero$, we assume $\Q\w \to \tilde\w$, which can change independent of $\Q$. 
\begin{equation*}
    \SR_\infty(\BFzero, \tilde\w)
    = 
    \kappa_0^{-1/2} \cdot 
    \frac{ \bmu^\top \bsigma^{-1} (\frac{\bmu}{2\omega^\ast}  + \tilde\w) }{ \sqrt{ (\frac{\bmu}{2\omega^\ast}  + \tilde\w)^\top \bsigma^{-1} (\frac{\bmu}{2\omega^\ast}  + \tilde\w) } } \,.
\end{equation*}
We can show that a necessary condition of maximizing the above Sharpe ratio over $\tilde\w$ is that $\tilde\w^\ast \propto \bmu$. 
By plugging in this optimal $\tilde\w^\ast$, the Sharpe ratio becomes  
\begin{equation*}
    \SR_\infty(\BFzero, \tilde\w^\ast)
    = \kappa_0^{-1/2} \cdot \sqrt{\bmu^\top \bsigma^{-1} \bmu} \,.
\end{equation*}
This shows that the pseudo Sharpe ratio has been adjusted to align with the maximum Sharpe ratio. 
However, the presence of the additional scalar $\kappa_0$ still diminishes the final asymptotic Sharpe ratio (see below for the proof that $\tilde s(0) < 0$ so that $\kappa_0 > 1$).  
This implies that if we do not apply a proper regularization for $\hat\bsigma$, i.e. naively using $\Q \to \BFzero$, even with the freedom of tuning $\tilde\w$, we cannot reach the maximum Sharpe ratio.

Next, we show that using a nonzero $\Q$ is indeed necessary to regularize the sample covariance $\hat\bsigma$ in high dimensional settings. 
We would like to argue that the asymptotic SR is typically maximized at $\rho >0$. 
Consider the simpler case of $\bsigma = \sigma^2 \I$, $\Q = \rho\I$ and $\w = \BFzero$ as an example. We have
\begin{equation*}
\SR_\infty(\rho\I, \BFzero) = \kappa_\rho^{-1/2} \cdot \frac{\|\bmu\|_2}{\sigma} \,.
\end{equation*}
The scalar $\kappa_\rho^{-1/2}$ plays an essential role, where $s$ solves $s = c (1/(1+s) + \rho / \sigma^2)^{-1}$ and $\tilde s$ solves $\tilde s = c(\tilde s / (1+s)^2 - 1) (1/(1+s) + \rho / \sigma^2)^{-2}$. When $\rho = 0$, it is easy to calculate $s = c/(1-c)$, $\tilde s = -c / (1-c)^3$ and the SR scalar $\kappa_\rho^{-1/2} = \sqrt{1-c}$, which reduces SR when $c = p/n\in (0,1)$\endnote{Note that when $c \ge 1$, $\hat\bsigma$ is not invertible and we cannot choose $\rho=0$.}.
In this case, SR is maximized when $\rho = \infty$, which gives $s=0, \tilde s = 0$ and the scalar $\kappa_\rho = 1$. This is very intuitive since when $\bsigma \propto \I$, the penalty $\rho \I$ actually applies the correct covariance, so the best way to construct portfolio is to drag $\hat\bsigma$ to $\I$ as much as possible. When $\bsigma$ is not proportional to identity, $\text{SR}_{\infty} / \kappa_\rho$ is actually maximized exactly at $\rho=0$. We call this term pseudo SR, computed as if we use the scaled true covariance $\bsigma / (1+s)$ to build the portfolio.
Therefore, there exists a tradeoff between the scalar $\kappa_\rho$ and the pseudo SR: $\kappa_\rho$ is maximized at $\rho =\infty$ and the pseudo SR is maximized at $\rho = 0$. In general, this tradeoff leads to the selection of an optimal positive $\rho$. 

To illustrate this numerically, we conduct a toy simulation: set $p=200$ and $c \in \{ 0.5, 0.7 \}$; $\bmu$ is generated with $80\%$ zeros and $20\%$ from uniform $-0.1$ to $0.1$; $\bsigma$ is a diagonal matrix with diagonal elements generated from Gamma distribution with shape $1$ and scale $1$; $\bar\Q$ is set as an identity matrix and recall $\Q = \rho\bar\Q$. 
For a series of $\rho$ values, we solve $s(\rho), \tilde s(\rho)$ and calculate $\kappa_\rho$, pseudo SR and $\text{SR}_\infty$ which is the product of the two. The curves versus $\rho$ are shown in Figure~\ref{fig_SR}. 
It is very clear $\kappa_\rho$ increases and pseudo SR decreases as $\rho$ grows. 
To balance these two curves, the optimal $\rho^\ast$ to maximize $\text{SR}_\infty$ is achieved at a positive value. 

\begin{figure}[htbp]
    \centering
    \begin{subfigure}[t]{0.45\textwidth}
        \centering
        \includegraphics[width=\textwidth]{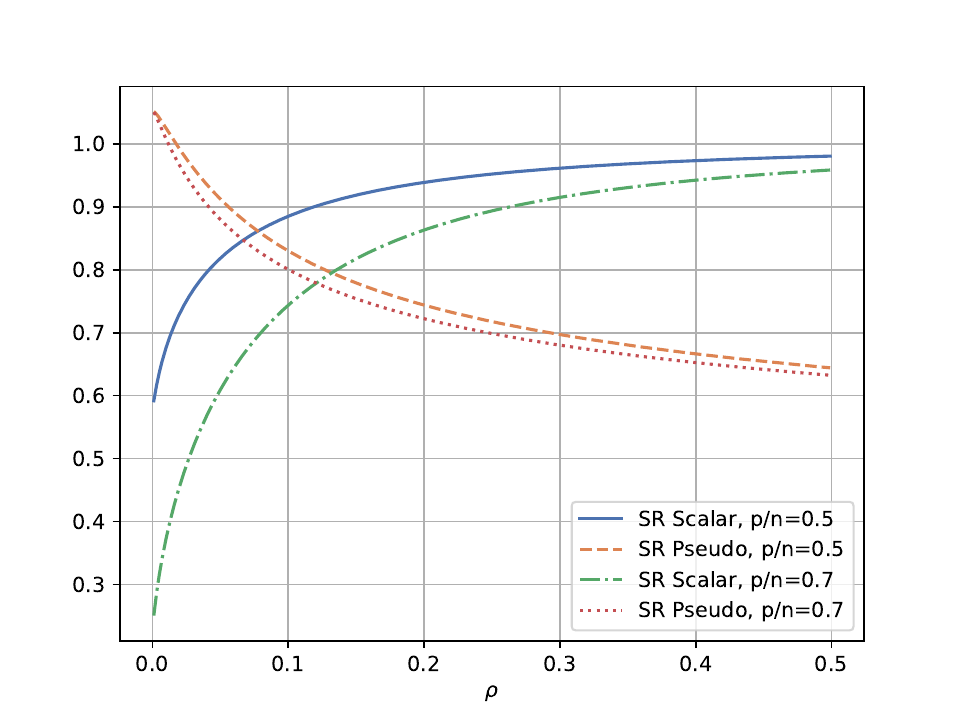}
        \caption{SR scalar and pseudo SR}
        \label{fig_SR_a}
    \end{subfigure}
    \hfill
    \begin{subfigure}[t]{0.45\textwidth}
        \centering
        \includegraphics[width=\textwidth]{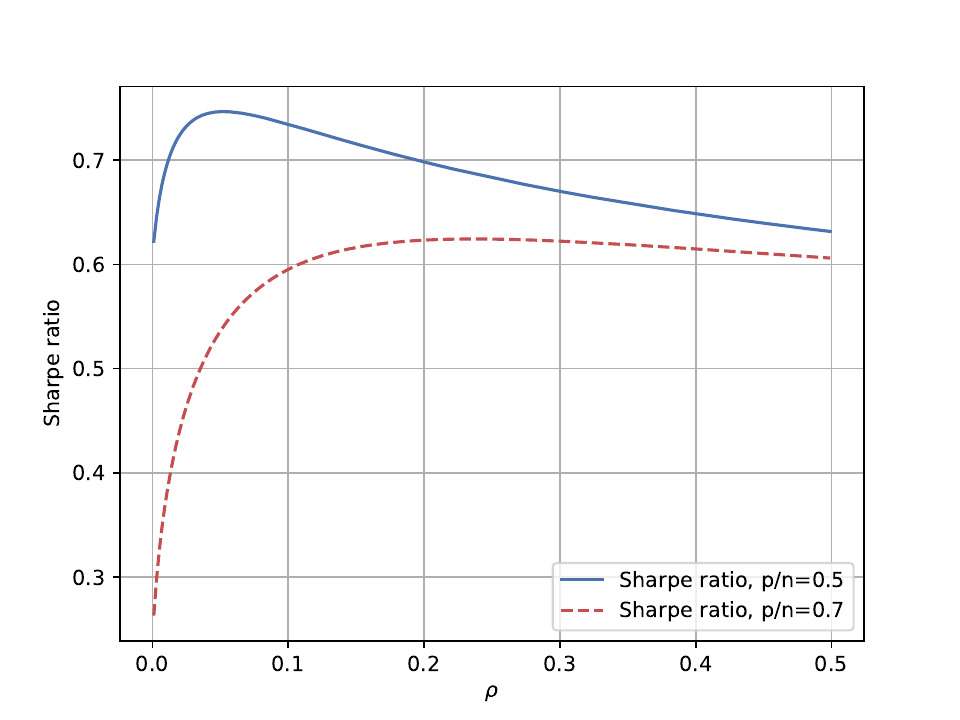}
        \caption{Sharpe ratio $\SR_{\infty}$}
        \label{fig_SR_b}
    \end{subfigure}
    \caption{$\SR_{\infty}$ versus regularization $\rho$ with estimation error in $\hat\bsigma$ when $T=1$. The SR scalar $\kappa_\rho$ increases and the pseudo SR decreases as $\rho$ grows. Therefore, $\SR_{\infty} = \kappa_\rho \cdot \text{pseudo SR}$ is maximized at $\rho^\ast > 0$.}
    \label{fig_SR}
\end{figure}

\subsection{Performance of Multiperiod RRMV portfolio}\label{sse_multiperiod_RRMV}

In this section, we analyze the general multiperiod RRMV portfolio. 
The out-of-sample Sharpe ratio has been derived in Proposition~\ref{prop:eff_frontier}. 
Following the single-period RRMV discussions, we account for the two primary sources of estimation risk, arising from errors in estimating the mean and the covariance matrix, in the two subsections below. 
Our analysis is organized around two questions. 
The first concerns how the regularization term affects the out-of-sample Sharpe ratio, similar to the single-period setting.
The second investigates how this effect propagates in a multiperiod environment. We will reveal new phenomena, which cannot be observed under the single-period RRMV framework and is unraveled for the first time in the multiperiod portfolio optimization literature.

\subsubsection{Scenario (i): known \texorpdfstring{$\bsigma$}~ and estimated \texorpdfstring{$\bmu$}.} \label{sec3.2.1}

We assume the covariance matrix $\hat\bsigma = \bsigma$ is known and focus on estimation error in the sample mean $\hat\bmu$. 
For the general multiperiod RRMV portfolio with $\Q_k$ and $\w_k$ in each period, we first show that the system of parameters defined in \eqref{equation:rrmv_rf_para} will almost surely converge as $n \to \infty$.  
\begin{lemma}
\label{lemma:param_mu}
    Under the assumptions of known $\bsigma$ and sample mean $\hat\bmu_k=\hat\bmu$, the parameters in RRMV policy  almost surely converge as follows. 
    Let $\bdelta_k^{\ast} = \bsigma + \Q_k /a_{k+1}^{\ast} $, 
    \begin{align*}
        \hat{a}_k \xrightarrow{a.s} 
        & a_k^\ast
        = \frac{  (a_{k+1}^{\ast})^{-1}   (a_{k+1}^\ast r + \w_k^\top \Q_k ( \bdelta_k^{\ast} )^{-1} \bmu )^2 }{1+\bmu^\top ( \bdelta_k^{\ast} )^{-1} \bmu + \tr(( \bdelta_k^{\ast} )^{-1} \bsigma/n)} \\
        & + ( a_{k+1}^{\ast} )^{-1} \w_k^\top \Q_k \Big(a_{k+1}^{\ast} \Q_k^{-1} - ( \bdelta_k^{\ast} )^{-1} \Big) \Q_k \w_k \,,
        \\
        \hat{b}_k \xrightarrow{a.s} 
        & b_k^\ast 
        = {{b}}_{k+1}^\ast \frac{r + (a_{k+1}^{\ast})^{-1} \bmu^\top ( \bdelta_k^{\ast} )^{-1} \Q_k\w_k}{1 + \bmu^\top ( \bdelta_k^{\ast} )^{-1} \bmu + \tr( (\bdelta_k^{\ast}  )^{-1}  \bsigma/n)} \,,
        \\
        \hat{c}_k \xrightarrow{a.s} 
        & c_k^\ast 
        = c_{k+1}^\ast + (b_{k+1}^{\ast})^2 (a_{k+1}^{\ast})^{-1} 
        \\
        & \cdot \frac{\bmu^\top ( \bdelta_k^{\ast} )^{-1} \bmu + \tr( (\bdelta_k^{\ast}  )^{-1}  \bsigma/n)}{1 + \bmu^\top ( \bdelta_k^{\ast} )^{-1} \bmu + \tr( (\bdelta_k^{\ast}  )^{-1}  \bsigma/n)} \,.
    \end{align*}
\end{lemma}
As a direct consequence, the portfolio policy and the out-of-sample Sharpe ratio will also converge almost surely. 
The corresponding asymptotic Sharpe ratio if we set $T=1$ is identical to that of the single-period RRMV portfolio given in Section~\ref{sec3.1.1}. 

\begin{theorem}
\label{theorem:multi_SR_mu}
    Under the assumptions of known $\bsigma$ and sample mean $\hat\bmu_k=\hat\bmu$, for general $T$, the asymptotic Sharpe ratio of the $\RRMV$ portfolio, given $\Q_k$ and $\w_k$, converges almost surely. Its limit is 
    \begin{align*}
        \SR_\infty(T,\Q_k,\w_k) 
        = \frac{\E[\hat X_T]_\infty - r^T}{\sqrt{T \cdot \Var(\hat X_T)_\infty}} \,, 
    \end{align*}
    where $\E[\hat X_T]_\infty = \prod_{i=0}^{T-1} \eta^\ast_i + \sum_{i=0}^{T-1} (\lambda^\ast_{i} \prod_{j=i+1}^{T-1} \eta^\ast_{j})$, $\Var(\hat X_T)_\infty = \sum_{i=0}^{T-1} \nu_i^\ast \prod_{j=i+1}^{T-1}{g}_j^\ast$. $\lambda_i^\ast$, $\eta_i^\ast$, $\nu_i^\ast$ and $g_i^\ast$ are the limiting values of their corresponding estimators, which are constant depending on $a_k^\ast$, $b_k^\ast$, $c_k^\ast$, $\Q_k$ and $\w_k$.
\end{theorem}

Theorem~\ref{theorem:multi_SR_mu} is too general and abstract to provide useful understanding. To examine separately the contributions from the regularization $\Q$ and $\w$, we first fix $\w_k = \BFzero$ and then fix $\Q_k \to \BFzero$. 
The following corollary shows the corresponding asymptotic Sharpe ratio when $\Q_k$ aligns with the true $\bsigma$ and $\w_k = \BFzero$.  

\begin{corollary}
    \label{coro:multi_SR_mu_w_0}
    Under the same assumptions in Proposition~\ref{theorem:multi_SR_mu}, if $\w_k = \BFzero$, $\Q_k = a_{k+1} \rho \bsigma$, the out-of-sample Sharpe ratio of the RRMV portfolio almost surely converges to
    \begin{align*}
        \SR_\infty(T, \rho, \BFzero) = \frac{T^{-1/2} ((\bmu^\top \bsigma^{-1} \bmu + e_{\rho,\bmu})^T - e_{\rho,\bmu}^T)}{\sqrt{ ( \bmu^\top \bsigma^{-1} \bmu + p/n + e_{\rho,\bmu}^2)^T - e_{\rho,\bmu}^{2T} }} \,,
    \end{align*} 
    where $e_{\rho,\bmu} = e_\bmu + \rho + 1$ and following Theorem~\ref{theorem:single_mu_SR}, $e_\bmu = p/n$ in this case. 
    When $p/n=0$, we have $\frac{\partial \SR_\infty(T, \rho, \BFzero)}{\partial \rho}\big|_{\rho=0} = 0$, that is, the optimal choice of $\rho^* = 0$. 
    When $p/n > 0$ and $T > 1$, the derivative $\frac{\partial \SR_\infty(T, \rho, \BFzero)}{\partial \rho}\big|_{\rho=0}$ is positive when $\bmu^\top \bsigma^{-1} \bmu < 1$, and non-positive otherwise. That is, when the maximum SR is in a modest range and when the dimensionality is high, we can improve SR by choosing $\rho^* > 0$.
\end{corollary}

Note that the effect of $\Q$ is very different from the single-period case. The conventional understanding for the effect for $\Q$ and $\w$ is that $\Q$ is used to regularize $\hat\bsigma$ while $\w$ is used to regularize $\hat\bmu$. However when $T > 1$ and $p/n > 0$, even with $\w_k = \BFzero$ and known true covariance $\bsigma$, $\Q$ can be used to regularize $\hat\bmu$. 

To be specific, When $T=1$, the asymptotic SR formula in Corollary~\ref{coro:multi_SR_mu_w_0} reduces to the single-period case, where the asymptotic Sharpe ratio equals to $\frac{\bmu^\top \bsigma^{-1} \bmu}{\sqrt{\bmu^\top \bsigma^{-1} \bmu + p/n}}$, independent of $\rho$ or $\Q$. 
In contrast to the single-period result, the estimation error in $\hat\bmu$ for multiperiod, high-dimensional portfolios accumulates as the system evolves over time. 
In the multiperiod setting, shrinkage of the covariance matrix additionally serves to adjust the estimation error in $\hat\bmu$, but we can only improve the out-of-sample Sharpe ratio when $\bmu^\top \bsigma^{-1} \bmu < 1$, that is, when the maximum SR in (\ref{SR_max}) is in a modest range. 

\begin{figure}[htbp]
    \centering
    \begin{subfigure}[t]{0.45\textwidth}
        \centering
        \includegraphics[width=\textwidth]{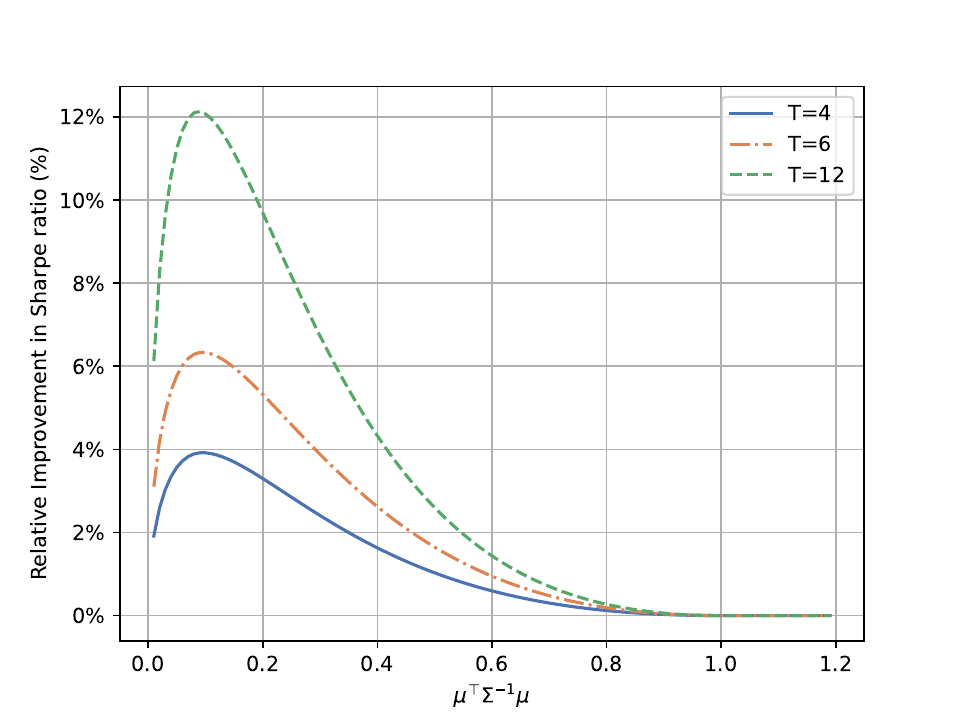}
        \caption{Relative improvement versus $\bmu^\top \bsigma^{-1} \bmu$}
        \label{fig:asym_SR_mu_w=0_SNR}
    \end{subfigure}
    \hfill
    \begin{subfigure}[t]{0.45\textwidth}
        \centering
        \includegraphics[width=\textwidth]{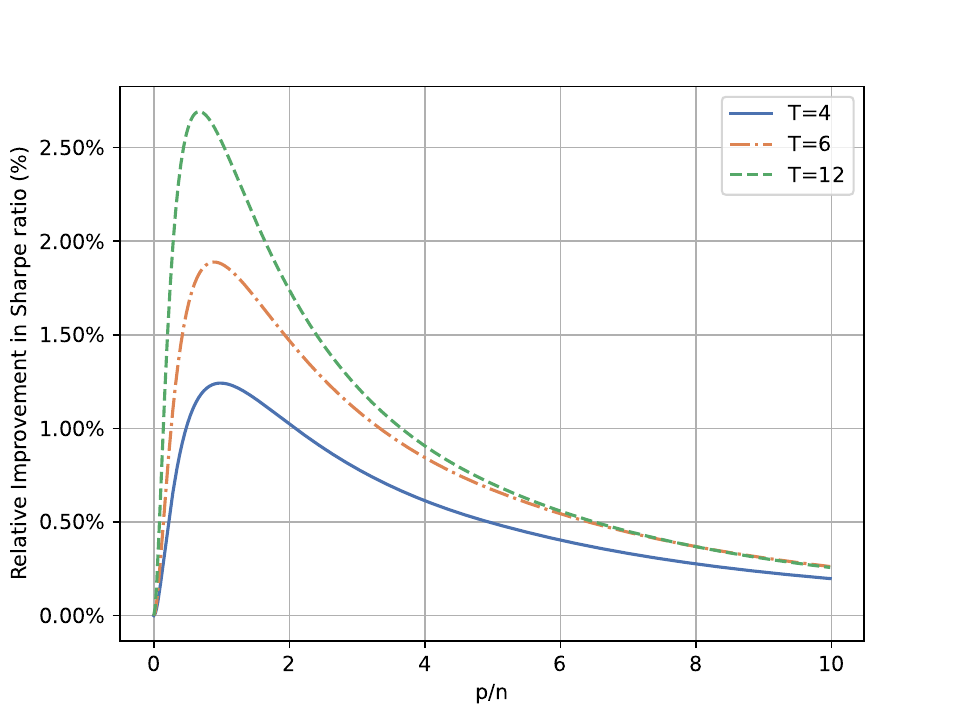}
        \caption{Relative improvement versus $p/n$}
        \label{fig:asym_SR_mu_w=0_HD}
    \end{subfigure}

    \caption{Relative improvement of $\SR_\infty(T,\rho^\ast,\BFzero)$ over
    $\SR_\infty(T,0,\BFzero)$ with estimation error in $\hat{\bmu}$ when $T > 1$.}
    \label{fig_multi_SR_mu_w_0}

    \caption*{\footnotesize
    Notes. The relative improvement is defined as $\frac{\SR_\infty(T,\rho^\ast,\BFzero)-\SR_\infty(T,0,\BFzero)} {\SR_{\max}}\,,$
    where $\rho^\ast$ is computed numerically. In Figure~\ref{fig:asym_SR_mu_w=0_SNR},
    $p/n=0.5$. In Figure~\ref{fig:asym_SR_mu_w=0_HD},
    \(\bmu^\top \bsigma^{-1}\bmu=0.5\). }
\end{figure} 

We further illustrate the effect of $\rho$ or $\Q$ numerically using the same simulation as in Section~\ref{sec3.1.2}. 
Figure~\ref{fig_multi_SR_mu_w_0} shows the relative improvement in the asymptotic out-of-sample Sharpe ratio $\SR_\infty(T, \rho^\ast, 0)$ compared to the baseline case of $\rho=0$. 
Figure~\ref{fig:asym_SR_mu_w=0_SNR} shows that the relative improvement is substantial when $\bmu^\top \bsigma^{-1} \bmu < 1$, emphasizing the beneficial role of regularization in mitigating estimation error in low signal environments. 
However, as $\bmu^\top \bsigma^{-1} \bmu \geq 1$, the improvement diminishes entirely, consistent with the theoretical result that regularization offers no advantage when the maximum SR is sufficiently high. 
Figure~\ref{fig:asym_SR_mu_w=0_HD} examines the behavior of $\SR_\infty(T, \rho^\ast, 0)$ as a function of the high-dimensional ratio $p/n$ when $\bmu^\top \bsigma^{-1} \bmu = 0.5$. 
When $p/n$ is small, the improvement is limited because estimation error plays a negligible role in relatively low-dimensional settings. 
As $p/n$ increases, the relative improvement in $\SR_\infty$ attains its maximum at moderate dimensionality, particularly when $p/n$ is close to one. 
With further increases in dimensionality, the impact of estimation error in $\hat\bmu$ is significantly amplified, which reduces the magnitude of the possible improvement.  
However, the improvement remains positive as our theory predicts. 
Across both panels, the relative improvement is larger for longer investment horizons, reflecting the compounding effect of regularization over time. 
Overall, these patterns are all consistent with Corollary~\ref{coro:multi_SR_mu_w_0} and highlight the practical value of regularization for improving Sharpe ratios in high-dimensional, multiperiod settings.

Next, let us evaluate the effect of $\w$ when we set $\Q_k \to \BFzero$. In this case, we assume $\Q_k\w_k \to \tilde\w$ to separate the effect of $\w$.  
By examining the first-order condition with respect to $\tilde\w$, we find that the optimal direction satisfies $\tilde\w^\ast \propto \bmu$. 
Accordingly, we set $\tilde\w = \alpha \bmu$ and study how the choice of $\alpha$ influences the Sharpe ratio. 
Here we directly present numerical results to illustrate this effect, employing the same simulation setup as in Section~\ref{sec3.1.2}. 
From Figure~\ref{fig_multi_SR_mu_Q_0}, we observe that in both settings of $T=1$ and $T=4$, an appropriate choice of $\alpha^\ast>0$ leads to an improvement in the Sharpe ratio. 
In the single-period case, selecting the optimal $\alpha$ achieves the maximal Sharpe ratio as we have already proved so in Corollary~\ref{coro:3.2}\endnote{In Corollary~\ref{coro:3.2}, we use $\Q = \rho\bsigma$ so that we do not need to refined $\Q\w$ as $\w$, but it has a similar effect as choosing $\Q \to \BFzero$ here.}. 
In contrast, in the multiperiod setting, a gap persists between the maximal Sharpe ratio and the value attained under the optimal choice $\alpha^\ast$.

\begin{figure}[htbp]
    \centering
    \begin{subfigure}[t]{0.45\textwidth}
        \centering
        \includegraphics[width=\textwidth]{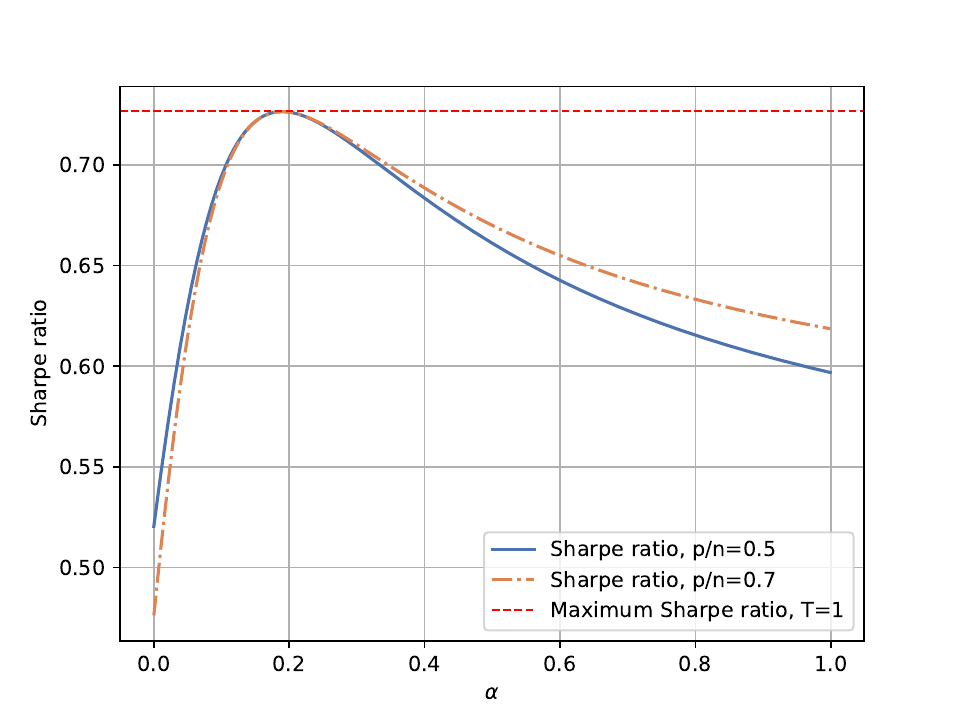}
        \caption{$T=1$}
        \label{fig:asym_SR_mu_T_1_Q_0_w}
    \end{subfigure}
    \hfill
    \begin{subfigure}[t]{0.45\textwidth}
        \centering
        \includegraphics[width=\textwidth]{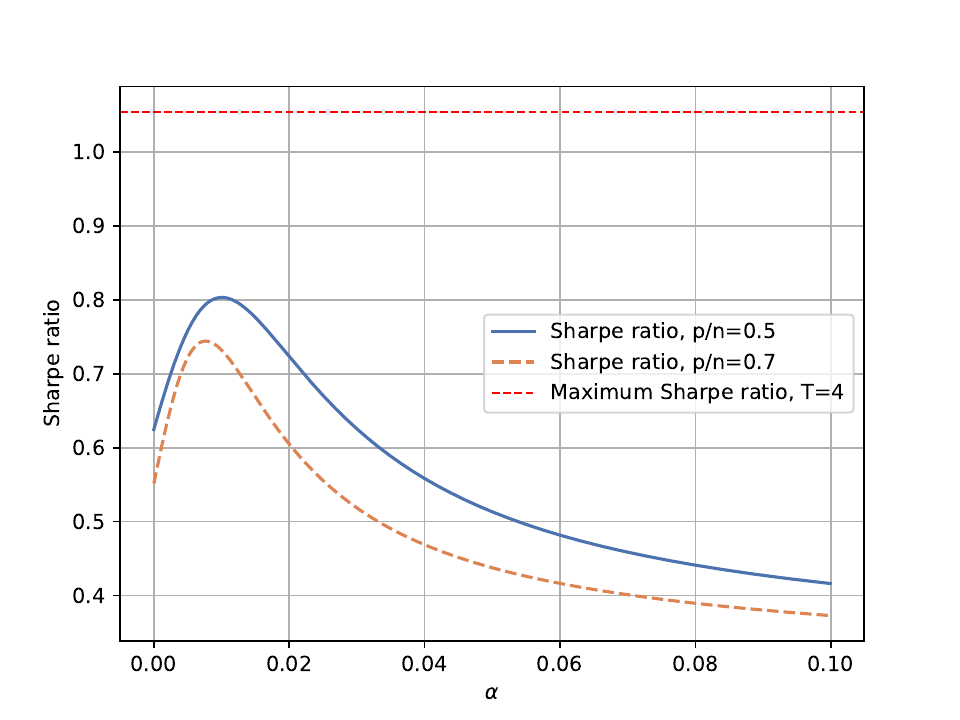}
        \caption{$T=4$}
        \label{fig:asym_SR_mu_T_4_Q_0_w}
    \end{subfigure}

    \caption{$\SR_{\infty}$ versus reference scale $\alpha$ with estimation error in $\hat{\bmu}$.}
    \label{fig_multi_SR_mu_Q_0}

    \caption*{\footnotesize
    Notes. The asymptotic Sharpe ratio $\SR_{\infty}$ is maximized at
    $\alpha^\ast > 0$ in both the single-period and multiperiod cases.}
\end{figure}

\subsubsection{Scenario (ii): known \texorpdfstring{$\bmu$}~ and estimated \texorpdfstring{$\bsigma$}.} \label{sec3.2.2}
Similarly, we can show the convergence of the parameter system when $\bmu$ is known and $\bsigma$ is estimated by sample covariance. 
To this end, we introduce the following lemma and notations for the limiting value of $\x^\top \hat\D_k^{-1} \y$ and $\x^\top \hat\D_k^{-1} \bsigma \hat\D_k^{-1} \y$ for any fixed $\x,\y$.  

\begin{lemma}
    \label{lemma:sig_converge}
    Under the assumptions of known $\bmu$ and sample covariance $\hat\bsigma_k=\hat\bsigma$, for any deterministic $\x, \y \in \mathbb{R}^p$ with bounded norm, we have $\x^\top \hat\D_k^{-1} \y \xrightarrow{a.s} ~ C_{\x,\y,k}^\ast$ and $\x^\top \hat\D_k^{-1} \bsigma \hat\D_k^{-1} \y \xrightarrow{a.s} B_{\x,\y,k}^\ast$
    , where
    \begin{align*}
        & C_{\x,\y,k}^\ast
        \triangleq  (a_{k+1}^{\ast} )^{-1}  \x^\top \Big( \A_{\Q_k,\w_k} + \bmu\bmu^\top \Big)^{-1} \y \,,  
        \\
        & B_{\x,\y,k}^\ast
        \triangleq (a_{k+1}^{\ast})^{-2}  \Big(1-\frac{\tilde s_k}{(1+s_k)^2}\Big) 
        \\
        & \cdot \x^\top \Big( \A_{\Q_k,\w_k} + \bmu\bmu^\top \Big)^{-1} \bsigma \Big( \A_{\Q_k,\w_k} + \bmu\bmu^\top \Big)^{-1} \y \,,
    \end{align*}
    where $\A_{\Q_k,\w_k} = \frac{\bsigma}{1+s_k} + \frac{\Q_k}{a_{k+1}^\ast}$, $s_k$ is the solution of $s_k = \frac{c}{p} \tr( \bsigma (\A_{\Q_k,\w_k} + \bmu \bmu^\top)^{-1})$ 
    and 
    $\tilde s_k$ is the solution to $\tilde s_k= (\frac{\tilde s_k}{(1+s_k)^2} - 1) \frac{c}{p} \tr( \bsigma^2 (\A_{\Q_k,\w_k} + \bmu \bmu^\top)^{-2})$. 
\end{lemma}
Using the notations defined in Lemma~\ref{lemma:sig_converge}, we have the following theorem. 
\begin{theorem}
    \label{theorem:pram_converge_sig}
    Under the assumptions of known $\bmu$ and sample covariance $\hat\bsigma_k=\hat\bsigma$, for general $T$, the parameters in RRMV policy almost surely convergence as follows. 
    \begin{align*}
    \hat{a}_k \xrightarrow{a.s} a_k^\ast 
    & = a_{k+1}^{\ast} r^2 - (a_{k+1}^{\ast})^2 r^2 C_{\bmu,\bmu,k}^\ast 
    \\
    & \quad + 2 r a_{k+1}^\ast C_{\bmu,\Q_k\w_k,k}^\ast + \w_k^\top \Q_k \w_k 
    \\
    & \quad -  C_{\Q_k\w_k,\Q_k\w_k,k}^\ast \,,
    \\
    \hat{b}_k \xrightarrow{a.s} b_k^\ast 
    & = b_{k+1}^\ast r - a_{k+1}^\ast b_{k+1}^\ast r C_{\bmu,\bmu,k}^\ast 
    \\
    & \quad + b_{k+1}^\ast C_{\bmu,\Q_k\w_k,k}^\ast \,,
    \\
    \hat{c}_k \xrightarrow{a.s} {c}_k^\ast 
    & = {c}_{k+1}^\ast - (b_{k+1}^{\ast})^2 C_{\bmu,\bmu,k}^\ast \,.
    \end{align*}
    Let $\phi_k^\ast = r a_{k+1}^\ast E[\hat X_k]_\infty + \frac{b_0^\ast - X_\tg}{{c}_0^\ast} b_{k+1}^\ast$ where $\E[\hat{X}_k]_\infty
        = \prod_{i=0}^{k-1} \eta_i^\ast + \sum_{i=0}^{k-1}( \lambda_i^\ast \prod_{j=i+1}^{k-1} \eta_j^\ast)$, the asymptotic Sharpe ratio of the $\RRMV$ portfolio, given $\Q_k$ and $\w_k$, will almost surely converge to
    \begin{align*}
        \SR_\infty(T, \Q_k, \w_k)
        = 
        \frac{X_\tg-r^T}{\sqrt{T \cdot \sum_{i=0}^{T-1} \nu_i^\ast \prod_{j=i+1}^{T-1}{g}_j^\ast} } \,,
    \end{align*}
    where following Theorem~\ref{theorem:multi_SR_mu}, the  detailed formulas for $\nu_k^\ast$ and $g_k^\ast$ can be derived: $\nu_k^\ast = ( \phi_k^\ast B_{\bmu,\bmu,k}^\ast )^2 + \E[\hat X_{k}]_\infty^2 B_{\Q_k\w_k,\Q_k\w_k,k}^\ast - 2 \phi_k^\ast E[\hat X_k]_\infty B_{\bmu,\Q_k\w_k,k}^\ast$ and $g_k^\ast = r^2 (a_{k+1}^{\ast})^2 B_{\bmu,\bmu,k}^\ast + B_{\Q_k\w_k,\Q_k\w_k,k}^\ast - 2 r a_{k+1}^\ast B_{\bmu,\Q_k\w_k,k}^\ast + (\eta_k^{\ast})^2$.
\end{theorem}

After obtaining the expression for the asymptotic Sharpe ratio, we address how regularization $\Q$ and $\w$ function in this scenario. Similar to Section~\ref{sec3.2.1}, we consider $\Q_k \to \BFzero$ and $\w_k = \BFzero$ separately. 
To isolate the effect of $\w$, we set $\Q_k \to 0$ and assume $\Q_k \w_k \to \tilde\w_k$. 
As in the single-period case, a necessary condition for the optimal reference weight is $\tilde\w_k \propto \bmu$ since true $\bmu$ is known. 
We therefore parameterize the reference weight as $\tilde\w_k = \tilde\w = \alpha \bmu$ and use this ideal case to examine how the choice of $\alpha$ influences the limiting Sharpe ratio. 
Numerical results are illustrated using the same simulation design as before. 
From Figure~\ref{fig_multi_SR_sig_Q_0}, we find that the role of the reference weight $\w$ differs substantially in the multiperiod RRMV portfolio. 
As shown in Theorem~\ref{theorem:SR_converge_sig_single}, in the single-period case, the scale of the reference weight when $\tilde\w \propto \bmu$ does not affect the Sharpe ratio at all. 
Consistent with this observation, Figure~\ref{fig_multi_SR_sig_Q_0_T_1} shows that the limiting Sharpe ratio deteriorates as dimensionality increases and cannot be improved by choosing a positive $\alpha$. 
In contrast, the multiperiod setting exhibits a fundamentally different behavior. 
Figure~\ref{fig_multi_SR_sig_Q_0_T_4} demonstrates that the asymptotic Sharpe ratio can be maximized at some $\alpha^\ast > 0$. 
This finding indicates that even without applying shrinkage $\Q$ to the covariance estimator, an appropriately chosen reference weight can still mitigate the estimation error in $\hat\bsigma$ and improve the out-of-sample Sharpe ratio. This again expands our conventional understanding that $\Q$ is used to regularize $\hat\bsigma$ while $\w$ is used to regularize $\hat\bmu$. Actually when $T > 1$ and $p/n > 0$, even with $\Q_k \to \BFzero$ and known true expected return $\bmu$, $\w$ can be used to regularize $\hat\bsigma$. 

\begin{figure}[htbp]
    \centering
    \begin{subfigure}[t]{0.45\textwidth}
        \centering
        \includegraphics[width=\textwidth]{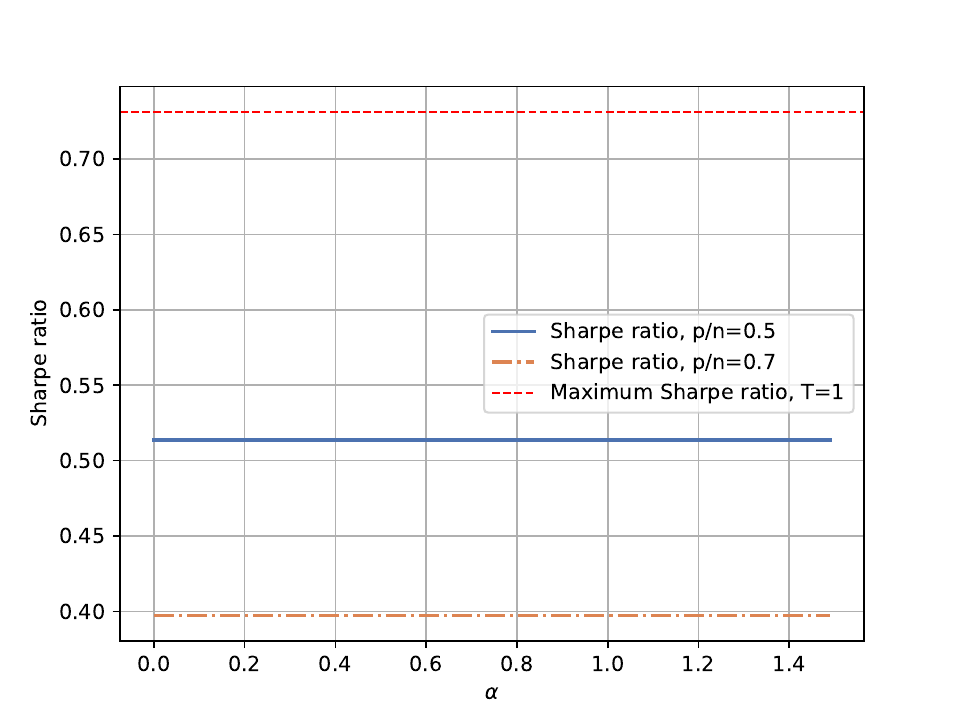}
        \caption{$T=1$}
        \label{fig_multi_SR_sig_Q_0_T_1}
    \end{subfigure}
    \hfill
    \begin{subfigure}[t]{0.45\textwidth}
        \centering
        \includegraphics[width=\textwidth]{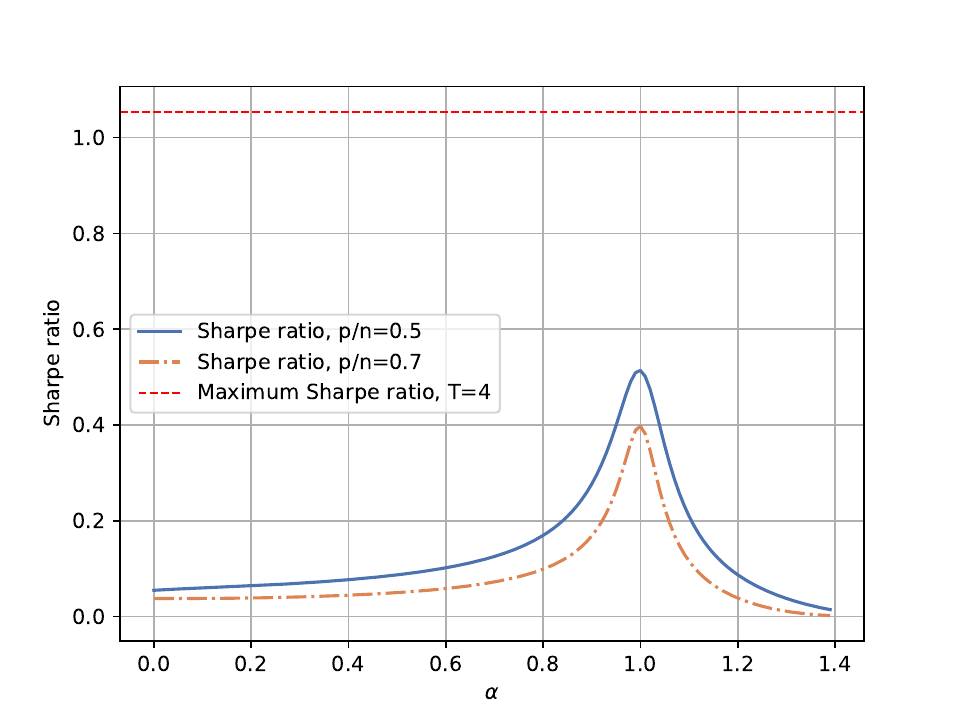}
        \caption{$T=4$}
        \label{fig_multi_SR_sig_Q_0_T_4}
    \end{subfigure}

    \caption{$\SR_{\infty}$ versus reference scale $\alpha$ with estimation error in $\hat{\bsigma}$.}
    \label{fig_multi_SR_sig_Q_0}

    \caption*{\footnotesize
    Notes. The asymptotic Sharpe ratio $\SR_{\infty}$ does not change with
    $\alpha$ in the single-period case and is maximized at $\alpha^\ast > 0$
    in the multiperiod case.}
\end{figure}

Next, we turn to the effect of $\Q$ by setting $\w_k = \BFzero$. 
In Corollary~\ref{corr:multi_SR_sig}, we derive the asymptotic Sharpe ratio with $\w_k = \BFzero$. 
Following the single-period analysis, we decompose this expression by defining the SR scalar 
$\kappa_\rho^{-T/2}$ 
and interpreting the remaining component as the pseudo Sharpe ratio. 

\begin{corollary}
    \label{corr:multi_SR_sig}
    Under the same assumptions in Theorem~\ref{theorem:pram_converge_sig}, if $\w = \BFzero$, $\Q_k = a_{k+1} \rho \bar\Q$, the Sharpe ratio will almost surely converge to $\SR_\infty(T, \rho, \BFzero)$,
    \begin{align*}
       \SR_\infty(T, \rho, \BFzero) 
       = \kappa_\rho^{-T/2} \cdot 
       \frac{ T^{-1/2}  ((\bmu^\top \A_\rho^{-1} \bmu + 1)^T - 1) }{\sqrt{(\bmu^\top \A_\rho^{-1} \bsigma \A_\rho^{-1}\bmu + 1)^T - 1}}\,,
    \end{align*}
    where $\kappa_\rho = 1 - \frac{\tilde s(\rho)}{(1+s(\rho))^2}$, $\A_\rho = \frac{\bsigma}{1+s(\rho)} + \rho \bar\Q$, 
    $s(\rho)$ is the solution of $s(\rho) = \frac{c}{p} \tr \bsigma ((\A_\rho + \bmu \bmu^\top)^{-1})$, 
    and $\tilde s(\rho)$ is given as the solution of $\tilde s(\rho) = -\kappa_\rho \frac{c}{p} \tr(\bsigma^2 (\A_\rho + \bmu \bmu^\top)^{-2})$. 
\end{corollary}

We now compare this result with that in Corollary~\ref{coro:multi_SR_mu_w_0}. 
To simplify the the comparison, we choose $\bar\Q = \bsigma$. 
In the low-dimensional setting where $p/n \to 0$, we obtain $s(\rho) = 0$ and $\tilde s(\rho) = 0$. 
In this case, we can verify that regardless of whether the estimation error arises from estimating the expected return or the covariance matrix, the asymptotic Sharpe ratio is identical, namely
\begin{equation*}
    \SR_\infty^{low}(T, \rho, \BFzero) 
    = \frac{T^{-1/2}((\bmu^\top \bsigma^{-1} \bmu + e_{\rho,\bmu})^T - e_{\rho,\bmu}^T)}{\sqrt{ ( \bmu^\top \bsigma^{-1} \bmu + e_{\rho,\bmu}^2)^T - e_{\rho,\bmu}^{2T} }}\,,
\end{equation*}
where $e_{\rho,\bmu} = \rho + 1$ in the low-dimensional setting. 
In Corollary~\ref{coro:multi_SR_mu_w_0}, we showed the optimal $\rho$ to maximize the asymptotic Sharpe ratio is 0 and the optimal value is indeed the maximum Sharpe in (\ref{SR_max}). 
However, under the high-dimensional regime where $p/n \to c > 0$, the asymptotic Sharpe ratio is substantially diminished by the estimation error in the covariance matrix. 
To be exact, the diminishing effect arises from the SR scalar 
$\kappa_\rho^{-T/2}$
, which is also our main motivation to apply a positive regularization. 

\begin{figure}[htbp]
    \centering
    \begin{subfigure}[t]{0.45\textwidth}
        \centering
        \includegraphics[width=\textwidth]{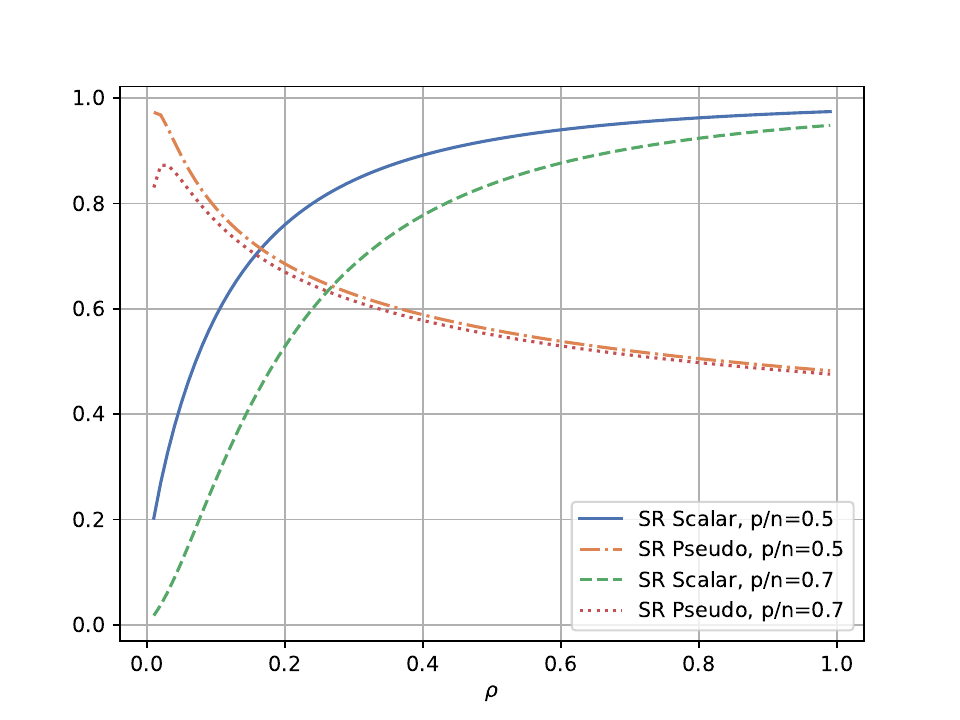}
        \caption{SR scalar and pseudo SR}
        \label{fig:multi_SR_sig_scalar_pseudo}
    \end{subfigure}
    \hfill
    \begin{subfigure}[t]{0.45\textwidth}
        \centering
        \includegraphics[width=\textwidth]{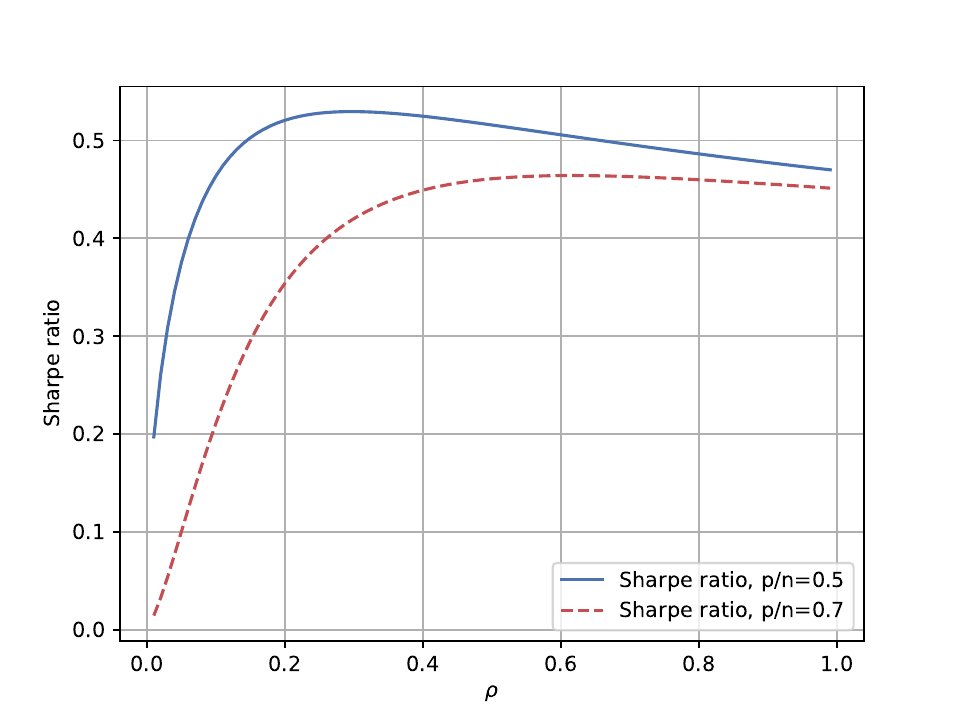}
        \caption{Sharpe ratio $\SR_\infty(T,\rho,\BFzero)$}
        \label{fig:multi_SR_sig_product}
    \end{subfigure}

    \caption{$\SR_{\infty}$ versus regularization $\rho$ with estimation error in $\hat{\bsigma}$ when $T>1$.}
    \label{fig:multi_SR_sig}

    \caption*{\footnotesize
    Notes. We set $T=4$. The SR scalar $\kappa_\rho$ increases and the pseudo SR decreases as $\rho$ grows. Hence,
    $\SR_\infty=\kappa_\rho \cdot \text{pseudo SR}$ is maximized at $\rho^\ast>0$.}
\end{figure}

\begin{figure}[htbp]
    \centering
    \begin{subfigure}[t]{0.45\textwidth}
        \centering
        \includegraphics[width=\textwidth]{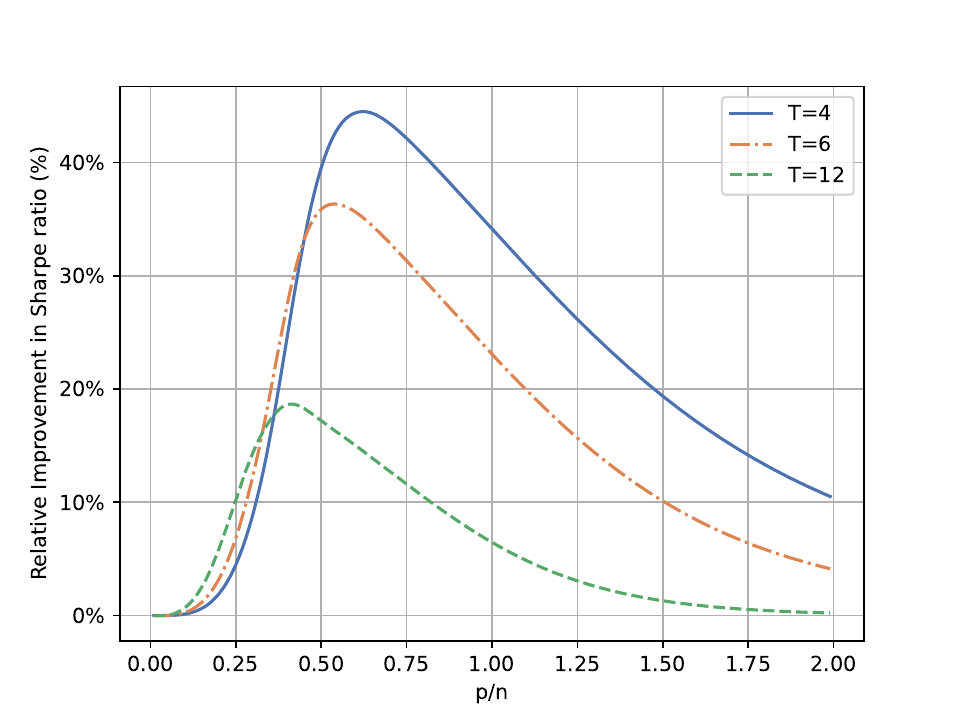}
        \caption{Relative improvement versus $p/n$}
        \label{fig:multi_SR_sig_T}
    \end{subfigure}
    \hfill
    \begin{subfigure}[t]{0.45\textwidth}
        \centering
        \includegraphics[width=\textwidth]{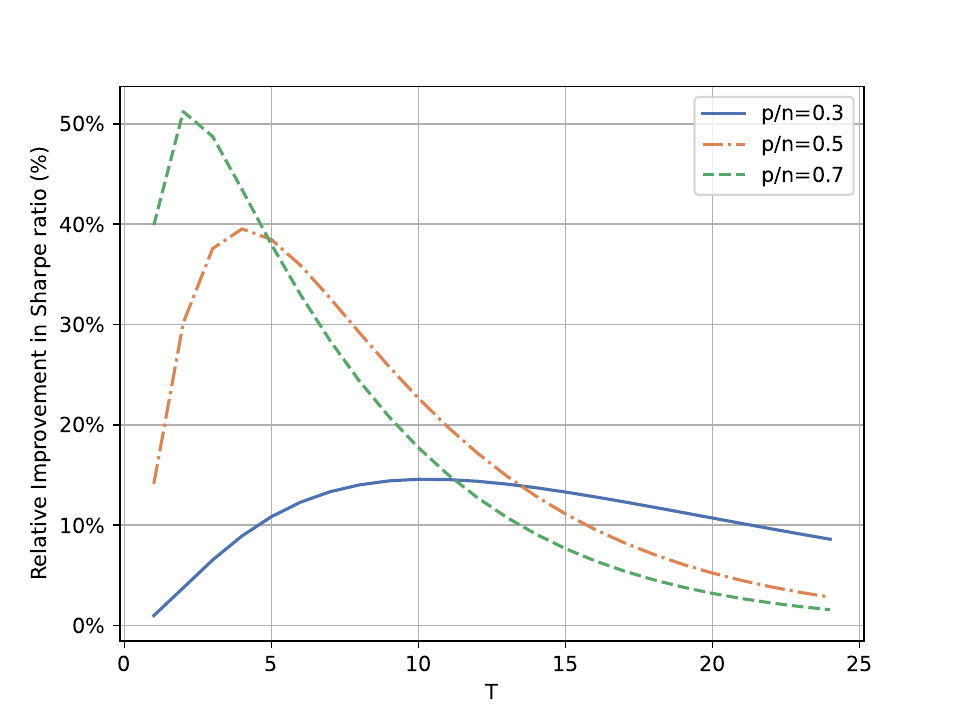}
        \caption{Relative improvement versus $T$}
        \label{fig:multi_SR_sig_rel_imp}
    \end{subfigure}

    \caption{The effect of $T$ and $p/n$ on the asymptotic Sharpe ratio with estimation error in $\hat{\bsigma}$.}
    \label{fig_multi_SR_sig_T_effect}

    \caption*{\footnotesize
    Notes. The relative improvement is defined as $\frac{\SR_\infty(T,\rho^\ast,\BFzero)-\SR_\infty(T,0,\BFzero)} {\SR_{\max}}$, 
    where $\rho^\ast$ is computed numerically.}
\end{figure}

We illustrate this numerically by conducting the same toy simulation as previous subsections. 
The curves of asymptotic Sharpe ratio versus $\rho$ are shown in Figure~\ref{fig:multi_SR_sig}, where we set $T=4$ and $c \in \{0.5, 0.7\}$. 
This figure clearly demonstrates that the optimal $\rho^\ast$ could be a proper positive value, in order to maximize the asymptotic SR value. 
Similar to the single-period case, a non-zero $\rho$ helps balance the SR scalar and the psuedo SR,  
so that $\SR_\infty(T, \rho, \BFzero)$ is always optimized at some $\rho^\ast >0$, which depends on both $p/n$ and $T$.

Finally, we examine how the investment horizon $T$ and $p/n$ influence the Sharpe ratio. 
We consider the relative improvement of the RRMV portfolio, defined as $(\SR_\infty(T, \rho^\ast, 0) - \SR_\infty(T, 0, 0))/ \SR_{\max}$, where $\SR_{\max}$ is the maximum Sharpe ratio defined in \eqref{SR_max}. 
Figure~\ref{fig_multi_SR_sig_T_effect} shows that the relative improvement in Sharpe ratio is bounded, and achieves its maximum at a finite $T>0$. As $T$ increases it converges to zero, because the maximum Sharpe ratio as the denominator grows exponentially in $T$. 
This pattern suggests that selecting a moderate investment horizon yields the greatest relative improvement. %while limiting the exposure to uncertainty from time-varying return distributions. 
The relative improvement is also predominant and robust across different choices of $p/n$. Comparing to the small relative improvement we observed in Figure~\ref{fig_multi_SR_mu_w_0} when we have estimation error in $\hat\bmu$, the regularization $\Q$ is generally much more effective in protecting against estimation error in $\hat\bsigma$.

\section{Empirical Evaluation} \label{se_numerical_evaluation} 

In this section, we evaluate and compare the out-of-sample performance of the proposed RRMV portfolio against a range of benchmark portfolio policies, using both  simulation experiments and empirical tests on real market data. 
We consider two industry portfolio datasets---the \textit{ten industry portfolios} (``10-Ind'') and the \textit{forty-eight industry portfolios} (``48-Ind'')---which represent the U.S.\ equity market at different levels of aggregation. 
Both datasets are obtained from Ken French's data library.\endnote{The industry portfolio datasets are available at \url{https://mba.tuck.dartmouth.edu/pages/faculty/ken.french/data_library.html}.} 
In addition, we study a stock-level dataset comprising constituents of the S\&P~500 index, obtained from the Center for Research in Security Prices (CRSP). 
After excluding stocks with missing observations, the resulting investment universe contains $p=280$ stocks. 
The sample period spans January~2000 to December~2024. Throughout our analysis, the risk-free return is proxied by the one-month U.S.\ Treasury bill yield, also sourced from CRSP.

\subsection{Benchmark Portfolios and Implementation} \label{sse:implement}
We compare the proposed portfolios with a comprehensive set of benchmark portfolios that are widely used in the literature. 
Table~\ref{table:portfolio} lists all benchmark portfolios included in our analysis. 
These benchmarks can be grouped into three categories. The first category consists of \textit{static portfolio policies}, including the equally weighted portfolio (EW), as well as the classical mean-variance (MV) and global minimum-variance (GMV) portfolios. 
For the MV and GMV policies, portfolios are constructed using either the sample covariance matrix or the linear shrinkage estimator of \citet{ledoit2004well} to mitigate the impact of estimation error. 
We denote shrinkage-based policies by appending ``-SH'' to the policy name, where ``SH'' indicates the use of the shrinkage estimator. 
In addition, we consider the three-fund hybrid portfolio policy proposed by \citet{kanOptimalPortfolioChoice2007}, where its risk averse is set as 3. 

The second category consists of conventional dynamic MMV policies without reference regulation, including the multiperiod mean-variance (MMV) policy of \citet{LiNg:2000MF}. This class serves as the primary dynamic benchmark for evaluating the effect of reference regulation.

The third category comprises our proposed reference-regulated multiperiod mean-variance (RRMV) policies. 
The regulation term in the RRMV framework is specified through two sets of hyperparameters, $\Q_k$ and $\w_k$, for $k = 0,\ldots,T-1$. 
In our numerical experiments, we set $\Q_k = \Q = \rho \I$, where $\rho > 0$ is a tuning parameter, and take $\w_k \equiv \w_0$ as a predetermined reference portfolio constructed from in-sample returns. 
We allow the reference portfolio $\w_0$ to correspond to several static benchmark portfolios. 
Specifically, we consider the EW portfolio, the GMV-SH portfolio, and the zero portfolio (i.e., $\w_0=\BFzero$), as these references remain well defined and stable in high-dimensional settings. 
When $\w_0=\BFzero$, the regulation term reduces to a pure $\ell_2$ penalty, and we denote this specification by $\text{RRMV-}\ell_2$.

\begin{table}[htbp]
    \centering
    \caption{List of portfolio policies.}
    \label{table:portfolio}
    \small
    \setlength{\tabcolsep}{8pt}
    \renewcommand{\arraystretch}{1.1}
    \begin{tabularx}{\textwidth}{>{\raggedright\arraybackslash}X l}
        \hline
        Portfolio & Abbreviation \\
        \hline

        \multicolumn{2}{l}{\textit{Static portfolios}} \\
        Equal weighted portfolio ($1/N$) & EW \\
        Mean variance portfolio with sample covariance & MV \\
        Mean variance portfolio with shrinkage covariance & MV-SH \\
        Global minimum variance portfolio with sample covariance & GMV \\
        Global minimum portfolio with shrinkage covariance & GMV-SH \\
        Three-fund portfolio by \cite{kanOptimalPortfolioChoice2007} & KZ \\

        \multicolumn{2}{l}{\textit{MMV portfolios}} \\
        Using sample covariance & MMV \\
        Using shrinkage covariance & MMV-SH \\

        \multicolumn{2}{l}{\textit{RRMV portfolios $(\Q_k=\rho\I$ and $\Q_k\w_k=\w_0)$}} \\
        Using EW as reference weight $\w_0$ & $\RRMV$-EW \\
        Using GMV-SH as reference weight $\w_0$ & $\RRMV$-GMV-SH \\
        Using a zero-vector portfolio as reference weight $\w_0$ & $\RRMV$-$\ell_2$ \\
        \hline
    \end{tabularx}
\end{table}

For our empirical analysis, we adopt a rolling estimation scheme with a fixed in-sample size of $n = 120$ months. 
At the beginning of each month $t$, portfolio policies are estimated using return data from the most recent $n$ months. 
The sample size $n$ is an important aspect for the construction of a high-dimensional portfolio. 
The resulting portfolio is then implemented and evaluated out-of-sample at month $t+T$, where $T$ is the investment horizon.
When $T > 1$, we apply single-period portfolio policies in a buy-and-hold manner: the portfolio constructed at month $t$ is held fixed over the next $T$ months, and performance is assessed based on the terminal wealth at the end of the $T$-month holding period. 

The out-of-sample performances of different portfolios are evaluated based on three commonly-used criteria: the out-of-sample Risk (standard deviation), the out-of-sample Sharpe ratio and the average turnover of portfolio weights. 
In practice, portfolio turnover and the associated trading costs constitute an important component of overall portfolio performance \citep{novy2016taxonomy}. 
All metrics are reported at the monthly frequency to facilitate comparison across portfolio policies. 

Suppose the portfolios with investment horizon $T$ are applied to $p$ assets for $M$ times.
Let $X_T^{(m)}$ be the terminal wealth and $\w_t^{(m)}$ is the normalized portfolio weight on risky assets in the $m$-th experiment. 
The expected return is $\hat{\mu} = \frac{1}{T} \frac{1}{M} \sum_{m=1}^{M} (X_T^{(m)} - r_f^{T})$, and these criteria above are defined as follows: 
\begin{equation} \label{eq:os_metrics}
    \begin{aligned} 
        \text{Risk}&:=\hat{\sigma} = \Big(\frac{1}{T} \frac{1}{M-1} \sum_{m=1}^{M} (X_T^{(m)} -\hat{\mu})^2 \Big)^{1/2}\,,\\
        \text{Sharpe}&: ={\hat{\mu}}/{\hat{\sigma}}\,, \\
        \text{Turnover}&:= \frac{1}{M(T-1)} \sum_{m=1}^M \sum_{t=0}^{T-2} (\|\w_{t+1}^{(m)} - \w_t^{(m)} \|_1) \,.
    \end{aligned}
\end{equation} 

\subsection{Simulation Comparisons}
We evaluate the performance of our RRMV portfolios relative to several benchmarks using simulated data calibrated from the 48-industry portfolio returns. 
The simulations test both normal and heavy-tailed (Student-$t$) distributions for asset returns. 
%with 1,000 Monte Carlo replications conducted for each scenario. 
The estimation window $n$ varies between 60 and 120 months, and investment horizons $T$ range from 1 to 6 years. 
The complete results can be found in  Appendix~\ref{sec:simulation}. 

Simulation results show that RRMV portfolios outperform unregularized strategies in high dimensions and exhibit steadily increasing Sharpe ratios as $T$ grows. 
With different estimation windows, RRMV portfolios deliver the highest Sharpe ratios while maintaining significantly lower turnover than multiperiod MV strategies. 
We also plot the Sharpe ratio curves as a function of $\rho$. The advantage of applying a positive $\rho$ is more pronounced in the high-dimensional regime. 
Overall, all the observations match with our theoretical understanding very well and reveal clearly the benefits of RRMV portfolios in the multiperiod settings.

\subsection{Real Data Analysis} 
In this section, we present out-of-sample results based on real market data. 
We select the regularization parameter using a rolling validation procedure based solely on historical data. 
The validation stage consists of $\tau$ rolling runs. Specifically, the portfolio decision at time $t$ uses data from $t-n-\tau-T+1$ to $t-1$.
In each run, $\hat\bmu$ and $\hat\bsigma$ are estimated using a window of $n$ past observations, and the resulting portfolio rule is applied over the subsequent $T$ periods to compute a realized multiperiod return. 
The estimation and evaluation windows then move forward by one period, and this process is repeated $\tau$ times, ending immediately before the test date $t$. 
The parameter with the highest empirical validation Sharpe ratio is selected and applied over the subsequent out-of-sample horizon of length $T$. 
Particularly, our validation procedure is used to select the regularization parameter $\rho$ for $\Q = \rho \I$. 

\subsubsection{Investment on industry portfolios}

Tables~\ref{table:real_10_ind} and \ref{table:real_48_ind} report the Risk, Sharpe, and Turnover for each benchmark portfolio under different investment horizons ($T=1$ and $T=6$). 
When we have $T=1$, MMV (MMV-SH) portfolios are equivalent to the MV  (MV-SH) portfolios. 
%o, we only keep the results of MV portfolio. 
From these tables, we observe that no static portfolio consistently dominates others. 
Specifically, the equally weighted (EW) strategy outperforms other static benchmarks for the ``10-Ind'' data, whereas the GMV-SH portfolio performs better for the ``48-Ind'' data. 
Notably, the RRMV policy achieves the highest Sharpe across all benchmarks.
Additionally, the RRMV policies with EW reference or zero reference exhibit relatively low turnover.
When we set $T=6$, we observe that the RRMV portfolios achieve even higher Sharpe ratios. 
In comparison, the MMV portfolios perform worse than the MV portfolios 
% with $T=1$
, which can be explained by the accumulation of estimation errors. 
As discussed in Section~\ref{sse:comparison_policies}, the RRMV policy can take advantage of the multiple investment horizon.

\begin{table}[htbp]
    \centering
    \caption{Out-of-sample performance on 10-industry portfolios.}
    \label{table:real_10_ind}
    \small
    \renewcommand{\arraystretch}{1.1}
    \setlength{\tabcolsep}{4pt}
    \begin{tabularx}{\textwidth}{l *{6}{>{\centering\arraybackslash}X}}
        \hline
        & \multicolumn{3}{c}{Investment period $T=1$}
        & \multicolumn{3}{c}{Investment period $T=6$} \\
        \cline{2-4} \cline{5-7}
        Portfolio & Risk & Sharpe & Turnover & Risk & Sharpe & Turnover \\
        \hline
        KZ              & 0.154 & 0.119 & 1.608 & 0.149 & 0.106 & 1.629 \\
        EW              & 0.043 & 0.230 & 0.000 & 0.034 & 0.305 & 0.000 \\
        MV              & \textbf{0.029} & 0.146 & 0.322 & \textbf{0.027} & 0.131 & 0.325 \\
        MV-SH           & \textbf{0.028} & 0.171 & 0.257 & \textbf{0.026} & 0.162 & 0.256 \\
        GMV             & 0.035 & 0.204 & 0.134 & 0.029 & 0.237 & 0.135 \\
        GMV-SH          & 0.035 & 0.218 & 0.086 & 0.028 & 0.260 & 0.086 \\
        MMV             & -- & -- & -- & 0.057 & 0.084 & 2.610 \\
        MMV-SH          & -- & -- & -- & 0.053 & 0.104 & 1.837 \\
        RRMV-EW         & \textbf{0.034} & \textbf{0.278} & 0.164 & \textbf{0.026} & \textbf{0.362} & 0.340 \\
        RRMV-GMV-SH     & 0.035 & 0.235 & 0.181 & 0.028 & 0.283 & 0.360 \\
        RRMV-$\ell_2$   & \textbf{0.032} & \textbf{0.266} & 0.147 & \textbf{0.027} & \textbf{0.322} & 0.350 \\
        \hline
    \end{tabularx}

    \caption*{\footnotesize
    Notes. Estimation uses a rolling window of length $n=120$. For RRMV portfolios, the regularization matrix is $Q=\rho I$, where $\rho$ is chosen by validation. MMV and MMV-SH are identical to MV and MV-SH, respectively, when $T=1$. The two best Sharpe values and four best Risk values are highlighted.}
\end{table}

\begin{table}[htbp]
    \centering
    \caption{Out-of-sample performance on 48-industry portfolios.}
    \label{table:real_48_ind}
    \small
    \renewcommand{\arraystretch}{1.1}
    \setlength{\tabcolsep}{4pt}
    \begin{tabularx}{\textwidth}{l *{6}{>{\centering\arraybackslash}X}}
        \hline
        & \multicolumn{3}{c}{Investment period $T=1$}
        & \multicolumn{3}{c}{Investment period $T=6$} \\
        \cline{2-4} \cline{5-7}
        Portfolio & Risk & Sharpe ratio & Turnover
        & Risk & Sharpe ratio & Turnover \\
        \hline
        KZ              & 0.197 & 0.143 & 5.210 & 0.210 & 0.111 & 5.308 \\
        EW              & 0.047 & 0.200 & 0.000 & 0.040 & 0.254 & 0.000 \\
        MV              & \textbf{0.022} & 0.148 & 0.600 & \textbf{0.020} & 0.128 & 0.611 \\
        MV-SH           & \textbf{0.019} & 0.171 & 0.360 & \textbf{0.018} & 0.163 & 0.365 \\
        GMV             & 0.040 & 0.219 & 0.737 & 0.036 & 0.236 & 0.742 \\
        GMV-SH          & 0.035 & 0.261 & 0.321 & 0.029 & 0.295 & 0.322 \\
        MMV             & -- & -- & -- & 0.167 & 0.019 & 18.696 \\
        MMV-SH          & -- & -- & -- & 0.072 & 0.075 & 10.074 \\
        RRMV-EW         & \textbf{0.030} & \textbf{0.275} & 0.268 & \textbf{0.022} & \textbf{0.312} & 0.400 \\
        RRMV-GMV-SH     & 0.033 & \textbf{0.267} & 0.404 & 0.026 & \textbf{0.329} & 0.717 \\
        RRMV-$\ell_2$   & \textbf{0.027} & 0.235 & 0.226 & \textbf{0.022} & 0.273 & 0.406 \\
        \hline
    \end{tabularx}

    \captionsetup{justification=raggedright,singlelinecheck=false}
    \caption*{\footnotesize
    Notes. Same as Table~\ref{table:real_10_ind}.}
\end{table}

\subsubsection{Investment on S\&P 500 constituents}

We next report the out-of-sample performance of all benchmark portfolios using individual stocks from the S\&P 500 index. 
The cross-sectional dimension is $p=280$, which exceeds the estimation sample size $n=120$. 
To account for this, we consider two empirical settings. 
In the first setting, we randomly select subsets of 100 stocks from the full universe of 280 constituents. 
This design allows us to evaluate portfolio performance in a moderately high-dimensional environment where the dimensionality ratio $p/n$ is close to, but still less than, one. 
The second setting uses all 280 stocks, representing a more challenging high-dimensional regime with only limited samples.

Table~\ref{table:real_sp500_sub} reports the results for the first setting, where $p/n \approx 1$. 
The results show that benchmark policies based on the sample covariance matrix (MV, KZ, GMV, and MMV) all perform poorly out-of-sample, as evidenced by low Sharpe ratios and high turnover rates. 
This underperformance becomes more pronounced when the investment horizon is extended to $T=6$. 
In contrast, the equally weighted policy and benchmark portfolios constructed using shrinkage covariance estimators exhibit greater robustness in this setting. 
Consistent with the findings in Appendix~\ref{sec:simulation}, the RRMV-EW policy 
achieves the highest Sharpe ratio among all competing portfolio strategies. 
Extending the investment horizon for RRMV portfolios leads to higher Sharpe ratios, highlighting the advantage of feedback-based multiperiod portfolio policies.

Table~\ref{table:real_sp500_all} presents the out-of-sample performance using all 280 S\&P 500 constituent stocks. 
In this case, the curse of high dimensionality becomes more severe, with $p/n = 280/120 > 1$. 
As a result, the sample covariance matrix is singular and cannot be inverted, rendering it unsuitable for portfolio rules to rely on sample covariance. 
Consequently, the benchmark set is restricted to portfolios based on shrinkage covariance estimators and the equally weighted (EW) policy. 
As shown in Table~\ref{table:real_sp500_all}, when $T=6$, the GMV portfolio constructed using the shrinkage covariance estimator (GMV-SH) performs well, achieving a Sharpe ratio of 0.323. 
Nevertheless, our RRMV policy  with the EW reference portfolio continues to outperform all benchmark strategies while also exhibiting a lower turnover rate than GMV-SH. 
Although our approach also utilizes the sample covariance matrix, the regularization term in the RRMV formulation effectively induces shrinkage in both the covariance matrix and the expected return vector. 
This joint regularization substantially improves robustness to estimation error and results in superior out-of-sample Sharpe ratios. 

\begin{table}[htbp]
    \centering
    \caption{Out-of-sample performance on 100 stocks selected from the S\&P 500 constituents.}
    \label{table:real_sp500_sub}
    \small
    \renewcommand{\arraystretch}{1.1}
    \setlength{\tabcolsep}{4pt}
    \begin{tabularx}{\textwidth}{l *{6}{>{\centering\arraybackslash}X}}
        \hline
        & \multicolumn{3}{c}{Investment period $T=1$}
        & \multicolumn{3}{c}{Investment period $T=6$} \\
        \cline{2-4} \cline{5-7}
        Portfolio & Risk & Sharpe ratio & Turnover
        & Risk & Sharpe ratio & Turnover \\
        \hline
        KZ              & 0.112 & 0.034  & 4.850 & 0.104  & -0.039 & 4.896 \\
        EW              & 0.044 & 0.250  & 0.000 & 0.037  & \textbf{0.331} & 0.000 \\
        MV              & \textbf{0.030} & -0.008 & 1.402 & 0.032  & -0.062 & 1.421 \\
        MV-SH           & \textbf{0.016} & 0.112  & 0.249 & \textbf{0.017} & 0.104 & 0.251 \\
        GMV             & 0.063 & 0.104  & 2.694 & 0.064  & 0.118 & 2.667 \\
        GMV-SH          & 0.035 & 0.241  & 0.317 & 0.034  & 0.237 & 0.311 \\
        MMV             & -- & -- & -- & 15.161 & 0.019 & 490.964 \\
        MMV-SH          & -- & -- & -- & 0.087  & 0.061 & 8.647 \\
        RRMV-EW         & \textbf{0.023} & \textbf{0.312} & 0.214 & \textbf{0.019} & \textbf{0.349} & 0.520 \\
        RRMV-GMV-SH     & 0.033 & 0.247  & 0.398 & \textbf{0.030} & 0.274 & 0.939 \\
        RRMV-$\ell_2$   & \textbf{0.019} & \textbf{0.296} & 0.184 & \textbf{0.019} & 0.278 & 0.481 \\
        \hline
    \end{tabularx}

    \captionsetup{justification=raggedright,singlelinecheck=false}
    \caption*{\footnotesize
    Notes. Estimation uses a rolling window of length $n=120$ with cross-sectional dimension $p=100$. For RRMV portfolios, the regularization matrix is $Q=\rho I$, where $\rho$ is chosen by validation. The two best Sharpe values and four best Risk values are highlighted.}
\end{table}

\begin{table}[htbp]
    \centering
    \caption{Out-of-sample performance on all S\&P 500 constituents.}
    \label{table:real_sp500_all}
    \small
    \renewcommand{\arraystretch}{1.1}
    \setlength{\tabcolsep}{4pt}
    \begin{tabularx}{\textwidth}{l *{6}{>{\centering\arraybackslash}X}}
        \hline
        & \multicolumn{3}{c}{Investment period $T=1$}
        & \multicolumn{3}{c}{Investment period $T=6$} \\
        \cline{2-4} \cline{5-7}
        Portfolio & Risk & Sharpe ratio & Turnover
        & Risk & Sharpe ratio & Turnover \\
        \hline
        EW              & 0.047 & 0.241 & 0.000 & 0.038 & 0.309 & 0.000 \\
        MV-SH           & \textbf{0.013} & 0.176 & 0.263 & \textbf{0.012} & 0.192 & 0.271 \\
        GMV-SH          & 0.033 & \textbf{0.323} & 0.494 & 0.033 & 0.323 & 0.490 \\
        MMV-SH          & -- & -- & -- & 0.095 & 0.127 & 17.258 \\
        RRMV-EW         & \textbf{0.023} & 0.263 & 0.167 & \textbf{0.015} & \textbf{0.467} & 0.398 \\
        RRMV-GMV-SH     & \textbf{0.031} & \textbf{0.327} & 0.525 & \textbf{0.029} & \textbf{0.413} & 0.882 \\
        RRMV-$\ell_2$   & \textbf{0.019} & 0.247 & 0.157 & \textbf{0.015} & 0.364 & 0.350 \\
        \hline
    \end{tabularx}

    \captionsetup{justification=raggedright,singlelinecheck=false}
    \caption*{\footnotesize
    Notes. Same as Table~\ref{table:real_sp500_sub}, except that the cross-sectional dimension is $p=280$.}
\end{table}

\subsubsection{Index enhancement portfolios}\label{subsec_index_enhance}

Finally, we present an application of the RRMV model that exploits its flexibility in specifying the reference policy. 
Since the reference portfolio in the RRMV framework can be chosen arbitrarily, we adopt an index-tracking policy (see \citet{jansen2002optimal, benidis2017sparse, strub2018optimal}) as the reference in order to incorporate information embedded in a benchmark index within the mean-variance framework. The corresponding index-tracking policy is denoted by IT. 
We refer to this approach as the \emph{index enhancement} policy and denote it by RRMV-IT.

For practical implementation, rather than directly using the index weights, we construct an index-tracking portfolio using the same set of assets as in the investment universe. 
This ensures that both the RRMV optimal policy and the reference portfolio are defined over an identical asset set. 
Specifically, we obtain the index-tracking policy by minimizing the empirical tracking error $\frac{1}{n}\sum_{t=1}^{n} \big| \R_t^\top \w - r_t \big|$, 
where $\R_t \in \mR^p$ denotes the return vector of the constituent stocks in the investment universe at time $t$, $r_t \in \mR$ is the return of the index, and $n$ is the length of the historical sample. 
We further impose the constraints $\w \ge 0$ and $\e^{\top} \w = 1$. 
We use the monthly returns of the S\&P 500 as the index returns and consider the 10-Ind and 48-Ind datasets as the asset pools. 
The optimal solution of the above minimization is then used as the reference allocation in the RRMV framework, and all remaining experimental settings follow the previous sections.

Table~\ref{table:real_index_enhance} reports detailed out-of-sample performance metrics for the evaluated portfolio policies. 
We employ the same validation procedure to select the regularization parameter $\rho$ for $\Q = \rho \I$. 
As shown in Table~\ref{table:real_index_enhance}, the index enhancement policy (RRMV-IT) outperform all the benchmarks, including RRMV using other references, when $T=1$ and $T=6$.  
Consistent with earlier results, we find that RRMV-IT  portfolios exhibit the highest Sharpe ratio and lower turnover rates, leading to lower transaction costs in practice.  
This section provides another meaningful candidate as the reference of the RRMV portfolio. 

\begin{table}[htbp]
    \centering
    \caption{Comparison of RRMV-IT with multiperiod portfolio performance on industry portfolios.}
    \label{table:real_index_enhance}
    \small
    \renewcommand{\arraystretch}{1.1}
    \setlength{\tabcolsep}{4pt}
    \begin{tabularx}{\textwidth}{l *{6}{>{\centering\arraybackslash}X}}
        \hline
        & \multicolumn{3}{c}{10-Ind}
        & \multicolumn{3}{c}{48-Ind} \\
        \cline{2-4} \cline{5-7}
        Portfolio & Risk & Sharpe ratio & Turnover
        & Risk & Sharpe ratio & Turnover \\
        \hline

        \multicolumn{7}{l}{\textbf{Investment period $T=1$}} \\
        IT       & 0.041 & 0.243 & 0.021 & 0.041 & 0.254 & 0.057 \\
        EW       & 0.043 & 0.230 & 0.000 & 0.047 & 0.200 & 0.000 \\
        MV-SH    & \textbf{0.028} & 0.171 & 0.257 & \textbf{0.019} & 0.171 & 0.360 \\
        RRMV-EW  & 0.034 & 0.278 & 0.164 & 0.030 & 0.275 & 0.268 \\
        RRMV-IT  & 0.034 & 0.288 & 0.167 & 0.028 & \textbf{0.336} & 0.261 \\

        \multicolumn{7}{l}{\textbf{Investment period $T=6$}} \\
        IT       & 0.033 & 0.326 & 0.021 & 0.033 & \textbf{0.336} & 0.057 \\
        EW       & 0.034 & 0.305 & 0.000 & 0.040 & 0.254 & 0.000 \\
        MV-SH    & \textbf{0.026} & 0.162 & 0.256 & \textbf{0.018} & 0.163 & 0.365 \\
        MMV-SH   & 0.053 & 0.104 & 1.837 & 0.072 & 0.075 & 10.074 \\
        RRMV-EW  & \textbf{0.026} & \textbf{0.362} & 0.340 & \textbf{0.022} & 0.312 & 0.400 \\
        RRMV-IT  & \textbf{0.026} & \textbf{0.383} & 0.339 & \textbf{0.021} & \textbf{0.399} & 0.422 \\
        \hline
    \end{tabularx}

    \captionsetup{justification=raggedright,singlelinecheck=false}
    \caption*{\footnotesize
    Notes. The estimation window has length $n=120$. The regularization matrix is $Q=\rho I$, where $\rho$ is selected using the same validation procedure. The results for MV-SH, EW, MMV-SH, and RRMV-EW are copied from the previous tables for comparison. The two best Sharpe values and four best Risk values are highlighted.}
\end{table}

\section{Conclusion}\label{sec:Conclusion}
In this paper, we investigate the multiperiod mean–variance portfolio optimization. 
As shown in the previous literature and our derivations, incorporating feedback policies can achieve Sharpe ratios higher than static allocation rules. 
Nevertheless, estimation error remains a central practical difficulty, as its effect is harder to analyze given that a sequence of future decisions is influenced rather than a single allocation. 
The issue is even more serious in large portfolios where the number of assets is not small relative to the available data. 
In order to protect portfolio performance against estimation error, we propose the reference-regulated mean–variance (RRMV) portfolio optimization model. The corresponding algorithm based on dynamic programming techniques for solving this problem is carefully presented.

To understand the role of reference regulation in high-dimensional portfolio allocation in multiperiod settings, we provide a general framework to analyze the asymptotic out-of-sample Sharpe ratio in high-dimensional environments. 
Our theoretical results show that regularization and the use of a reference portfolio can improve out-of-sample Sharpe ratios by mitigating the adversarial impact of estimation errors in the mean vector and the covariance matrix. 
Specifically, when $T>1$ and $p/n > 0$, we unravel a new phenomenon of cross regularization: $\w$ can be used to regularize $\hat\bsigma$ and $\Q$ can be used to regularize $\hat\bmu$, which expands our conventional understanding for the single-period portfolio allocation where $\Q$ regularizes $\hat\bsigma$ and $\w$ regularizes $\hat\bmu$. We also explicitly present the effect of dimensionality $p/n$, optimization horizon $T$ and regularization terms to the asymptotic Sharpe ratio. This could lead to new insights for portfolio managers about how to impose a more proper regularization and how to select a meaningful reference policy. In addition, to our best knowledge, we are the first to present a rigorous result on how the number of periods affects the Sharpe ratio with the existence of  estimation errors.
 
Extensive simulation and empirical evidence further demonstrate that the proposed framework can effectively tackle the estimation error and lead to more reliable portfolio performance. 
The focus of the present paper is the proposed multiperiod RRMV portfolios, through which we reveal the role of regularization and reference portfolios in improving out-of-sample performance. 
Extending the framework to allow for time-varying return moments, alternative estimators, and other types of regularization or optimization settings is left for our future research. 

% \THEEndNotes
\begingroup \parindent 0pt \parskip 0.0ex \def\enotesize{\normalsize} \theendnotes \endgroup

% References here (outcomment the appropriate case)
\bibliographystyle{agsm}
\bibliography{dynamic+mv}

@article{Lassance2025JFQA,
	author = {Lassance, Nathan and Vanderveken, Rodolphe and Vrins, Fr{\'e}d{\'e}ric},
	edition = {2025/12/04},
	isbn = {0022-1090},
	journal = {Journal of Financial and Quantitative Analysis},
	pages = {1-41},
	publisher = {Cambridge University Press},
	title = {Optimal Portfolio Size Under Parameter Uncertainty},
	year = {2025},
    }

@article{DeMiguel-JFQA-2015,
	author = {DeMiguel, Victor and Mart{\'\i}n-Utrera, Alberto and Nogales, Francisco J.},
	edition = {2016/01/20},
	isbn = {0022-1090},
	journal = {Journal of Financial and Quantitative Analysis},
	number = {6},
	pages = {1443-1471},
	publisher = {Cambridge University Press},
	title = {Parameter Uncertainty in Multiperiod Portfolio Optimization with Transaction Costs},
	volume = {50},
	year = {2015},}

@article{Ao2019,
	author = {Ao, Mengmeng and Yingying, Li and Zheng, Xinghua},
	journal = {The Review of Financial Studies},
	month = jul,
	number = {7},
	pages = {2890--2919},
	title = {Approaching mean-variance efficiency for large portfolios},
	urldate = {2025-02-14},
	volume = {32},
	year = 2019}

@article{bai2009enhancement,
	author = {Bai, Zhidong and Liu, Huixia and Wong, Wing-Keung},
	journal = {Mathematical Finance: An International Journal of Mathematics, Statistics and Financial Economics},
	number = {4},
	pages = {639--667},
	publisher = {Wiley Online Library},
	title = {Enhancement of the applicability of {{Markowitz}}'s portfolio optimization by utilizing random matrix theory},
	volume = {19},
	year = 2009}

@article{basak2010dynamic,
	author = {Basak, Suleyman and Chabakauri, Georgy},
	journal = {The Review of Financial Studies},
	number = {8},
	pages = {2970--3016},
	publisher = {Oxford University Press},
	title = {Dynamic mean-variance asset allocation},
	volume = {23},
	year = 2010}

@article{benidis2017sparse,
	author = {Benidis, Konstantinos and Feng, Yiyong and Palomar, Daniel P},
	journal = {IEEE Transactions on signal processing},
	number = {1},
	pages = {155--170},
	publisher = {IEEE},
	title = {Sparse portfolios for high-dimensional financial index tracking},
	volume = {66},
	year = 2017}

@article{BertsimasGuptaPaschalidis2012,
	author = {Bertsimas, Dimitris and Gupta, Vishal and Paschalidis, Ioannis Ch.},
	journal = {Operations Research},
	number = {6},
	pages = {1389--1403},
	title = {Inverse optimization: a new perspective on the black-litterman model},
	volume = {60},
	year = 2012}

@article{bickel2008regularized,
	author = {Bickel, Peter J and Levina, Elizaveta},
	journal = {The Annals of Statistics},
	pages = {199--227},
	publisher = {JSTOR},
	title = {Regularized estimation of large covariance matrices},
	year = 2008}

@article{BieleckiJinZhou,
	author = {Bielecki, Tomasz R. and Jin, Hanqing and Pliska, Stanley R. and Zhou, Xun Yu},
	journal = {Mathematical Finance},
	number = {2},
	pages = {213--244},
	title = {Continuous-time mean-variance portfolio selection with bankruptcy prohibition},
	volume = {15},
	year = 2005}

@article{bjork2014mean,
	author = {Bj{\"o}rk, Tomas and Murgoci, Agatha and Zhou, Xun Yu},
	journal = {Mathematical Finance: An International Journal of Mathematics, Statistics and Financial Economics},
	number = {1},
	pages = {1--24},
	publisher = {Wiley Online Library},
	title = {Mean--variance portfolio optimization with state-dependent risk aversion},
	volume = {24},
	year = 2014}

@article{BlackLitterman1991,
	author = {Black, F. and Litterman, R. B.},
	journal = {Journal of Fixed Income},
	pages = {7--18},
	title = {Asset allocation: {{Combining}} investor views with market equilibrium},
	volume = {1},
	year = 1991}

@article{BlanchetChenZhou:2022,
	author = {Blanchet, J. and Chen, L. and Zhou, X. Y.},
	journal = {Management Science},
	number = {9},
	pages = {6382--6410},
	title = {Distributionally robust mean-variance portfolio selection with {{Wasserstein}} distances},
	volume = {68},
	year = 2022}

@incollection{BRANDT2010-book,
	address = {San Diego},
	author = {Brandt, Michael W.},
	booktitle = {Handbook of financial econometrics: {{Tools}} and techniques},
	editor = {{A{\"I}T-SAHALIA}, {\relax YACINE} and HANSEN, LARS PETER},
	pages = {269--336},
	publisher = {North-Holland},
	series = {Handbooks in finance},
	title = {{{CHAPTER}} 5 - portfolio choice problems},
	volume = {1},
	year = 2010}

@article{cai2011adaptive,
	author = {Cai, Tony and Liu, Weidong},
	journal = {Journal of the American Statistical Association},
	number = {494},
	pages = {672--684},
	publisher = {Taylor \& Francis},
	title = {Adaptive thresholding for sparse covariance matrix estimation},
	volume = {106},
	year = 2011}

@article{cai2020high,
	author = {Cai, T Tony and Hu, Jianchang and Li, Yingying and Zheng, Xinghua},
	journal = {Journal of Econometrics},
	number = {2},
	pages = {482--494},
	publisher = {Elsevier},
	title = {High-dimensional minimum variance portfolio estimation based on high-frequency data},
	volume = {214},
	year = 2020}

@article{ChiuZhou2011,
	author = {Chiu, C. H. and Zhou, X. Y.},
	journal = {Quantitative Finance},
	pages = {115--123},
	title = {The prmium of dynamic trading},
	volume = {11},
	year = 2011}

@article{ChoraZiema1993,
	author = {Chopra, V. K. and Ziemba, W. T.},
	journal = {Journal of Portfolio Management},
	pages = {6--11},
	title = {The effect of errors in means, variances, and covariances on optimal portfolio choice},
	volume = {19},
	year = 1993}

@article{CuiGaoLiLi:2014,
	author = {Cui, X. Y. and Gao, J. J. and Li, X. and Li, D.},
	number = {2},
	pages = {459--468},
	title = {Optimal multi-period mean-variance policy under no-shorting constraint},
	volume = {234},
	year = 2014}

@article{demiguel2009generalized,
	author = {DeMiguel, Victor and Garlappi, Lorenzo and Nogales, Francisco J and Uppal, Raman},
	journal = {Management science},
	number = {5},
	pages = {798--812},
	publisher = {INFORMS},
	title = {A generalized approach to portfolio optimization: {{Improving}} performance by constraining portfolio norms},
	volume = {55},
	year = 2009}

@article{demiguel2009optimal,
	author = {DeMiguel, Victor and Garlappi, Lorenzo and Uppal, Raman},
	journal = {The Review of Financial Studies},
	number = {5},
	pages = {1915--1953},
	publisher = {Oxford University Press},
	title = {Optimal versus naive diversification: {{How}} inefficient is the 1/{{N}} portfolio strategy?},
	volume = {22},
	year = 2009}

@article{fan2012vast,
	author = {Fan, Jianqing and Zhang, Jingjin and Yu, Ke},
	journal = {Journal of the American Statistical Association},
	number = {498},
	pages = {592--606},
	publisher = {Taylor \& Francis},
	title = {Vast portfolio selection with gross-exposure constraints},
	volume = {107},
	year = 2012}

@article{fan2013large,
	author = {Fan, Jianqing and Liao, Yuan and Mincheva, Martina},
	journal = {Journal of the Royal Statistical Society Series B: Statistical Methodology},
	number = {4},
	pages = {603--680},
	publisher = {Oxford University Press},
	title = {Large covariance estimation by thresholding principal orthogonal complements},
	volume = {75},
	year = 2013}

@article{fan2016overview,
	author = {Fan, Jianqing and Liao, Yuan and Liu, Han},
	journal = {The Econometrics Journal},
	number = {1},
	pages = {C1--C32},
	publisher = {Oxford University Press Oxford, UK},
	title = {An overview of the estimation of large covariance and precision matrices},
	volume = {19},
	year = 2016}

@book{FollmerSchied:2004,
	author = {F{\"o}llmer, H. and Schied, A.},
	publisher = {Walter De Gruyter:Berlin},
	series = {De gruyter studies in mathematics},
	title = {Stochastic finance: {{An}} introduction in discrete time},
	year = 2004}

@article{GarleanuPedersen:2013JOF,
	author = {G{\^a}rleanu, N. and Pedersen, L. H.},
	journal = {Journal of Finance},
	number = {6},
	pages = {2309--2340},
	title = {Dynamic trading with predictable returns and transaction cost},
	volume = {68},
	year = 2013}

@article{Goldfarb:2003,
	author = {Goldfarb, D. and Iyengar, G.},
	journal = {Mathematics of Operations Research},
	number = {1},
	pages = {1--38},
	title = {Robust portfolio selection problems},
	volume = {28},
	year = 2003}

@article{Jagannathan_jof2003,
	author = {Jagannathan, Ravi and Ma, Tongshu},
	journal = {The journal of finance},
	number = {4},
	pages = {1651--1683},
	publisher = {Wiley Online Library},
	title = {Risk reduction in large portfolios: {{Why}} imposing the wrong constraints helps},
	volume = {58},
	year = 2003}

@article{jansen2002optimal,
	author = {Jansen, Roel and Van Dijk, Ronald},
	journal = {Journal of Portfolio Management},
	number = {2},
	pages = {33},
	publisher = {Pageant Media},
	title = {Optimal benchmark tracking with small portfolios},
	volume = {28},
	year = 2002}

@article{kanOptimalPortfolioChoice2007,
	author = {Kan, Raymond and Zhou, Guofu},
	journal = {Journal of Financial and Quantitative Analysis},
	month = sep,
	number = {3},
	pages = {621--656},
	title = {Optimal portfolio choice with parameter uncertainty},
	urldate = {2022-10-03},
	volume = {42},
	year = 2007}

@article{KanWang-MS-2024,
	author = {Kan, Raymond and Wang, Xiaolu},
	journal = {Management Science},
	number = {9},
	pages = {6117--6138},
	title = {Optimal portfolio choice with unknown benchmark efficiency},
	volume = {70},
	year = 2024}

@article{KanWangZhou-allrisky-2022,
	author = {Kan, Raymond and Wang, Xiaolu and Zhou, Guofu},
	journal = {Management Science},
	number = {3},
	pages = {2047--2068},
	title = {Optimal portfolio choice with estimation risk: {{No}} risk-free asset case},
	volume = {68},
	year = 2022}

@article{Karoui2010,
	author = {Karoui, Noureddine El},
	journal = {The Annals of Statistics},
	number = {6},
	pages = {3487--3566},
	publisher = {Institute of Mathematical Statistics},
	title = {High-dimensionality effects in the markowitz probelm and other quadratic programs with linear constraints: {{Risk}} underestimation},
	urldate = {2024-05-28},
	volume = {38},
	year = 2010}

@article{Kolm:2014,
	author = {{Kolm.P.N.} and T{\"u}t{\"u}nc{\"u}, R. and Fabozzi, F.J.},
	journal = {European Journal of Operational Research},
	pages = {356--371},
	publisher = {null},
	title = {60 years of portfolio optimization: {{Practical}} challenges and current trends},
	volume = {234(2)},
	year = 2014}

@article{Lai2011,
	author = {Lai, Tze Leung and Xing, Haipeng and Chen, Zehao},
	journal = {The Annals of Applied Statistics},
	number = {2A},
	pages = {798--823},
	publisher = {Institute of Mathematical Statistics},
	title = {Mean-variance portfolio optimization when means and covariances are unknown},
	urldate = {2026-01-30},
	volume = {5},
	year = 2011}

@article{lam2009sparsistency,
	author = {Lam, Clifford and Fan, Jianqing},
	journal = {Annals of statistics},
	number = {6B},
	pages = {4254},
	publisher = {NIH Public Access},
	title = {Sparsistency and rates of convergence in large covariance matrix estimation},
	volume = {37},
	year = 2009}

@article{Lassance-MS-2024,
	author = {Lassance, Nathan and {Mart{\'\i}n-Utrera}, Alberto and Simaan, Majeed},
	journal = {Management Science},
	number = {11},
	pages = {7644--7663},
	title = {The risk of expected utility under parameter uncertainty},
	volume = {70},
	year = 2024}

@article{ledoit_improved_2003,
	author = {Ledoit, Olivier and Wolf, Michael},
	journal = {Journal of Empirical Finance},
	number = {5},
	pages = {603--621},
	title = {Improved estimation of the covariance matrix of stock returns with an application to portfolio selection},
	urldate = {2022-01-13},
	volume = {10},
	year = 2003}

@article{ledoit2004well,
	author = {Ledoit, Olivier and Wolf, Michael},
	journal = {Journal of multivariate analysis},
	number = {2},
	pages = {365--411},
	publisher = {Elsevier},
	title = {A well-conditioned estimator for large-dimensional covariance matrices},
	volume = {88},
	year = 2004}

@article{ledoit2012nonlinear,
	author = {Ledoit, Olivier and Wolf, Michael},
	journal = {The Annals of Statistics},
	number = {2},
	pages = {1024--1060},
	title = {Nonlinear shrinkage estimation of large-dimensional covariance matrices},
	volume = {40},
	year = 2012}

@article{LiNg:2000MF,
	author = {Li, Duan and Ng, Wan-Lung},
	journal = {Mathematical Finance},
	number = {3},
	pages = {387--406},
	publisher = {Wiley Online Library},
	title = {Optimal dynamic portfolio selection: {{Multiperiod}} mean-variance formulation},
	volume = {10},
	year = 2000}

@article{markowitz1952harry,
	author = {Markowitz, Harry},
	journal = {Journal of Finance},
	number = {1},
	pages = {77--91},
	title = {Portfolio selection},
	volume = {7},
	year = 1952}

@article{marvcenko1967distribution,
	author = {Mar{\v c}enko, Vladimir A and Pastur, Leonid Andreevich},
	journal = {Mathematics of the USSR-Sbornik},
	number = {4},
	pages = {457},
	publisher = {IOP Publishing},
	title = {Distribution of eigenvalues for some sets of random matrices},
	volume = {1},
	year = 1967}

@article{MengCaoWang2025,
	author = {Meng, Xuran and Cao, Yuan and Wang, Weichen},
	journal = {Journal of the American Statistical Association},
	number = {0},
	pages = {1--13},
	publisher = {Taylor \& Francis},
	title = {Estimation of out-of-sample sharpe ratio for high dimensional portfolio optimization},
	volume = {0},
	year = 2025}

@article{novy2016taxonomy,
	author = {{Novy-Marx}, Robert and Velikov, Mihail},
	journal = {The Review of Financial Studies},
	number = {1},
	pages = {104--147},
	publisher = {Oxford University Press},
	title = {A taxonomy of anomalies and their trading costs},
	volume = {29},
	year = 2016}

@article{Rockafellar:2002,
	author = {Rockafellar, R.T. and Uryasev, S.},
	journal = {Journal of Banking \& Finance},
	pages = {1443--1471},
	publisher = {null},
	title = {Conditional value-at-risk for general loss distributions},
	volume = {26},
	year = 2002}

@article{rubio2011spectral,
	author = {Rubio, Francisco and Mestre, Xavier},
	journal = {Statistics \& probability letters},
	number = {5},
	pages = {592--602},
	publisher = {Elsevier},
	title = {Spectral convergence for a general class of random matrices},
	volume = {81},
	year = 2011}

@article{strub2018optimal,
	author = {Strub, Oliver and Baumann, Philipp},
	journal = {European journal of operational research},
	number = {1},
	pages = {370--387},
	publisher = {Elsevier},
	title = {Optimal construction and rebalancing of index-tracking portfolios},
	volume = {264},
	year = 2018}

@article{Tu2011,
	author = {Tu, Jun and Zhou, Guofu},
	journal = {Journal of Financial Economics},
	month = jan,
	number = {1},
	pages = {204--215},
	shorttitle = {Markowitz meets talmud},
	title = {Markowitz meets talmud: a combination of sophisticated and naive diversification strategies},
	urldate = {2025-08-07},
	volume = {99},
	year = 2011}

@article{van2021distribution,
	author = {{van Staden}, Pieter M and Dang, Duy-Minh and Forsyth, Peter A},
	journal = {SIAM Journal on Financial Mathematics},
	number = {2},
	pages = {566--603},
	publisher = {SIAM},
	title = {On the distribution of terminal wealth under dynamic mean-variance optimal investment strategies},
	volume = {12},
	year = 2021}

@article{van2021surprising,
	author = {{van Staden}, Pieter M and Dang, Duy-Minh and Forsyth, Peter A},
	journal = {European Journal of Operational Research},
	number = {2},
	pages = {774--792},
	publisher = {Elsevier},
	title = {The surprising robustness of dynamic mean-variance portfolio optimization to model misspecification errors},
	volume = {289},
	year = 2021}

@article{vigna2014efficiency,
	author = {Vigna, Elena},
	journal = {Quantitative finance},
	number = {2},
	pages = {237--258},
	publisher = {Taylor \& Francis},
	title = {On efficiency of mean--variance based portfolio selection in defined contribution pension schemes},
	volume = {14},
	year = 2014}

@article{Vigna2020,
	author = {Vigna, Elena},
	journal = {International Journal of Theoretical and Applied Finance},
	number = {06},
	pages = {2050042},
	publisher = {World Scientific},
	title = {On time consistency for mean-variance portfolio selection},
	volume = {23},
	year = 2020}

@article{widder1938stieltjes,
	author = {Widder, {\relax DV}},
	journal = {Transactions of the American Mathematical Society},
	number = {1},
	pages = {7--60},
	publisher = {JSTOR},
	title = {The stieltjes transform},
	volume = {43},
	year = 1938}

@article{Yao02082016,
	author = {Yao, Haixiang and Li, ZhongFei and Li, Xingyi},
	journal = {Quantitative Finance},
	number = {8},
	pages = {1237--1257},
	publisher = {Routledge},
	title = {The premium of dynamic trading in a discrete-time setting},
	volume = {16},
	year = 2016}

@article{Zhou2000,
	author = {Zhou, X. Y. and Li, D.},
	journal = {Applied Mathematics and Optimization},
	number = {1},
	pages = {19--33},
	title = {Continuous-time mean-variance portfolio selection: a stochastic {{LQ}} framework},
	volume = {42},
	year = 2000}
\clearpage

\appendix 

\bigskip
\begin{center}
{\Large\bf SUPPLEMENTARY MATERIAL for\\ ``On Reference-Regulated Multiperiod Mean-Variance \\ Portfolio Optimization in High Dimensions''}
\end{center} 
 
% =========================
% Appendix numbering
% =========================

% Equations
\ECEquationsNumberedThrough
% If you want equations numbered as (S1), (S2), ...
\setcounter{equation}{0}
\renewcommand{\theequation}{S\arabic{equation}}
\renewcommand{\theHequation}{ECeq.\arabic{equation}}

% Figures
\setcounter{figure}{0}
\renewcommand{\thefigure}{S\arabic{figure}}
\renewcommand{\theHfigure}{ECfig.\arabic{figure}}

% Tables
\setcounter{table}{0}
\renewcommand{\thetable}{S\arabic{table}}
\renewcommand{\theHtable}{ECtab.\arabic{table}}

% Sections
\setcounter{section}{0}
\renewcommand{\thesection}{S\arabic{section}}
\renewcommand{\theHsection}{ECsec.\arabic{section}}

% Subsections
\setcounter{subsection}{0}
\renewcommand{\thesubsection}{\thesection.\arabic{subsection}}
\renewcommand{\theHsubsection}{ECsec.\arabic{section}.\arabic{subsection}}

% Subsubsections, if needed
\setcounter{subsubsection}{0}
\renewcommand{\thesubsubsection}{\thesubsection.\arabic{subsubsection}}
\renewcommand{\theHsubsubsection}{ECsec.\arabic{section}.\arabic{subsection}.\arabic{subsubsection}}

% Theorems, if theorem/lemma/proposition/corollary share the theorem counter
\setcounter{theorem}{0}
\renewcommand{\thetheorem}{S\arabic{theorem}}
\renewcommand{\theHtheorem}{ECthm.\arabic{theorem}}

% Assumptions, if assumption has an independent counter
\setcounter{assumption}{0}
\renewcommand{\theassumption}{S\arabic{assumption}}
\renewcommand{\theHassumption}{ECasm.\arabic{assumption}}
 
\medskip 
\onehalfspacing
\fontsize{11pt}{16.5pt}\selectfont

\section{Simulation Comparisons} \label{sec:simulation}
In this section, we evaluate the performance of the RRMV portfolio relative to several benchmark strategies using simulated data calibrated from real market returns.
The simulation design is based on monthly returns of the 48-industry portfolios over the period from January 2000 to December 2024.
Let $\{\P_{-t}\}_{t=1}^{N}$ denote the $p$-dimensional vector of excess returns in the calibration sample, where $N=300$ in our data.
We set the population mean vector and covariance matrix in the simulation equal to their corresponding sample estimates, 
\[
\bmu := \frac{1}{N}\sum_{t=1}^{N} \P_{-t} ~ \text{and} ~ \bsigma := \frac{1}{N-1}\sum_{t=1}^{N} (\P_{-t}-\bmu)(\P_{-t}-\bmu)^\top \,.
\]
The risk-free rate $r_f$ is fixed at a constant equal to the in-sample average return of the one-month U.S. Treasury bill over the same period.
For each Monte Carlo replication, we generate a synthetic out-of-sample return path whose distribution is calibrated to $(\bmu,\bsigma)$.

In the baseline data-generating process (DGP), returns are assumed to be independent and identically distributed according to a multivariate normal distribution, $\P_t \sim \mathcal{N}(\bmu,\bsigma)$, for $t=1,2,\ldots$. 
To assess robustness with respect to tail behavior, we further consider a multivariate Student-$t$ model with $\nu=6$ degrees of freedom, where the DGP is given by $\P_t \sim t_\nu(\bmu,\S)$, where $\nu=6$. 
For a multivariate $t_\nu(\bmu,\S)$ distribution with $\nu>2$, the covariance matrix satisfies
$\operatorname{Cov}(\P_t)=\frac{\nu}{\nu-2}\bm{S}$.
To ensure that the heavy-tailed specification shares the same population covariance $\bsigma$ as the Gaussian benchmark, we set the scale matrix to $\bm{S} = \frac{\nu-2}{\nu}\,\bsigma$.

We conduct $1,000$ Monte Carlo replications to evaluate portfolio performance.
Within each replication, we follow the same rolling implementation scheme as in the empirical analysis.
At each rebalancing date, the most recent $n$ months of simulated returns are used to estimate the required inputs, based on which portfolio weights are computed. Performance is then evaluated over the subsequent out-of-sample period.
We consider estimation window lengths $n\in\{60,120\}$ and investment horizons $T\in\{1,3,6\}$.
Across replications, we report the average values of Risk, Sharpe ratio and Turnover defined in \eqref{eq:os_metrics}. 
Moreover, we report the $p$-values from tests of Sharpe ratio differences relative to the RRMV-EW benchmark with $\rho = 10^{-3}$.
Results under the Gaussian specification are reported in Tables~\ref{table:sim_48ind_60} and~\ref{table:sim_48ind_120}.

First, We consider the case of a short estimation window ($n=60$) in Table~\ref{table:sim_48ind_60}, where the ratio $p/n \approx 1$ corresponds to a high-dimensional setting with substantial estimation error. Under this regime, standard MV and MMV portfolios suffer from severe estimation noise and consequently deliver relatively low out-of-sample Sharpe ratios. By contrast, shrinkage-based strategies achieve substantially higher Sharpe ratios than their unregularized counterparts. RRMV portfolios further stabilize portfolio weights and yield meaningful improvements in out-of-sample Sharpe ratios.

Second, the advantage of RRMV becomes more pronounced as the investment horizon $T$ increases in Table~\ref{table:sim_48ind_60}.
For static portfolios such as EW, GMV, and their shrinkage variants, Sharpe ratios remain largely unchanged across different values of $T$.
However, RRMV portfolios could benefit from the multiperiod structure and exhibit steadily increasing Sharpe ratios as $T$ grows.
This pattern is consistent with the theoretical advantage of feedback policies discussed in Section~\ref{sse:comparison_policies}.

When the estimation window increases to $n=120$ in Table~\ref{table:sim_48ind_120}, the performance gap between regularized and unregularized portfolios becomes smaller, reflecting the reduced severity of high-dimensional estimation error as $p/n$ declines. 
With more observations available, MV and MMV strategies achieve improved Sharpe ratios relative to the high-dimensional case. 
However, RRMV continues to deliver the highest Sharpe ratios, particularly at longer horizons. 
Moreover, the improvement in unregularized multiperiod portfolios is accompanied by a substantial increase in turnover. 
For example, at $T=6$, turnover for MMV rises sharply, whereas RRMV maintains significantly lower trading intensity. 
Consequently, once transaction costs are considered, the apparent convergence in Sharpe ratios would likely favor RRMV more strongly.

Figures~\ref{fig_simulate_48IND_n60_1_3_6} and~\ref{fig_simulate_48IND_n120_1_3_6} plot the out-of-sample Sharpe ratio as a function of the regularization parameter $\rho$ for different investment horizons $T$ and estimation window lengths $n$. When $n=60$ (Figure~\ref{fig_simulate_48IND_n60_1_3_6}), estimation error is substantial, and all RRMV variants outperform the unregularized MMV benchmark for certain positive values of $\rho$. In contrast, when $n=120$ (Figure~\ref{fig_simulate_48IND_n120_1_3_6}), estimation noise is milder and the Sharpe ratio at $\rho=0$ is already relatively high, so the incremental benefit of regularization is smaller.
% This evidence suggests that the effectiveness of regularization depends jointly on the level of estimation noise and the choice of reference portfolio.

Overall, the simulation results indicate that the value of regularization depends on both the estimation noise level and the investment horizon. RRMV delivers the largest improvements in Sharpe ratios in high-noise environments and at longer horizons, while maintaining substantially lower turnover than competing multiperiod strategies.
These findings suggest that RRMV strikes a favorable balance between risk-adjusted performance and trading stability, particularly when data are limited.

\begin{table}[htbp]
    \centering
    \caption{Summary of out-of-sample performance of the portfolios under comparison $(n=60)$.}
    \label{table:sim_48ind_60}
    \scriptsize
    \renewcommand{\arraystretch}{1.15}
    \setlength{\tabcolsep}{2pt}
    \begin{tabularx}{\textwidth}{l *{12}{>{\centering\arraybackslash}X}}
        \hline
        & \multicolumn{4}{c}{$T=1$}
        & \multicolumn{4}{c}{$T=3$}
        & \multicolumn{4}{c}{$T=6$} \\
        \cline{2-5} \cline{6-9} \cline{10-13}
        Portfolio
        & Risk & Sharpe & Turnover & $p$-value
        & Risk & Sharpe & Turnover & $p$-value
        & Risk & Sharpe & Turnover & $p$-value \\
        \hline

        \multicolumn{13}{l}{\textbf{Benchmarks}} \\
        KZ      & 0.077 & 0.048 & 5.429 & 0.000 & 0.139 & 0.044 & 5.429 & 0.000 & 0.210 & 0.044 & 5.426 & 0.000 \\
        EW      & 0.049 & 0.157 & 0.000 & 0.000 & 0.085 & 0.156 & 0.000 & 0.000 & 0.123 & 0.156 & 0.000 & 0.000 \\
        MV      & 0.021 & 0.102 & 1.329 & 0.000 & 0.036 & 0.103 & 1.329 & 0.000 & 0.051 & 0.105 & 1.329 & 0.000 \\
        MV-SH   & \textbf{0.012} & 0.192 & 0.313 & 0.000 & \textbf{0.020} & 0.194 & 0.313 & 0.000 & \textbf{0.029} & 0.196 & 0.313 & 0.000 \\
        GMV     & 0.062 & 0.086 & 3.583 & 0.000 & 0.108 & 0.085 & 3.583 & 0.000 & 0.155 & 0.085 & 3.583 & 0.000 \\
        GMV-SH  & 0.034 & 0.173 & 0.532 & 0.000 & 0.059 & 0.174 & 0.532 & 0.000 & 0.084 & 0.174 & 0.532 & 0.000 \\
        MMV     & -- & -- & -- & -- & 0.289 & 0.032 & 31.275 & 0.000 & 8.554 & 0.004 & 556.045 & 0.000 \\
        MMV-SH  & -- & -- & -- & -- & 0.038 & 0.177 & 2.037 & 0.000 & 0.108 & 0.182 & 7.203 & 0.000 \\

        \hline
        \multicolumn{13}{l}{$\rho=10^{-2}$} \\
        RRMV-EW        & 0.015 & 0.211 & 0.162 & 0.157 & 0.030 & \textbf{0.216} & 0.351 & 0.216 & \textbf{0.048} & \textbf{0.228} & 0.418 & 0.000 \\
        RRMV-GMV-SH    & 0.030 & 0.174 & 0.569 & 0.000 & 0.054 & 0.182 & 0.845 & 0.000 & 0.077 & 0.196 & 0.897 & 0.000 \\
        RRMV-$\ell_2$  & \textbf{0.014} & 0.189 & 0.176 & 0.000 & \textbf{0.027} & 0.195 & 0.337 & 0.000 & \textbf{0.044} & 0.206 & 0.401 & 0.000 \\

        \hline
        \multicolumn{13}{l}{$\rho=10^{-3}$} \\
        RRMV-EW        & \textbf{0.011} & \textbf{0.213} & 0.244 & -- & \textbf{0.027} & \textbf{0.218} & 0.830 & -- & \textbf{0.050} & \textbf{0.238} & 0.985 & -- \\
        RRMV-GMV-SH    & 0.023 & 0.174 & 0.573 & 0.000 & 0.047 & 0.183 & 1.606 & 0.000 & 0.076 & 0.204 & 1.739 & 0.000 \\
        RRMV-$\ell_2$  & \textbf{0.011} & \textbf{0.203} & 0.248 & 0.000 & \textbf{0.027} & 0.208 & 0.820 & 0.000 & \textbf{0.050} & 0.225 & 0.974 & 0.000 \\
        \hline
    \end{tabularx}

    \captionsetup{justification=raggedright,singlelinecheck=false}
    \caption*{\footnotesize
    Notes. The data-generating process is calibrated using the sample mean vector and covariance matrix estimated from the 48-industry portfolio returns. Reported $p$-values are based on paired tests of Sharpe-ratio differences relative to the RRMV-EW portfolio with $\rho=10^{-3}$; the reference row therefore has no corresponding $p$-value. Static portfolios follow a buy-and-hold strategy. MMV and MMV-SH are identical to MV and MV-SH, respectively, when $T=1$. The two best Sharpe values and four best Risk values are highlighted.}
\end{table}

\begin{table}[htbp]
    \centering
    \caption{Summary of out-of-sample performance of the portfolios under comparison $(n=120)$.}
    \label{table:sim_48ind_120}
    \scriptsize
    \renewcommand{\arraystretch}{1.1}
    \setlength{\tabcolsep}{2pt}
    \begin{tabularx}{\textwidth}{l *{12}{>{\centering\arraybackslash}X}}
        \hline
        & \multicolumn{4}{c}{$T=1$}
        & \multicolumn{4}{c}{$T=3$}
        & \multicolumn{4}{c}{$T=6$} \\
        \cline{2-5} \cline{6-9} \cline{10-13}
        Portfolio
        & Risk & Sharpe & Turnover & $p$-value
        & Risk & Sharpe & Turnover & $p$-value
        & Risk & Sharpe & Turnover & $p$-value \\
        \hline

        \multicolumn{13}{l}{\textbf{Benchmarks}} \\
        KZ      & 0.094 & 0.203 & 2.508 & 0.000 & 0.172 & 0.196 & 2.508 & 0.000 & 0.264 & 0.189 & 2.507 & 0.000 \\
        EW      & 0.049 & 0.158 & 0.000 & 0.000 & 0.086 & 0.157 & 0.000 & 0.000 & 0.123 & 0.157 & 0.000 & 0.000 \\
        MV      & \textbf{0.015} & 0.236 & 0.370 & 0.000 & \textbf{0.026} & 0.236 & 0.370 & 0.000 & \textbf{0.037} & 0.237 & 0.370 & 0.000 \\
        MV-SH   & \textbf{0.013} & 0.263 & 0.258 & 0.000 & \textbf{0.023} & 0.263 & 0.258 & 0.000 & \textbf{0.033} & 0.264 & 0.258 & 0.000 \\
        GMV     & 0.034 & 0.159 & 0.593 & 0.000 & 0.060 & 0.159 & 0.593 & 0.000 & 0.086 & 0.158 & 0.593 & 0.000 \\
        GMV-SH  & 0.031 & 0.183 & 0.329 & 0.000 & 0.054 & 0.183 & 0.329 & 0.000 & 0.078 & 0.182 & 0.329 & 0.000 \\
        MMV     & -- & -- & -- & -- & 0.045 & 0.203 & 3.435 & 0.000 & 0.126 & 0.195 & 15.935 & 0.000 \\
        MMV-SH  & -- & -- & -- & -- & 0.033 & 0.260 & 1.835 & 0.000 & 0.068 & 0.301 & 3.842 & 0.000 \\

        \hline
        \multicolumn{13}{l}{$\rho=10^{-2}$} \\
        RRMV-EW        & 0.019 & 0.244 & 0.125 & 0.000 & 0.034 & 0.250 & 0.329 & 0.000 & 0.051 & 0.266 & 0.390 & 0.000 \\
        RRMV-GMV-SH    & 0.029 & 0.204 & 0.356 & 0.000 & 0.051 & 0.213 & 0.610 & 0.000 & 0.072 & 0.228 & 0.664 & 0.000 \\
        RRMV-$\ell_2$  & 0.017 & 0.236 & 0.136 & 0.000 & 0.031 & 0.243 & 0.317 & 0.000 & 0.048 & 0.258 & 0.374 & 0.000 \\

        \hline
        \multicolumn{13}{l}{$\rho=10^{-3}$} \\
        RRMV-EW        & \textbf{0.013} & \textbf{0.277} & 0.186 & -- & \textbf{0.027} & \textbf{0.291} & 0.758 & -- & \textbf{0.044} & \textbf{0.328} & 0.898 & -- \\
        RRMV-GMV-SH    & 0.020 & 0.226 & 0.340 & 0.000 & 0.039 & 0.241 & 1.163 & 0.000 & 0.059 & 0.272 & 1.329 & 0.000 \\
        RRMV-$\ell_2$  & \textbf{0.013} & \textbf{0.273} & 0.188 & 0.000 & \textbf{0.027} & \textbf{0.286} & 0.752 & 0.000 & \textbf{0.044} & \textbf{0.321} & 0.889 & 0.000 \\
        \hline
    \end{tabularx}

    \captionsetup{justification=raggedright,singlelinecheck=false}
    \caption*{\footnotesize
    Notes. The data-generating process is calibrated using the sample mean vector and covariance matrix estimated from the 48-industry portfolio returns. Reported $p$-values are based on paired tests of Sharpe-ratio differences relative to the RRMV-EW portfolio with $\rho=10^{-3}$; the reference row therefore has no corresponding $p$-value. Static portfolios follow a buy-and-hold strategy. MMV and MMV-SH are identical to MV and MV-SH, respectively, when $T=1$. The two best Sharpe values and four best Risk values are highlighted.}
\end{table}

\begin{figure}[htbp]
    \centering
    \begin{subfigure}[t]{0.32\textwidth}
        \centering
        \includegraphics[width=\textwidth]{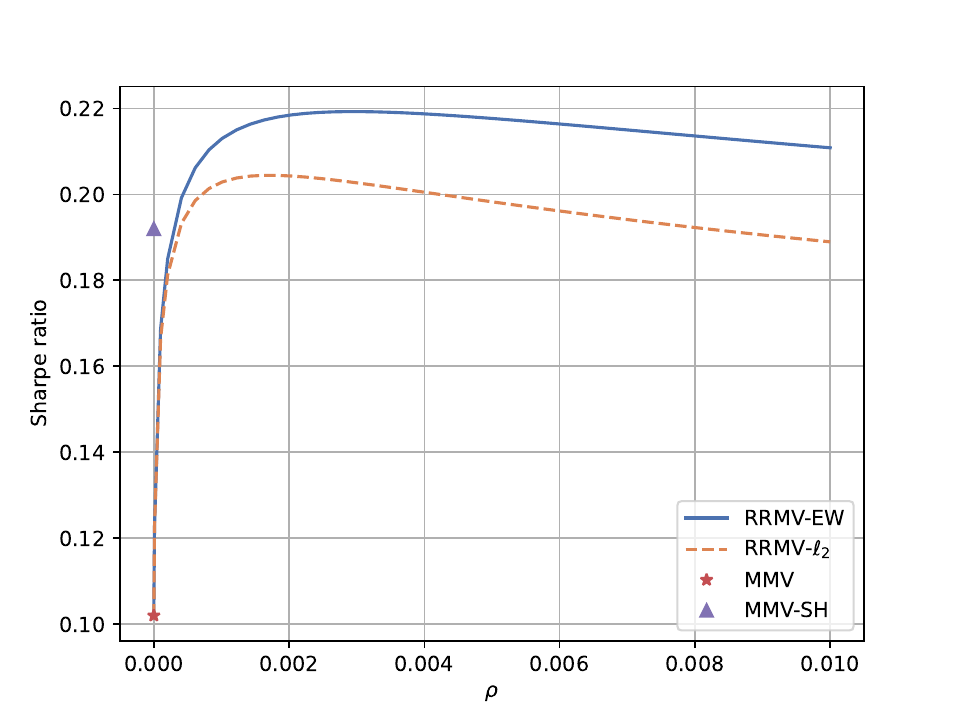}
        \caption{$T=1$}
        \label{fig:SR_multi_48Ind_1000_60_1}
    \end{subfigure}
    \hfill
    \begin{subfigure}[t]{0.32\textwidth}
        \centering
        \includegraphics[width=\textwidth]{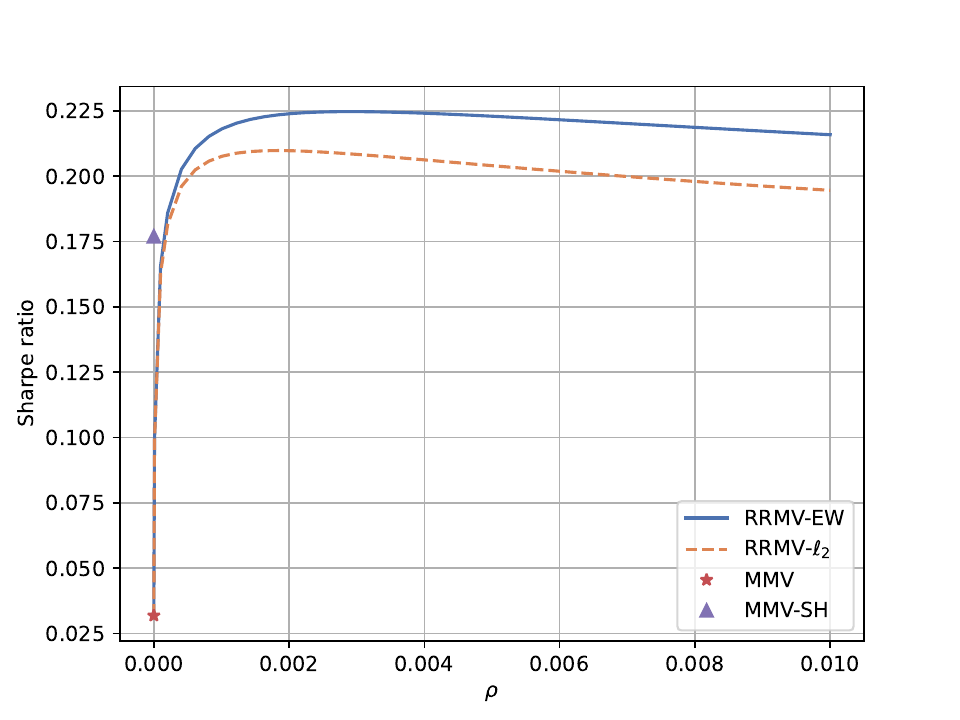}
        \caption{$T=3$}
        \label{fig:SR_multi_48Ind_1000_60_3}
    \end{subfigure}
    \hfill
    \begin{subfigure}[t]{0.32\textwidth}
        \centering
        \includegraphics[width=\textwidth]{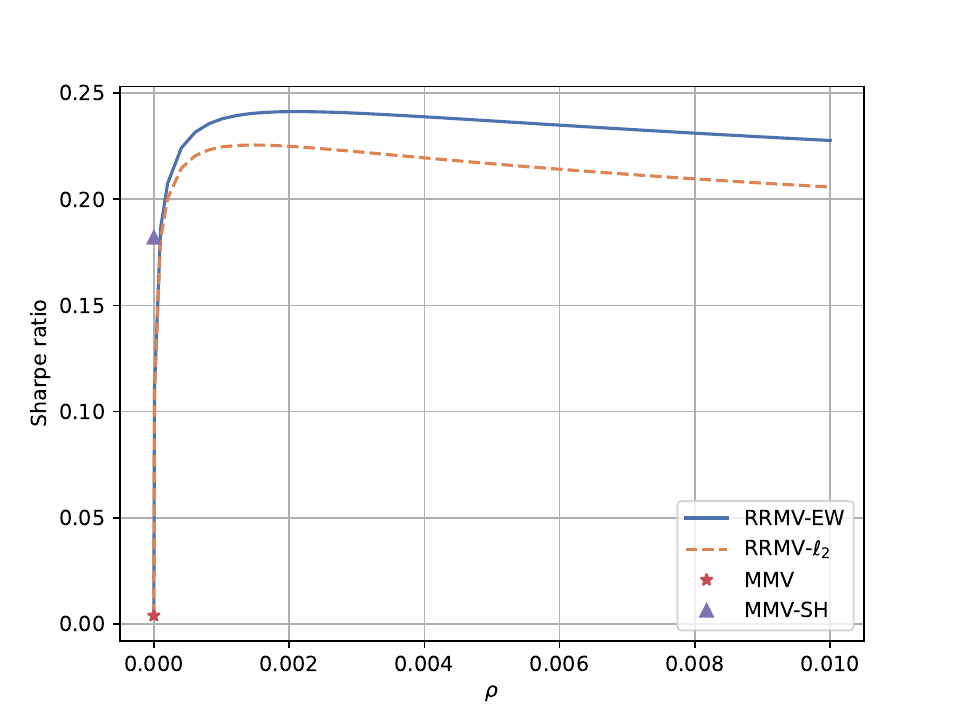}
        \caption{$T=6$}
        \label{fig:SR_multi_48Ind_1000_60_6}
    \end{subfigure}

    \caption{Out-of-sample Sharpe ratio versus the regularization parameter $\rho$ $(n=60)$.}
    \label{fig_simulate_48IND_n60_1_3_6}

    \captionsetup{justification=raggedright,singlelinecheck=false}
    \caption*{\footnotesize
    Notes. Panels (a)--(c) correspond to investment horizons $T=1$, $T=3$, and $T=6$, respectively.}
\end{figure}

\begin{figure}[htbp]
    \centering
    \begin{subfigure}[t]{0.32\textwidth}
        \centering
        \includegraphics[width=\textwidth]{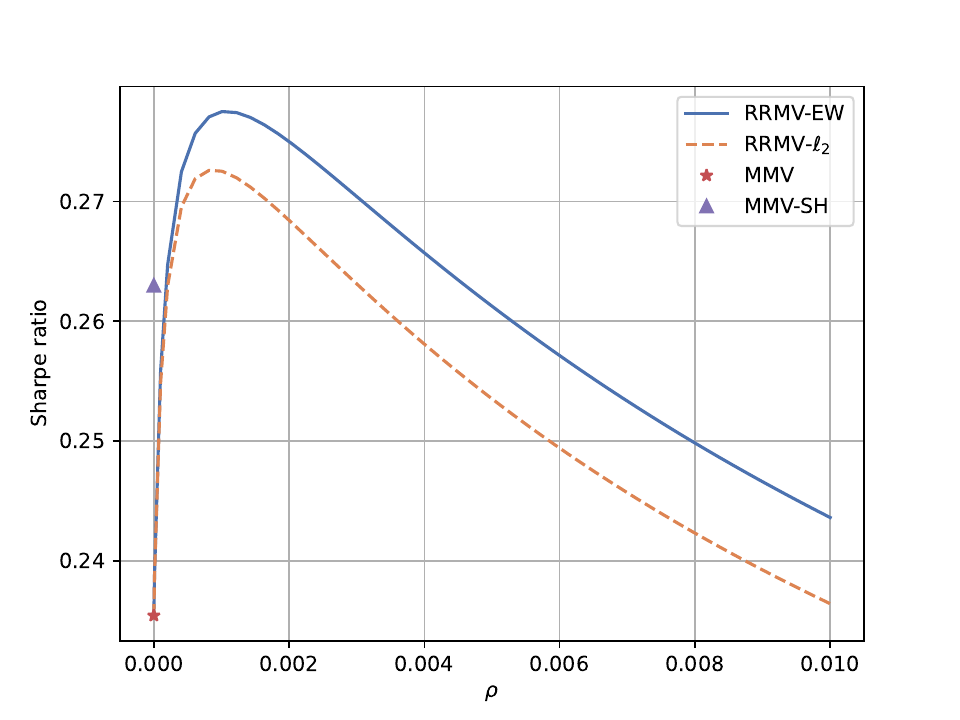}
        \caption{$T=1$}
        \label{fig:SR_multi_48Ind_1000_120_1}
    \end{subfigure}
    \hfill
    \begin{subfigure}[t]{0.32\textwidth}
        \centering
        \includegraphics[width=\textwidth]{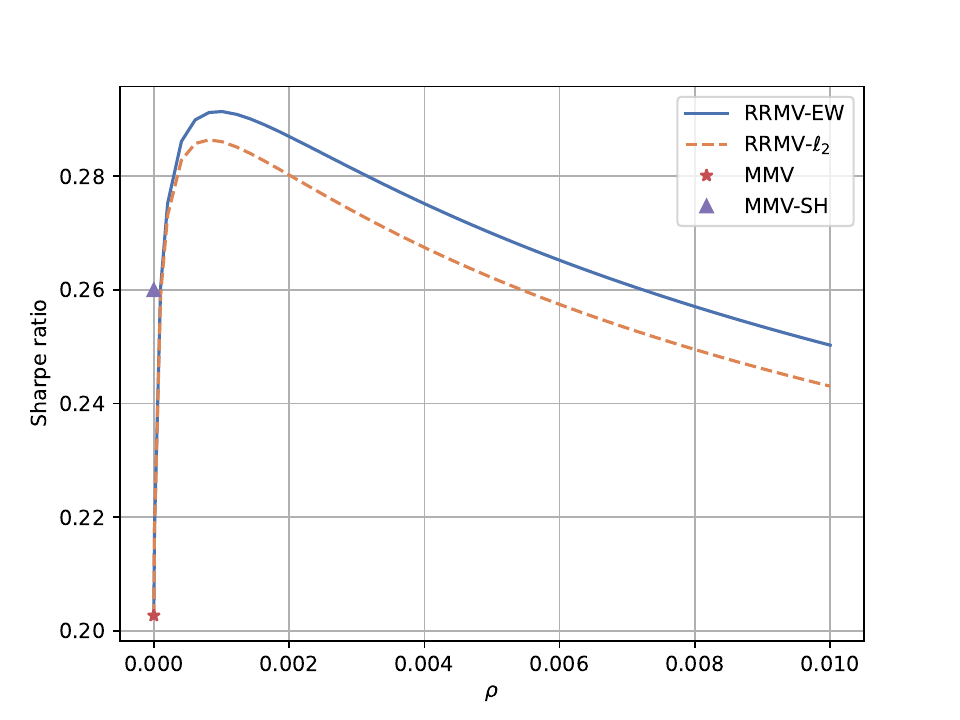}
        \caption{$T=3$}
        \label{fig:SR_multi_48Ind_1000_120_3}
    \end{subfigure}
    \hfill
    \begin{subfigure}[t]{0.32\textwidth}
        \centering
        \includegraphics[width=\textwidth]{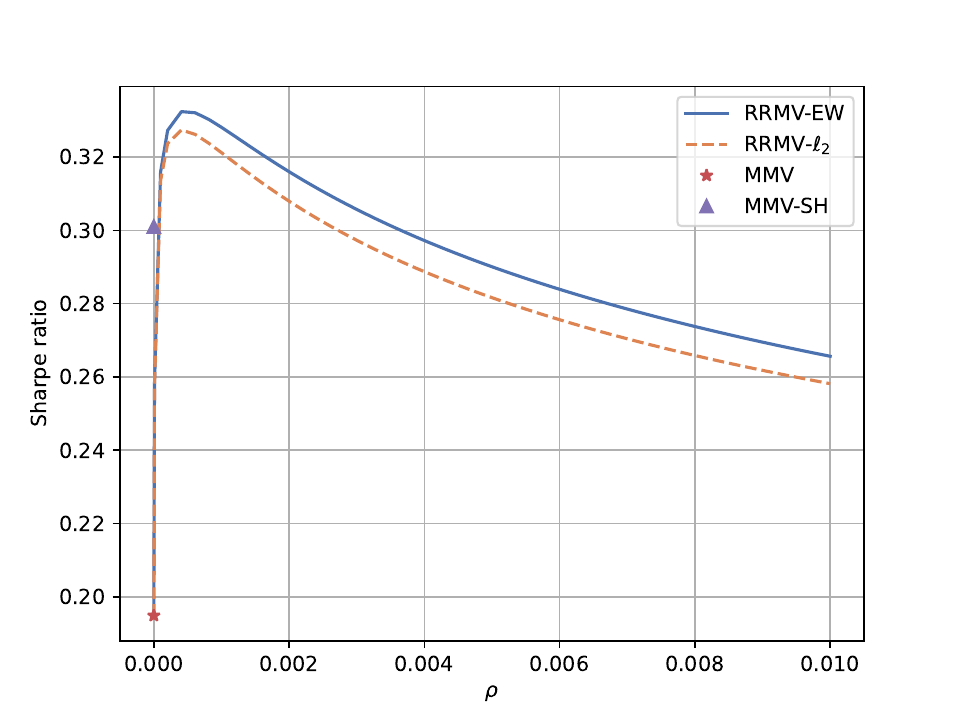}
        \caption{$T=6$}
        \label{fig:SR_multi_48Ind_1000_120_6}
    \end{subfigure}

    \caption{Out-of-sample Sharpe ratio versus the regularization parameter $\rho$ $(n=120)$.}
    \label{fig_simulate_48IND_n120_1_3_6}

    \captionsetup{justification=raggedright,singlelinecheck=false}
    \caption*{\footnotesize
    Notes. Panels (a)--(c) correspond to investment horizons $T=1$, $T=3$, and $T=6$, respectively.}
\end{figure}
 
\newpage 

Finally, we give the simulation analysis results for data with heavy-tails. 
In particular, we use multivariate t distribution in the simulation with 6 degree of freedom. 
Tables \ref{table:sim_48ind_t_60} and Table \ref{table:sim_48ind_t_120} report results under multivariate $t$ returns to examine robustness to heavy-tailed distributions. 
The overall ranking of portfolio rules remains consistent with the case assuming normal distribution. 
RRMV portfolios continue to outperform unregularized  strategies across different investment horizons.

\begin{table}[htbp]
    \centering
    \caption{Simulation performance on 48-industry portfolios under multivariate $t$ returns when $n=60$.}
    \label{table:sim_48ind_t_60}
    \scriptsize
    \renewcommand{\arraystretch}{1.1}
    \setlength{\tabcolsep}{2pt}
    \begin{tabularx}{\textwidth}{l *{12}{>{\centering\arraybackslash}X}}
        \hline
        & \multicolumn{4}{c}{$T=1$}
        & \multicolumn{4}{c}{$T=3$}
        & \multicolumn{4}{c}{$T=6$} \\
        \cline{2-5} \cline{6-9} \cline{10-13}
        Portfolio
        & Risk & Sharpe & Turnover & $p$-value
        & Risk & Sharpe & Turnover & $p$-value
        & Risk & Sharpe & Turnover & $p$-value \\
        \hline

        \multicolumn{13}{l}{\textbf{Benchmarks}} \\
        KZ      & 0.094 & 0.203 & 2.508 & 0.000 & 0.172 & 0.196 & 2.508 & 0.000 & 0.264 & 0.189 & 2.507 & 0.000 \\
        EW      & 0.049 & 0.158 & 0.000 & 0.000 & 0.086 & 0.157 & 0.000 & 0.000 & 0.123 & 0.157 & 0.000 & 0.000 \\
        MV      & \textbf{0.015} & 0.236 & 0.370 & 0.000 & \textbf{0.026} & 0.236 & 0.370 & 0.000 & \textbf{0.037} & 0.237 & 0.370 & 0.000 \\
        MV-SH   & \textbf{0.013} & 0.263 & 0.258 & 0.000 & \textbf{0.023} & 0.263 & 0.258 & 0.000 & \textbf{0.033} & 0.264 & 0.258 & 0.000 \\
        GMV     & 0.034 & 0.159 & 0.593 & 0.000 & 0.060 & 0.159 & 0.593 & 0.000 & 0.086 & 0.158 & 0.593 & 0.000 \\
        GMV-SH  & 0.031 & 0.183 & 0.329 & 0.000 & 0.054 & 0.183 & 0.329 & 0.000 & 0.078 & 0.182 & 0.329 & 0.000 \\
        MMV     & -- & -- & -- & -- & 0.045 & 0.203 & 3.435 & 0.000 & 0.126 & 0.195 & 15.935 & 0.000 \\
        MMV-SH  & -- & -- & -- & -- & 0.033 & 0.260 & 1.835 & 0.000 & 0.068 & 0.301 & 3.842 & 0.000 \\

        \hline
        \multicolumn{13}{l}{$\rho=10^{-2}$} \\
        RRMV-EW        & 0.019 & 0.244 & 0.125 & 0.000 & 0.034 & 0.250 & 0.329 & 0.000 & 0.051 & 0.266 & 0.390 & 0.000 \\
        RRMV-GMV-SH    & 0.029 & 0.204 & 0.356 & 0.000 & 0.051 & 0.213 & 0.610 & 0.000 & 0.072 & 0.228 & 0.664 & 0.000 \\
        RRMV-$\ell_2$  & 0.017 & 0.236 & 0.136 & 0.000 & 0.031 & 0.243 & 0.317 & 0.000 & 0.048 & 0.258 & 0.374 & 0.000 \\

        \hline
        \multicolumn{13}{l}{$\rho=10^{-3}$} \\
        RRMV-EW        & \textbf{0.013} & \textbf{0.277} & 0.186 & -- & \textbf{0.027} & \textbf{0.291} & 0.758 & -- & \textbf{0.044} & \textbf{0.328} & 0.898 & -- \\
        RRMV-GMV-SH    & 0.020 & 0.226 & 0.340 & 0.000 & 0.039 & 0.241 & 1.163 & 0.000 & 0.059 & 0.272 & 1.329 & 0.000 \\
        RRMV-$\ell_2$  & \textbf{0.013} & \textbf{0.273} & 0.188 & 0.000 & \textbf{0.027} & \textbf{0.286} & 0.752 & 0.000 & \textbf{0.044} & \textbf{0.321} & 0.889 & 0.000 \\
        \hline
    \end{tabularx}

    \captionsetup{justification=raggedright,singlelinecheck=false}
    \caption*{\footnotesize
    Notes. The data-generating process follows a multivariate $t$ distribution using the sample mean vector and covariance matrix estimated from the 48-industry portfolio returns. Reported $p$-values are based on paired tests of Sharpe-ratio differences relative to the RRMV-EW portfolio with $\rho=10^{-3}$; the reference row therefore has no corresponding $p$-value. Static portfolios follow a buy-and-hold strategy. MMV and MMV-SH are identical to MV and MV-SH, respectively, when $T=1$. The two best Sharpe values and four best Risk values are highlighted.}
\end{table}

\begin{table}[htbp]
    \centering
    \caption{Simulation performance on 48-industry portfolios under multivariate $t$ returns when $n=120$.}
    \label{table:sim_48ind_t_120}
    \scriptsize
    \renewcommand{\arraystretch}{1.1}
    \setlength{\tabcolsep}{2pt}
    \begin{tabularx}{\textwidth}{l *{12}{>{\centering\arraybackslash}X}}
        \hline
        & \multicolumn{4}{c}{$T=1$}
        & \multicolumn{4}{c}{$T=3$}
        & \multicolumn{4}{c}{$T=6$} \\
        \cline{2-5} \cline{6-9} \cline{10-13}
        Portfolio
        & Risk & Sharpe & Turnover & $p$-value
        & Risk & Sharpe & Turnover & $p$-value
        & Risk & Sharpe & Turnover & $p$-value \\
        \hline

        \multicolumn{13}{l}{\textbf{Benchmarks}} \\
        KZ      & 0.120 & 0.215 & 3.163 & 0.000 & 0.226 & 0.205 & 3.163 & 0.000 & 0.359 & 0.195 & 3.162 & 0.000 \\
        EW      & 0.049 & 0.151 & 0.000 & 0.000 & 0.085 & 0.151 & 0.000 & 0.000 & 0.122 & 0.150 & 0.000 & 0.000 \\
        MV      & \textbf{0.016} & 0.244 & 0.405 & 0.000 & \textbf{0.028} & 0.244 & 0.405 & 0.000 & \textbf{0.040} & 0.245 & 0.405 & 0.000 \\
        MV-SH   & \textbf{0.014} & 0.280 & 0.232 & 0.000 & \textbf{0.024} & 0.280 & 0.232 & 0.000 & \textbf{0.035} & 0.282 & 0.232 & 0.000 \\
        GMV     & 0.036 & 0.154 & 0.655 & 0.000 & 0.063 & 0.153 & 0.655 & 0.000 & 0.089 & 0.154 & 0.655 & 0.000 \\
        GMV-SH  & 0.031 & 0.187 & 0.270 & 0.000 & 0.054 & 0.188 & 0.270 & 0.000 & 0.078 & 0.189 & 0.270 & 0.000 \\
        MMV     & -- & -- & -- & -- & 0.052 & 0.197 & 4.146 & 0.000 & 0.153 & 0.172 & 19.631 & 0.000 \\
        MMV-SH  & -- & -- & -- & -- & 0.033 & 0.282 & 1.585 & 0.000 & 0.062 & 0.326 & 4.037 & 0.000 \\

        \hline
        \multicolumn{13}{l}{$\rho=10^{-2}$} \\
        RRMV-EW        & 0.019 & 0.244 & 0.122 & 0.000 & 0.035 & 0.252 & 0.325 & 0.000 & 0.052 & 0.269 & 0.385 & 0.000 \\
        RRMV-GMV-SH    & 0.029 & 0.211 & 0.299 & 0.000 & 0.050 & 0.222 & 0.554 & 0.000 & 0.071 & 0.238 & 0.608 & 0.000 \\
        RRMV-$\ell_2$  & 0.018 & 0.239 & 0.133 & 0.000 & 0.032 & 0.248 & 0.313 & 0.000 & 0.048 & 0.264 & 0.368 & 0.000 \\

        \hline
        \multicolumn{13}{l}{$\rho=10^{-3}$} \\
        RRMV-EW        & \textbf{0.014} & \textbf{0.287} & 0.183 & -- & \textbf{0.028} & \textbf{0.302} & 0.771 & -- & \textbf{0.045} & \textbf{0.339} & 0.903 & -- \\
        RRMV-GMV-SH    & 0.021 & 0.239 & 0.299 & 0.000 & 0.039 & 0.256 & 1.126 & 0.000 & 0.059 & 0.288 & 1.299 & 0.000 \\
        RRMV-$\ell_2$  & \textbf{0.014} & \textbf{0.283} & 0.186 & 0.000 & \textbf{0.028} & \textbf{0.298} & 0.763 & 0.000 & \textbf{0.044} & \textbf{0.334} & 0.892 & 0.000 \\
        \hline
    \end{tabularx}

    \captionsetup{justification=raggedright,singlelinecheck=false}
    \caption*{\footnotesize
    Notes. The data-generating process follows a multivariate $t$ distribution using the sample mean vector and covariance matrix estimated from the 48-industry portfolio returns. Reported $p$-values are based on paired tests of Sharpe-ratio differences relative to the RRMV-EW portfolio with $\rho=10^{-3}$; the reference row therefore has no corresponding $p$-value. Static portfolios follow a buy-and-hold strategy. MMV and MMV-SH are identical to MV and MV-SH, respectively, when $T=1$. The two best Sharpe values and four best Risk values are highlighted.}
\end{table}

\newpage

\section{Proofs for Section 2} 
\subsection{Proof of Lemma~\ref{lemma:param_a_nonneg}}
\begin{proof}

Recall that we define the parameters recursively, for $k=T-1,\ldots,0$
\begin{equation*}
\begin{aligned}
    \D_k &= a_{k+1} \E[\P_k \P_k^\top] + \Q_k,\\
    a_k  &= a_{k+1} r_k^2  + \w_k^\top \Q_k \w_k - a_{k+1}^2 r_k^2 \bmu_k^\top \D_k^{-1} \bmu_k - \w_k^\top \Q_k \D_k^{-1} \Q_k \w_k + 2 r_k a_{k+1} \w_k^\top \Q_k \D_k^{-1} \bmu_k,
\end{aligned} 
\end{equation*}
where $a_T=1$. 
By induction, we can easily verify that $a_k >0$ when $k=T$. 
Then, we assume that $a_i >0$ for $i=k+1,\ldots,T$ and we would like to show that $a_k>0$. 
By definition, we have
\begin{equation*}
\begin{aligned}
    a_k 
    &= a_{k+1} r_k^2 \left( 1 - \bmu_k^{\top} \left( \bsigma_k + \bmu_k \bmu_k^\top + a_{k+1}^{-1} \Q_k  \right)^{-1} \bmu_k \right) + \w_k^\top \Q_k \left( \Q_k^{-1} - \D_k^{-1} \right) \Q_k \w_k 
    \\
    &~ + 2 r_k \w_k^\top \Q_k \left( \bsigma_k + \bmu_k \bmu_k^\top + a_{k+1}^{-1} \Q_k  \right)^{-1} \bmu_k
\end{aligned} 
\end{equation*}
We can further show that 
\begin{equation*}
\begin{aligned}
    \bmu_k^{\top} \left( \bsigma_k + \bmu_k \bmu_k^\top + a_{k+1}^{-1} \Q_k  \right)^{-1} \bmu_k 
    % &= \bmu_k^{\top} \left( \left( \bsigma_k + a_{k+1}^{-1} \Q_k  \right)^{-1} - \frac{\left( \bsigma_k + a_{k+1}^{-1} \Q_k  \right)^{-1} \bmu_k \bmu_k^\top \left( \bsigma_k + a_{k+1}^{-1} \Q_k  \right)^{-1}}{1+\bmu_k^\top \left( \bsigma_k + a_{k+1}^{-1} \Q_k  \right)^{-1} \bmu_k} \right) \bmu_k 
    % \\
    &= \frac{\bmu_k^\top \left( \bsigma_k + a_{k+1}^{-1} \Q_k  \right)^{-1} \bmu_k}{1+\bmu_k^\top \left( \bsigma_k + a_{k+1}^{-1} \Q_k  \right)^{-1} \bmu_k}
    \\
    \w_k^\top \Q_k \left( \bsigma_k + \bmu_k \bmu_k^\top + a_{k+1}^{-1} \Q_k  \right)^{-1} \bmu_k
    &= 
    \frac{\w_k^\top \Q_k \left( \bsigma_k + a_{k+1}^{-1} \Q_k  \right)^{-1} \bmu_k}{1+\bmu_k^\top \left( \bsigma_k + a_{k+1}^{-1} \Q_k  \right)^{-1} \bmu_k}
\end{aligned} 
\end{equation*}
We can show that 
\begin{align*}
    a_k
    &= 
    \frac{a_{k+1} r_k^2 + 2 r_k \w_k^\top \Q_k \left( \bsigma_k + a_{k+1}^{-1} \Q_k  \right)^{-1} \bmu_k + a_{k+1}^{-1} (\w_k^\top \Q_k \left( \bsigma_k + a_{k+1}^{-1} \Q_k  \right)^{-1} \bmu_k)^2}{1+\bmu_k^\top \left( \bsigma_k + a_{k+1}^{-1} \Q_k  \right)^{-1} \bmu_k} 
    \\
    & ~~~~ + \w_k^\top \Q_k \w_k - a_{k+1}^{-1} \w_k^\top \Q_k \left( \bsigma_k + a_{k+1}^{-1} \Q_k  \right)^{-1} \Q_k \w_k
    \\
    &= 
    a_{k+1}^{-1} \frac{(a_{k+1} r_k + \w_k^\top \Q_k \left( \bsigma_k + a_{k+1}^{-1} \Q_k  \right)^{-1} \bmu_k)^2 }{1+\bmu_k^\top \left( \bsigma_k + a_{k+1}^{-1} \Q_k  \right)^{-1} \bmu_k} + \w_k^\top \Q_k \left( \Q_k^{-1} - (a_{k+1} \bsigma_k + \Q_k)^{-1} \right) \Q_k \w_k
\end{align*}
And we can apply the Woodbury identity to the second term, and show that $\Q_k^{-1} - (a_{k+1} \bsigma_k + \Q_k)^{-1}$ is positive definite. 
Therefore, we can show that $a_k > 0$.  
\end{proof}

\subsection{Proof of Theorem~\ref{theorem:opt_aux_prob}}
\begin{proof}
    Define the value function of the auxiliary problem $\cA(\omega, \lambda)$ as
    \begin{equation*}
        \begin{aligned}
            J_t(X_t) = & \min_{ \u_{\tau, \tau \geq t}} \E_t [\omega X_T^2 - \lambda X_T] + \omega \sum_{i=t}^{T-1} E_t[\| \u_i - X_i \w_i \|_{\Q}^2] \\
        \end{aligned}
    \end{equation*}
    By the smooth property of the conditional expectation, we have that
    \begin{equation*}
        \begin{aligned}
            J_t(X_t) 
            = & \min_{\e^{\top} \u_t = X_t} \{ \min_{ \u_{\tau, \tau \geq t+1}} \E_t [\omega X_T^2 - \lambda X_T] 
            + \omega \sum_{i=t}^{T-1}  E_t[\| \u_i - X_i \w_i \|_{\Q}^2] \}\\
            = & \min_{\e^{\top} \u_t = X_t}E_t[J_{t+1}(X_{t+1})] + \omega E_t[ \| \u_t - X_t \w_t \|_{\Q}^2 ]\\
            = & \min_{\e^{\top} \u_t = X_t}E_t[J_{t+1}(X_{t+1})] + \omega \| \u_t - X_t \w_t \|_{\Q}^2
        \end{aligned}
    \end{equation*}
    We claim that the value function is in a quadratic form of 
    \begin{equation*}
        \begin{aligned}
            J_k(X_k) &= \omega a_k X_k^2 - \lambda b_k X_k + \frac{\lambda^2}{4 \omega}c_k \\
        \end{aligned}
    \end{equation*}
    where the parameters $a_k$, $b_k$ and $c_k$ are defined in Eq. (\ref{equation:rrmv_rf_para}). And the claim can be proved by the 
    induction. At time $T$, $J_T(X_T) = \omega X_T^2 - \lambda X_T$ 
    is true since we have $a_T=1$, $b_T=1$ and $c_T=0$ at time $T$.
    Assume that the claim is true at period $k+1$, which is 
    $J_{k+1}(X_{k+1}) = \omega a_{k+1}X_{k+1}^2 - \lambda b_{k+1} X_{k+1} + \frac{\lambda^2}{4\omega} c_{k+1}$. 
    Recall that we have defined the following parameters:
    \begin{align*}
        \begin{dcases}
        \D_k  = a_{k+1} \E[\P_k \P_k^\top] + \Q_k,\\ 
         a_k  = a_{k+1} r_k^2  + \w_k^\top \Q_k \w_k -  a_{k+1}^2 r_k^2 \bmu_k^\top \D_k^{-1} \bmu_k - \w_k^\top \Q_k \D_k^{-1} \Q_k \w_k + 2 r_k a_{k+1} \w_k^\top \Q_k \D_k^{-1} \bmu_k,\\
         b_k   = r_k b_{k+1} - r_k b_{k+1} a_{k+1}  \bmu_k^\top \D_k^{-1} \bmu_k + b_{k+1} \w_k^\top \Q_k \D_k^{-1} \bmu_k,\\
         c_k = c_{k+1} - b_{k+1}^2 \bmu_k^\top \D_k^{-1} \bmu_k 
        \end{dcases} 
    \end{align*}
    We can show that 
    \begin{equation*}
        \begin{aligned}
            J_k(X_k) 
            & = 
            \min_{X_{k+1}= r_k X_k + \P_k^\top \u_k} \E_k \left[ J_{k+1}(X_{k+1}) \right] + \omega \| \u_k - X_k \w_k \|_{\Q_k}^2
            \\
            & = 
            \min_{X_{k+1}= r_k X_k + \P_k^\top \u_k} \E_k \left[ \omega a_{k+1}X_{k+1}^2 - \lambda b_{k+1} X_{k+1} + \frac{\lambda^2}{4\omega} c_{k+1} \right] + \omega \| \u_k-X_k \w_k \|_{\Q_k}^2
            \\
            & = 
            \min ~~~ \omega a_{k+1} r_k^2 X_k^2 + \omega a_{k+1} \u_k^\top \E [\P_k \P_k^\top] \u_k + 2 \omega r_k X_k a_{k+1} \bmu_k^\top \u_k - \lambda r_k X_k b_{k+1} - \lambda b_{k+1} \bmu_k^\top \u_k 
            \\
            & ~~~~~~~~~~~~~ + \frac{\lambda^2}{4\omega } c_{k+1}  + \omega \u_k^\top \Q_k \u_k - 2 \omega X_k \w_k^\top \Q_k \u_k + \omega X_k^2 \w_k^\top \Q_k \w_k 
            \\
            & = 
            \min ~~~ \omega a_{k+1} r_k^2 X_k^2 + \omega \u_k^\top \D_k \u_k + 2 \omega r_k X_k a_{k+1} \bmu_k^\top \u_k - \lambda r_k X_k b_{k+1} - \lambda b_{k+1} \bmu_k^\top \u_k + \frac{\lambda^2}{4\omega } c_{k+1} 
            \\
            & ~~~~~~~~~~~~~ - 2 \omega X_k \w_k^\top \Q_k \u_k + \omega X_k^2 \w_k^\top \Q_k \w_k
        \end{aligned}
    \end{equation*}
    Then the first order optimal condition is:
    \begin{equation*}
        \begin{aligned}
            2 \omega \D_k \u_k + 2 \omega r_k a_{k+1} X_k \bmu_k - \lambda b_{k+1} \bmu_k - 2 \omega X_k \Q_k \w_k= 0
        \end{aligned}
    \end{equation*}
    which yields the optimal policy:
    \begin{equation*}
        \begin{aligned}
            \u_k^\ast = \left( \frac{\lambda}{2\omega } b_{k+1} - r_k a_{k+1} X_k \right) \D_k^{-1} \bmu_k + X_k \D_k^{-1} \Q_k \w_k
        \end{aligned}
    \end{equation*}
    Then we can substitute the policy to  value function and verify that
    \begin{equation*}
        \begin{aligned}
            J_k(X_k) 
            &= 
            \omega a_{k+1} r_k^2 X_k^2 + \omega \u_k^\top \D_k \u_k - \lambda r_k X_k b_{k+1} - 2 \omega \left( \frac{\lambda}{2\omega } b_{k+1} - r_k a_{k+1} X_k \right) \bmu_k^\top \u_k 
            \\
            & + \frac{\lambda^2}{4\omega } c_{k+1} - 2 \omega X_k \w_k^\top \Q_k \u_k + \omega X_k^2 \w_k^\top \Q_k \w_k
            \\ 
            &= 
            \omega a_{k+1} r_k^2 X_k^2 + \frac{\lambda^2}{4\omega } c_{k+1} - \lambda r_k X_k b_{k+1} + \omega X_k^2 \w_k^\top \Q_k \w_k
            - \omega \left( \frac{\lambda}{2\omega } b_{k+1} - r_k a_{k+1} X_k \right)^2 \bmu_k^\top \D_k^{-1} \bmu_k \\
            &~~~~ - 2 \omega X_k \left( \frac{\lambda}{2\omega } b_{k+1} - r_k a_{k+1} X_k \right) \w_k^\top \Q_k \D_k^{-1} \bmu_k 
            - \omega X_k^2 \w_k^\top \Q_k \D_k^{-1} \Q_k \w_k  
            \\
            &= 
            \omega X_k^2 (a_{k+1} r_k^2  + \w_k^\top \Q_k \w_k - r_k^2 a_{k+1}^2 \bmu_k^\top \D_k^{-1} \bmu_k - \w_k^\top \Q_k \D_k^{-1} \Q_k \w_k + 2 r_k a_{k+1} \w_k^\top \Q_k \D_k^{-1} \bmu_k) \\
            & ~~ - \lambda X_k (r_k b_{k+1} - r_k b_{k+1} a_{k+1}  \bmu_k^\top \D_k^{-1} \bmu_k + b_{k+1} \w_k^\top \Q_k \D_k^{-1} \bmu_k ) + \frac{\lambda^2}{4\omega } c_{k+1} - \frac{\lambda^2}{4\omega } b_{k+1}^2 \bmu_k^\top \D_k^{-1} \bmu_k
            \\
            &= \omega \a_k X_k^2 - \lambda \b_k X_k + \frac{\lambda^2}{4\omega } \c_k
        \end{aligned}
    \end{equation*}
    We can verify the expectation at time $T$, is apparently true since $\b_T=1$ and $\c_T=0$. And we suppose that
    \begin{equation*}
        \E_{k+1}[X_T] = \b_{k+1} X_{k+1} - \frac{\lambda}{2\omega } \c_{k+1}
    \end{equation*}
    Then we have 
    \begin{equation*}
        \begin{aligned}
            \E_k[X_T] &= \E_k[\E_{k+1}[X_T]] = \E_k[\b_{k+1} X_{k+1} - \frac{\lambda}{2\omega } \c_{k+1}]\\
            &= \E_k[\b_{k+1} (r_k X_k+\P_k^\top \u_k^\ast) - \frac{\lambda}{2\omega } \c_{k+1}] = \b_{k+1} (r_k X_k+\P_k^\top \E_k[\u_k^\ast]) - \frac{\lambda}{2\omega } \c_{k+1}\\
            &= r_k \b_{k+1} X_k + \b_{k+1} \bmu_k^\top \left[ \left( \frac{\lambda}{2\omega } \b_{k+1} - r_k \a_{k+1} X_k \right) \D_k^{-1} \bmu_k + X_k \D_k^{-1} \Q_k \w_k \right] - \frac{\lambda}{2\omega } \c_{k+1}
            \\
            &= 
            (r_k \b_{k+1} + \b_{k+1} \bmu_k^\top \D_k^{-1} \Q_k \w_k - r_k \b_{k+1} \a_{k+1} \bmu_k^\top \D_k^{-1} \bmu_k ) X_k - \frac{\lambda}{2\omega } (\c_{k+1} - \b_{k+1}^2 \bmu_k^\top \D_k^{-1} \bmu_k)
            \\
            &= \b_k X_k -\frac{\lambda}{2\omega }\c_k
        \end{aligned}
    \end{equation*} 
\end{proof}

\subsection{Proof of Theorem~\ref{theorem:2.2}} 
\begin{proof}
    The objective function of $\cP_{\RRMV}(\omega,T)$ is equivalent to 
    \begin{align*}
    \omega \cdot \E[X_T^2] - \omega^2 \cdot \E[X_T]^2 - \E[X_T] + \omega \cdot \sum_{t=0}^{T-1} \E\big[\|\u_t - X_t \w_t \|_{\Q_t}^2 \big]
    \end{align*}
    And for any fixed value of $\omega$ and $\lambda$, the objective of auxiliary problem is 
    \begin{align*}
        \E[ \omega X_T^2 - \lambda X_T] + \omega \cdot \sum_{t=0}^{T-1} E[\| \u_t - X_t \w_t\|_{\Q_t}^2 ]
    \end{align*}
    Let $\Pi_\cA(\omega, \lambda)$ and $\Pi_\RRMV(\omega,T)$ be the optimal solution set of these two problems. 
    Following Theorem 1 in \cite{LiNg:2000MF}, it is not hard to show that $\Pi_\RRMV(\omega,T) \subseteq \cup_\lambda \Pi_\cA(\omega, \lambda)$ for fixed $\omega$. 
    So, solving the problem $\cP_{\RRMV}(\omega,T)$ can be reduced to optimizing over $\lambda$. 
    Assume $\pi^* \in \Pi_\cA(\omega, \lambda^*)$, the necessary condition of $\pi^* \in \Pi_\RRMV(\omega,T)$ can be derived from the first-order necessary optimality condition for both problems. 
    \begin{align*}
        & \omega
        \frac{\partial}{\partial\lambda}
        \Bigg(
        \E[X_T^2(\lambda^*,\omega)]
        + \sum_{t=0}^{T-1} \E\left[\|\u_t(\lambda^*)-X_t \w_t\|_{\Q_t}^2\right]
        \Bigg)
        -
        \Big[1+2\omega\,\mathbb E[X_T]\big|_{\pi^*}\Big]
        \frac{\partial}{\partial\lambda}
        \E[X_T(\lambda^*,\omega)]
        =0 
        \\
        & \omega
        \frac{\partial}{\partial\lambda}
        \Bigg(
        \E[X_T^2(\lambda^*,\omega)]
        + \sum_{t=0}^{T-1} \E \left[\|u_t(\lambda^*)-X_t \w_t\|_{\Q_t}^2\right]
        \Bigg)
        -
        \lambda^*
        \frac{\partial}{\partial\lambda}
        \E[X_T(\lambda^*,\omega)]
        =0
    \end{align*}
    which implies that $1+2\omega \E[X_T]\big|_{\pi^*} = \lambda^*$ must hold. 

    From Theorem~\ref{theorem:opt_aux_prob}, we can show that $\lambda^\ast = 2 \left( \b_0 X_0 - \frac{\lambda^\ast}{2\omega } \c_0 \right) \omega + 1$ for fixed $\omega$. 
    It implies that when $\lambda^* = (2\omega b_0 X_0 +1)/(1+ c_0)$, the policy $\u^{\cA}(X_k, \lambda^*)$ defined in (\ref{def_A_opt_u}) solves $\cP_{\RRMV}(\omega, T)$, yielding
    \begin{align*}
        \u^*_k(X_k)
        = \Bigg( \frac{(1 + 2 \omega b_0 X_0) b_{k+1} }{2\omega (1 + c_0)} - r a_{k+1} X_k \Bigg)
        \D_k^{-1} \bmu + X_k \D_k^{-1} \Q_k \w_k,
        \quad k=0,\ldots, T-1
    \end{align*}
     where $b_0,c_0$ is from the recursion in \eqref{equation:rrmv_rf_para}. 

     Furthermore, when the weighting parameter $\omega$ is solved by using the constraint $\E[X_T] = X_\tg$, which implies that
    \begin{align*}
        \omega^* = \frac{ c_0}{2\big( b_0 X_0 - X_{\tg} \cdot(1+ c_0) \big)}.
    \end{align*}
    Then, the policy in (\ref{RRMVF_omega_uk}) also solves $\tilde{\cP}_{\RRMV}(X_{\tg}, T)$ when $\omega = \omega^*$. 
    It completes the proof of Theorem~\ref{theorem:2.2}. 
\end{proof}

\subsection{Proof of Proposition~\ref{prop:sr_compare}}
\begin{proof}
Following the results in Corollary~\ref{coro:os_SR_w=0}, let $\Q_k =\Q =0$, $\w_k =\w =0$ and assume that we have the true parameters $\bmu$ and $\bsigma$. 
Define $ \eta := r(1- {a}_{k+1} \bmu^\top \D_k^{-1} \bmu) = \frac{r}{1 + \bmu^\top \bsigma^{-1} \bmu}$, we have
\begin{align*}
    {\D}_k & = {a}_{k+1}({\bsigma}+  {\bmu} {\bmu}^\top)  \,,\\
    {a}_k  & =  {a}_{k+1} r  \eta  \,,\\
    {b}_k  & = {b}_{k+1}  \eta \,,             \\
    {c}_k  & = {c}_{k+1} - {b}_{k+1}^2  {\bmu}^\top {\D}_k^{-1}  {\bmu} \,.
\end{align*}
And we can get the explicit formula of these parameters: 
\begin{align*}
    {a}_k  & =  (r  \eta)^{T-k} \,, \\
    {b}_k  & =  \eta^{T-k} = r^{-(T-k)} {a}_k \,,\\
    {c}_k  & = ( \eta/r)^{T-k} - 1 \,.
\end{align*}
It is straightforward to have the following parameters
\begin{equation*}     \begin{aligned}
    \eta_k 
    &= r(1- a_{k+1} \bmu^\top \bsigma^{-1} \bmu) 
    = r(1-\frac{\bmu^\top\bsigma^{-1}\bmu}{1+\bmu^\top \bsigma^{-1} \bmu}) = \eta\\
    \lambda_k 
    &= \frac{ b_0 - X_\tg}{ c_0}  b_{k+1} \bmu^\top \D_k^{-1} \bmu 
    = \frac{ b_0 - X_\tg}{ c_0}  b_{k+1}  a_{k+1}^{-1} (1 - \eta/r) \\
    &= \frac{ \eta^T - X_\tg}{ \eta^T - r^T} r^{k+1} (1 - \eta/r) 
    = r^{k+1} \frac{ \eta^T - X_\tg}{ \eta^T - r^T} \frac{\bmu^\top \bsigma^{-1} \bmu}{1+\bmu^\top \bsigma^{-1} \bmu}
\end{aligned} \end{equation*}

For the expectation part, 
\begin{equation*}     \begin{aligned}
    \E[X_{k}] 
    &= \prod_{i=0}^{k-1} \eta_i + \sum_{i=0}^{k-1}(\lambda_{i}\prod_{j=i+1}^{k-1}\eta_{j})\\
    &= \eta^k + \sum_{i=0}^{k-1} \eta^{k-i-1} r^{i+1} \frac{ \eta^T - X_\tg}{ \eta^T - r^T} \frac{\bmu^\top \bsigma^{-1} \bmu}{1+\bmu^\top \bsigma^{-1} \bmu} \\
    &= \eta^k + r^k \frac{ \eta^T - X_\tg}{ \eta^T - r^T} \frac{\bmu^\top \bsigma^{-1} \bmu}{1+\bmu^\top \bsigma^{-1} \bmu} \sum_{i=0}^{k-1} (1-\frac{\bmu^\top\bsigma^{-1}\bmu}{1+\bmu^\top \bsigma^{-1} \bmu})^{k-i-1} \\
    &= \eta^k \bigg[ 1 - \frac{ \eta^T - X_\tg}{ \eta^T - r^T} (1-(1-\frac{\bmu^\top\bsigma^{-1}\bmu}{1+\bmu^\top \bsigma^{-1} \bmu})^{-k}) \bigg]\\
    &= \eta^k \frac{X_\tg - r^T}{ \eta^T -r^T} + r^k \frac{ \eta^T - X_\tg}{ \eta^T - r^T} 
\end{aligned} \end{equation*}
We can see that $\E[X_T] = X_\tg$, and the results are valid even when we use $\hat\bsigma \neq \bsigma$ for the portfolio policy.
This is because all the $\bsigma$ above are from the optimal policy, not the out-of-sample distribution of return. 

To derive the variance, we need the following parameters again, from equation(\ref{para:variance}): 
\begin{equation*}     \begin{aligned}
        \nu_k &:= \bigg[(\frac{ b_0 -X_\tg}{ c_0})  b_{k+1}- r  a_{k+1} \E[X_k] \bigg]^2 \bmu^\top \D_k^{-1} \bsigma \D_k^{-1} \bmu\\
        &= \bigg[(\frac{ b_0 -X_\tg}{ c_0}) r^{-T+k} -  \E[X_k] \bigg]^2 r^2  a_{k+1}^2 \bmu^\top \D_k^{-1} \bsigma \D_k^{-1} \bmu\\
        &=  r^2 \bigg[(\frac{ b_0 -X_\tg}{ c_0}) r^{-T+k} - \eta^k \frac{X_\tg - r^T}{ \eta^T -r^T} - r^k \frac{ \eta^T - X_\tg}{ \eta^T - r^T}\bigg]^2 \frac{\bmu^\top ({\bsigma} + \rho \Q)^{-1} \bsigma ({\bsigma} + \rho \Q)^{-1} \bmu}{(1+\bmu^\top ({\bsigma} + \rho \Q)^{-1} \bmu)^2} \\
        &= r^2 \frac{\bmu^\top ({\bsigma} + \rho \Q)^{-1} \bsigma ({\bsigma} + \rho \Q)^{-1} \bmu}{(1+\bmu^\top ({\bsigma} + \rho \Q)^{-1} \bmu)^2} \bigg[\frac{X_\tg - r^T}{ \eta^T -r^T}\bigg]^2 \eta^{2k} 
        = \Gamma \bigg[\frac{X_\tg - r^T}{ \eta^T -r^T}\bigg]^2 \eta^{2k}\\
         g_k &:=r^2 [(1- a_{k+1} \bmu^\top \D_k^{-1} \bmu)^2 +  a_{k+1}^2 \bmu^\top \D_k^{-1} \bsigma \D_k^{-1} \bmu] = \Gamma + \eta^2 :=  g
\end{aligned} \end{equation*}
where $\Gamma = r^2 \frac{\bmu^\top \bsigma^{-1} \bmu}{(1+\bmu^\top \bsigma^{-1} \bmu)^2}$. 

The variance can be written as: 
\begin{equation*}     \begin{aligned}
    \Var(X_k) 
    &= \sum_{i=0}^{k-1} \nu_i \prod_{j=i+1}^{k-1}  g_j 
    = \sum_{i=0}^{k-1} \Gamma \bigg[\frac{X_\tg - r^T}{ \eta^T -r^T}\bigg]^2 \eta^{2i}(\Gamma + \eta^2)^{k-i-1}\\
    &= \Gamma (\Gamma + \eta^2)^{k-1} \bigg[\frac{X_\tg - r^T}{ \eta^T -r^T}\bigg]^2 \sum_{i=0}^{k-1}  (\frac{\eta^2}{\Gamma + \eta^2})^i\\
    \Var(X_T) &= \Gamma (\Gamma + \eta^2)^{T-1} \bigg[\frac{X_\tg - r^T}{ \eta^T -r^T}\bigg]^2 \sum_{i=0}^{T-1}  (\frac{\eta^2}{\Gamma + \eta^2})^i 
    = \bigg[\frac{X_\tg - r^T}{ \eta^T -r^T}\bigg]^2 ((\Gamma + \eta^2)^T - \eta^{2T})
\end{aligned} \end{equation*}
The efficient frontier becomes 
\begin{equation*}     \begin{aligned}
    \Var(X_T) &= \bigg[\frac{\E[X_T] - r^T}{ \eta^T -r^T}\bigg]^2 ((\Gamma + \eta^2)^T - \eta^{2T}),
\end{aligned} \end{equation*}
where $\eta = \frac{r}{1 + \bmu^\top \bsigma^{-1} \bmu}$ and $\Gamma = r^2 \frac{\bmu^\top \bsigma^{-1} \bmu}{(1+\bmu^\top \bsigma^{-1} \bmu)^2}$. 
We can show that $\eta^2 + \Gamma = \frac{r^2}{1 + \bmu^\top \bsigma^{-1} \bmu} = r \eta$, then the efficient frontier can be written as
\begin{equation*}     \begin{aligned}
    \Var(X_T) &= \left[ \E[X_T] - r^T \right]^2 \frac{\eta^T}{r^T - \eta^T}
\end{aligned} \end{equation*}

Since we have explicit constraint on the terminal expected return, we can easily get the standardized Sharpe ratio over the whole $T$ periods:
\begin{equation*}     \begin{aligned}
    \SR^2(X_T) 
    &= \frac{1}{T} \frac{(X_\tg - r^T)^2}{\Var(X_T)} 
    = \frac{1}{T} \frac{(X_\tg - r^T)^2}{\bigg[\frac{X_\tg - r^T}{ \eta^T -r^T}\bigg]^2 ((\Gamma + \eta^2)^T - \eta^{2T})}
    = \frac{1}{T} \frac{( \eta^T -r^T)^2}{((\Gamma + \eta^2)^T - \eta^{2T})}
    % \\
    % &= \frac{1}{T} \frac{((\frac{r}{1+\bmu^\top \bsigma^{-1} \bmu})^T -r^T)^2}{((r^2 \frac{\bmu^\top ({\bsigma} + \rho \Q)^{-1} \bsigma ({\bsigma} + \rho \Q)^{-1} \bmu}{(1+\bmu^\top ({\bsigma} + \rho \Q)^{-1} \bmu)^2} + r^2(1-\frac{\bmu^\top\bsigma^{-1}\bmu}{1+\bmu^\top \bsigma^{-1} \bmu})^2)^T - (r(1-\frac{\bmu^\top\bsigma^{-1}\bmu}{1+\bmu^\top \bsigma^{-1} \bmu}))^{2T})} \\
    % &= \frac{1}{T} \frac{(1 - (1+\bmu^\top \bsigma^{-1} \bmu)^T )^2}{(\bmu^\top ({\bsigma} + \rho \Q)^{-1} \bsigma ({\bsigma} + \rho \Q)^{-1} \bmu + (1+\varepsilon^\top\bsigma^{-1}\bmu)^2)^T - (1+\varepsilon^\top\bsigma^{-1}\bmu)^{2T}}
\end{aligned} \end{equation*}
Substituting the parameters, we can derive the result of the one-period Sharpe ratio achieved by the MMV policy, 
\begin{align*}
    \SR_{\max}
    &=
    \frac{1}{\sqrt{T}}\sqrt{(1+\bmu^\top\bsigma^{-1}\bmu)^T-1} \,.
\end{align*}

Next, we consider holding the static MV portfolio over $T$ periods. 
The corresponding terminal wealth will be $X_T = X_0 \cdot \prod_{t=0}^{T-1} (r + \frac{1}{2 \omega} \P_t^\top \bsigma^{-1} \bmu)$. 
It is straightforward that 
\begin{align*}
    \Var(\prod_{t=0}^{T-1} (r + \frac{1}{2 \omega} \P_t^\top \bsigma^{-1} \bmu))
    & = \E[\prod_{t=0}^{T-1} (r + \frac{1}{2 \omega} \P_t^\top \bsigma^{-1} \bmu)^2] - \E[\prod_{t=0}^{T-1} (r + \frac{1}{2 \omega} \P_t^\top \bsigma^{-1} \bmu)]^2 \,,
    \\
    \E[(r + \frac{1}{2 \omega} \P_t^\top \bsigma^{-1} \bmu)^2]
    & = (r + \frac{1}{2 \omega} \bmu^\top \bsigma^{-1} \bmu)^2 + (\frac{1}{2 \omega})^2 \bmu^\top \bsigma^{-1} \bmu \,.
\end{align*}
Then, the one-period Sharpe ratio obtained is 
\[
\SR_{\text{hold}} = \frac{1}{\sqrt{T}} \frac{(r + \frac{1}{2 \omega} \bmu^\top \bsigma^{-1} \bmu)^T - r^T}{\sqrt{((r + \frac{1}{2 \omega} \bmu^\top \bsigma^{-1} \bmu)^2 + (\frac{1}{2 \omega})^2 \bmu^\top \bsigma^{-1} \bmu )^T - (r + \frac{1}{2 \omega} \bmu^\top \bsigma^{-1} \bmu)^{2T}}} \,.
\]
Let $M=\frac{\bmu^\top\bsigma^{-1}\bmu}{\,2\omega r +  \bmu^\top\bsigma^{-1}\bmu\,}$, we have 
\[
\SR_{\mathrm{hold}} = \frac{1}{\sqrt{T}} \frac{1-(1-M)^T}{\sqrt{\left(1+\frac{M(1-M)}{2\omega r}\right)^T-1}} \,.
\]

Finally, we compare $\SR_{\max}$ with $\SR_{\mathrm{hold}}$. 
When $T=1$, $\SR_{\max} = \SR_{\mathrm{hold}} = \sqrt{\bmu^\top\bsigma^{-1}\bmu}$. 
However, for $T>1$, we can show that
$M \in (0,1)$, since $\omega>0$. 
Then, by Bernoulli's inequality, we have $(1-M)^T > 1-TM$, leading to $1-(1-M)^T < TM$. 
Similarly, we have $\sqrt{\left(1+\frac{M(1-M)}{2\omega r}\right)^T-1} > \sqrt{\frac{TM(1-M)}{2\omega r}}$. 
Therefore, we can show that 
\[
\SR_{\mathrm{hold}} 
< 
\sqrt{\frac{2 \omega r M}{1-M}}
= 
\sqrt{\bmu^\top\bsigma^{-1}\bmu} 
< 
\SR_{\max} \,,
\]
which completes the proof.
\end{proof}

\section{Proofs for Section 3}
\subsection{Efficient Frontier in Proposition \ref{prop:eff_frontier}}
\begin{proof}
Since $\E[\hat X_0] = 1$, we can derive the expression of $\E[\hat X_k]$ by
\begin{equation*} 
    \begin{aligned}
        \hat X_{k+1}
        & = r \hat X_k + \P_k^\top \hat \u_k = r \hat X_k + \P_k^\top \hat\D_k^{-1} \hat\bmu (\frac{\hat{b}_0 - X_\tg}{\hat{c}_0} \hat{b}_{k+1} - r \hat X_k \hat{a}_{k+1}) + \hat X_k \P_k^\top \hat\D_k^{-1} \Q_k \w_k  \\
        & = \hat X_k(r- r \hat{a}_{k+1} P_k^\top \hat\D_k^{-1}  \hat\bmu + \P_k^\top \hat\D_k^{-1} \Q_k \w_k) + \frac{\hat{b}_0 - X_\tg}{\hat{c}_0} \hat{b}_{k+1} \P_k^\top \hat\D_k^{-1} \hat\bmu \\
        \E[\hat X_{k+1}|\hat X_k] 
        & = \hat X_k (r- r \hat{a}_{k+1} \bmu^\top \hat\D_k^{-1}  \hat\bmu + \bmu^\top \hat\D_k^{-1} \Q_k \w_k) + \frac{\hat{b}_0 - X_\tg}{\hat{c}_0} \hat{b}_{k+1} \bmu^\top \hat\D_k^{-1} \hat\bmu \\
        \E[\hat X_{k+1}]
        & = \E[\hat X_k](r- r \hat{a}_{k+1} \bmu^\top \hat\D_k^{-1}  \hat\bmu + \bmu^\top \hat\D_k^{-1} \Q_k \w_k) + \frac{\hat{b}_0 - X_\tg}{\hat{c}_0} \hat{b}_{k+1} \bmu^\top \hat\D_k^{-1} \hat\bmu := \hat \eta_k \E[\hat X_k] + \hat \lambda_k \\
        \E[\hat X_{k+1}]
        & = \prod_{i=0}^{k} \hat\eta_i + \sum_{i=0}^k(\hat\lambda_{i}\prod_{j=i+1}^{k}\hat\eta_{j}) \\
        \E[\hat X_{T}]
        & = \prod_{i=0}^{T-1} \hat\eta_i + \sum_{i=0}^{T-1}(\hat\lambda_{i}\prod_{j=i+1}^{T-1}\eta_{j})
    \end{aligned} 
\end{equation*}
where $\hat\lambda_k = \frac{\hat{b}_0 - X_\tg}{\hat{c}_0} \hat{b}_{k+1} \bmu^\top \hat\D_k^{-1} \hat\bmu$ and $\hat\eta_k = r- r \hat{a}_{k+1} \bmu^\top \hat\D_k^{-1}  \hat\bmu +  \bmu^\top \hat\D_k^{-1} \Q_k \w_k$ for $k=0,\ldots,T-1$. 

For the expression of $\Var(\hat X_k)$, we follow the conditional variance of $\Var(\hat X_{k+1}) = \E[\Var(\hat X_{k+1}|\hat X_k)] + \Var(\E[\hat X_{k+1}|\hat X_k])$, where  
\begin{align*}
    \Var(\hat X_{k+1}|\hat X_k)
    & = (\frac{\hat{b}_0 - X_\tg}{\hat{c}_0} \hat{b}_{k+1} - \hat X_k r  \hat{a}_{k+1})^2 \hat\bmu^\top \hat\D_k^{-1} \bsigma \hat\D_k^{-1}  \hat\bmu +  \hat X_k^2 \w_k^\top \Q_k \hat\D_k^{-1}\bsigma\hat\D_k^{-1} \Q_k \w_k \\
    & + 2 X_k (\frac{\hat{b}_0 - X_\tg}{\hat{c}_0} \hat{b}_{k+1} - \hat X_k r  \hat{a}_{k+1})  \hat\bmu^\top \hat\D_k^{-1} \bsigma \hat\D_k^{-1} \Q_k \w_k
\end{align*}
and it is straightforward to have 
\begin{align*}
    &\E[\Var(\hat X_{k+1}|\hat X_k)] 
    \\
    = &(\Var(\hat X_k) + \E[\hat X_k]^2) \bigg[ r^2 \hat{a}_{k+1}^2  \hat\bmu^\top \hat\D_k^{-1} \bsigma \hat\D_k^{-1}  \hat\bmu + \w_k^\top \Q_k \hat\D_k^{-1}\bsigma\hat\D_k^{-1} \Q_k \w_k - 2 r \hat{a}_{k+1}  \hat\bmu^\top \hat\D_k^{-1} \bsigma \hat\D_k^{-1} \Q_k \w_k \bigg]\\
    & - \E[\hat X_k] \bigg[ -2 \frac{\hat{b}_0 - X_\tg}{\hat{c}_0}\hat{b}_{k+1}  \hat\bmu^\top \hat\D_k^{-1} \bsigma \hat\D_k^{-1} \Q_k \w_k + 2 \frac{\hat{b}_0 -X_\tg}{\hat{c}_0}\hat{b}_{k+1} r \hat{a}_{k+1}  \hat\bmu^\top \hat\D_k^{-1} \bsigma \hat\D_k^{-1}  \hat\bmu \bigg]\\
    & + (\frac{\hat{b}_0 -X_\tg}{\hat{c}_0} \hat{b}_{k+1})^2  \hat\bmu^\top \hat\D_k^{-1} \bsigma \hat\D_k^{-1}  \hat\bmu \\
    & \Var[\E(\hat X_{k+1}|\hat X_k)]
    = \Var(X_k) (r- r \hat{a}_{k+1} \bmu^\top \hat\D_k^{-1}  \hat\bmu +  \bmu^\top \hat\D_k^{-1} \Q_k \w_k)^2
\end{align*}
Combine the two terms, we have
\begin{align*}
    \Var(\hat X_{k+1})
    & = \Var(\hat X_k) \bigg[ r^2 \hat{a}_{k+1}^2  \hat\bmu^\top \hat\D_k^{-1} \bsigma \hat\D_k^{-1}  \hat\bmu + \w_k^\top \Q_k \hat\D_k^{-1} \bsigma \hat\D_k^{-1} \Q_k \w_k - 2 r \hat{a}_{k+1}  \hat\bmu^\top \hat\D_k^{-1} \bsigma \hat\D_k^{-1} \Q_k \w_k  \\
    & + (r- r \hat{a}_{k+1} \bmu^\top \hat\D_k^{-1}  \hat\bmu + \bmu^\top \hat\D_k^{-1} \Q_k \w_k)^2 \bigg] \\
    & + \E[\hat X_k]^2 \bigg[ r^2 \hat{a}_{k+1}^2  \hat\bmu^\top \hat\D_k^{-1} \bsigma \hat\D_k^{-1}  \hat\bmu + \w_k^\top \Q_k \hat\D_k^{-1} \bsigma \hat\D_k^{-1} \Q_k \w_k - 2 r \hat{a}_{k+1}  \hat\bmu^\top \hat\D_k^{-1} \bsigma \hat\D_k^{-1} \Q_k \w_k \bigg]\\
    & + \E[\hat X_k] \bigg[ 2 \frac{\hat{b}_0 - X_\tg}{\hat{c}_0}\hat{b}_{k+1}  \hat\bmu^\top \hat\D_k^{-1} \bsigma \hat\D_k^{-1} \Q_k \w_k - 2 \frac{\hat{b}_0 -X_\tg}{\hat{c}_0}\hat{b}_{k+1} r \hat{a}_{k+1}  \hat\bmu^\top \hat\D_k^{-1} \bsigma \hat\D_k^{-1}  \hat\bmu \bigg]\\
    & + (\frac{\hat{b}_0 -X_\tg}{\hat{c}_0} \hat{b}_{k+1})^2  \hat\bmu^\top \hat\D_k^{-1} \bsigma \hat\D_k^{-1}  \hat\bmu
\end{align*}
Define the following parameters, 
\begin{equation}
    \begin{aligned}
        \hat\nu_k &:= 
        \E[\hat X_k]^2 \bigg[ r^2 \hat{a}_{k+1}^2  \hat\bmu^\top \hat\D_k^{-1} \bsigma \hat\D_k^{-1}  \hat\bmu + \w_k^\top \Q_k \hat\D_k^{-1} \bsigma \hat\D_k^{-1} \Q_k \w_k - 2 r \hat{a}_{k+1}  \hat\bmu^\top \hat\D_k^{-1} \bsigma \hat\D_k^{-1} \Q_k \w_k \bigg]\\
        & ~~~~~ + \E[\hat X_k] \bigg[ 2 \frac{\hat{b}_0 - X_\tg}{\hat{c}_0}\hat{b}_{k+1}  \hat\bmu^\top \hat\D_k^{-1} \bsigma \hat\D_k^{-1} \Q_k \w_k - 2 \frac{\hat{b}_0 - X_\tg}{\hat{c}_0}\hat{b}_{k+1} r \hat{a}_{k+1}  \hat\bmu^\top \hat\D_k^{-1} \bsigma \hat\D_k^{-1}  \hat\bmu \bigg]\\
        & ~~~~~ + (\frac{\hat{b}_0 - X_\tg}{\hat{c}_0} \hat{b}_{k+1})^2  \hat\bmu^\top \hat\D_k^{-1} \bsigma \hat\D_k^{-1}  \hat\bmu\\
        \hat{g}_k &:= 
        r^2 \hat{a}_{k+1}^2  \hat\bmu^\top \hat\D_k^{-1} \bsigma \hat\D_k^{-1}  \hat\bmu + \w_k^\top \Q_k \hat\D_k^{-1} \bsigma \hat\D_k^{-1} \Q_k \w_k - 2 r \hat{a}_{k+1}  \hat\bmu^\top \hat\D_k^{-1} \bsigma \hat\D_k^{-1} \Q_k \w_k\\
        & ~~~~ + (r- r \hat{a}_{k+1} \bmu^\top \hat\D_k^{-1}  \hat\bmu + \bmu^\top \hat\D_k^{-1} \Q_k \w_k)^2 
    \end{aligned}
    \label{para:variance}
\end{equation}
then $\Var(\hat X_{k+1})$ can be written as
\begin{equation*}
\begin{aligned}
    \Var(\hat X_{k+1}) & = \Var(\hat X_k) \hat{g}_k + \hat\nu_k
\end{aligned} 
\end{equation*}
where, starting from $\Var(X_0) = 0$, $\Var(\hat X_T) = \sum_{i=0}^{T-1} \hat\nu_i \prod_{j=i+1}^{T-1}\hat{g}_j$.
Generally the out-of-sample expectation $\E[\hat X_T] \neq X_\tg$ because of the estimation error, the expression of Sharpe ratio is
\begin{align*}
    \SR^2(\hat{X}_T) 
    &= \frac{1}{T} \frac{(\E[X_T] - r^T)^2}{\Var(X_T)} 
    = \frac{(\prod_{i=0}^{T-1} \hat\eta_i + \sum_{i=0}^{T-1}(\hat\lambda_{i}\prod_{j=i+1}^{T-1}\eta_{j})-r^T)^2}{T \sum_{i=0}^{T-1} \hat\nu_i \prod_{j=i+1}^{T-1}\hat{g}_j} 
\end{align*}
Moreover, if $\hat\bmu = \bmu$ with no estimation error, the out-of-sample expectation of $\E[\hat X_T]$ will equal to the in-sample evaluation, therefore $\E[\hat X_T] = X_\tg$. 
\end{proof}

\subsection{Proof of Corollary \ref{coro:os_SR_w=0}} \label{proof:B2}
\begin{proof}
Under our framework in this paper, we allow the regularization to have time varying effect on each period, i.e. $\Q_k$.  
Consider the special case of $\Q_k = \hat a_{k+1} \rho \Q$. 
% \begin{align*}
%     \hat{a}_k  
%     =  \hat{a}_{k+1} r^2 - \hat{a}_{k+1}^2 r^2 \hat\bmu^{\top} \hat{\D}_k^{-1}\hat\bmu
%     =  \hat{a}_{k+1} r^2 (1-\hat\bmu^{\top} (\hat{\bsigma}+\hat\bmu\hat\bmu^\top + \rho\Q)^{-1}\hat\bmu) = \frac{\hat{a}_{k+1} r^2}{1 + \hat\bmu^\top (\hat{\bsigma} + \rho\Q)^{-1}\hat\bmu}
% \end{align*}
% \begin{equation}     \begin{aligned}

% \end{aligned} \end{equation}
% which means that $a_k \to 0$ as $T \to \infty$. 
Define $\hat C := r(1- \hat{a}_{k+1} \hat\bmu^\top \hat\D_k^{-1} \hat\bmu)$, then we have the new parameter dynamics when $\Q_k = \hat{a}_{k+1} \rho \Q$: 
\begin{align*}
\begin{dcases}
    \hat{\D}_k & = \hat{a}_{k+1}(\hat{\bsigma}+  \hat{\bmu} \hat{\bmu}^\top + \rho \Q ),\\
    \hat{a}_k  & =  \hat{a}_{k+1} r \hat C, \\
    \hat{b}_k  & = \hat{b}_{k+1} \hat C,              \\
    \hat{c}_k  & = \hat{c}_{k+1} - \hat{b}_{k+1}^2  \hat{\bmu}^\top \hat{\D}_k^{-1}  \hat{\bmu}
\end{dcases}
\end{align*}
And we can get the explicit formula of these parameters: 
\begin{align*}
    \begin{dcases}
    \hat{a}_k  & =  (r \hat C)^{T-k}, \\
    \hat{b}_k  & = \hat C^{T-k} = r^{-(T-k)} \hat{a}_k, \\
    \hat{c}_k  & = (\hat C/r)^{T-k} - 1
    \end{dcases}
\end{align*}
It is straightforward to have the following parameters
\begin{equation*}     \begin{aligned}
    \hat\eta_k 
    &= r(1-\hat a_{k+1} \bmu^\top \bsigma^{-1} \hat\bmu) 
    = r(1-\frac{\bmu^\top(\hat\bsigma + \rho \Q)^{-1}\hat\bmu}{1+\hat\bmu^\top (\hat\bsigma + \rho \Q)^{-1} \hat\bmu}) := \hat\eta\\
    \hat\lambda_k 
    &= \frac{\hat b_0 - X_\tg}{\hat c_0} \hat b_{k+1} \bmu^\top \D_k^{-1} \hat\bmu 
    = \frac{\hat b_0 - X_\tg}{\hat c_0} \hat b_{k+1} \hat a_{k+1}^{-1} (1 - \hat\eta/r) 
    % \\
    % &= \frac{\hat C_\rho^T - X_\tg}{\hat C_\rho^T - r^T} r^{k+1} (1 - \hat\eta/r) 
    = r^{k+1} \frac{\hat C^T - X_\tg}{\hat C^T - r^T} \frac{\bmu^\top (\hat\bsigma + \rho \Q)^{-1} \hat\bmu}{1+\hat\bmu^\top (\hat\bsigma + \rho \Q)^{-1} \hat\bmu}
\end{aligned} \end{equation*}

For the expectation part, 
\begin{equation*}     \begin{aligned}
    \E[X_{k}] 
    &= \prod_{i=0}^{k-1} \hat\eta_i + \sum_{i=0}^{k-1}(\hat\lambda_{i}\prod_{j=i+1}^{k-1}\hat\eta_{j})\\
    &= \hat\eta^k + \sum_{i=0}^{k-1} \hat\eta^{k-i-1} r^{i+1} \frac{\hat C^T - X_\tg}{\hat C^T - r^T} \frac{\bmu^\top (\hat\bsigma + \rho \Q)^{-1} \hat\bmu}{1+\hat\bmu^\top (\hat\bsigma + \rho \Q)^{-1} \hat\bmu} \\
    &= \hat\eta^k + r^k \frac{\hat C^T - X_\tg}{\hat C^T - r^T} \frac{\bmu^\top (\hat\bsigma + \rho \Q)^{-1} \hat\bmu}{1+\hat\bmu^\top (\hat\bsigma + \rho \Q)^{-1} \hat\bmu} \sum_{i=0}^{k-1} (1-\frac{\bmu^\top(\hat\bsigma + \rho \Q)^{-1}\hat\bmu}{1+\hat\bmu^\top (\hat\bsigma + \rho \Q)^{-1} \hat\bmu})^{k-i-1} \\
    &= \hat\eta^k \bigg[ 1 - \frac{\hat C^T - X_\tg}{\hat C^T - r^T} (1-(1-\frac{\bmu^\top(\hat\bsigma + \rho \Q)^{-1}\hat\bmu}{1+\hat\bmu^\top (\hat\bsigma + \rho \Q)^{-1} \hat\bmu})^{-k}) \bigg]\\
    &= \hat\eta^k \frac{X_\tg - r^T}{\hat C^T -r^T} + r^k \frac{\hat C^T - X_\tg}{\hat C^T - r^T} 
\end{aligned} \end{equation*}
To derive the variance, we need the following parameters again, from equation(\ref{para:variance}): 
\begin{equation*}     \begin{aligned}
        \hat\nu_k &:= \bigg[(\frac{\hat b_0 - X_\tg}{\hat c_0}) \hat b_{k+1}- r \hat a_{k+1} \E[X_k] \bigg]^2 \hat\bmu^\top \hat\D_k^{-1} \bsigma \hat\D_k^{-1} \hat\bmu\\
        &= \bigg[(\frac{\hat b_0 -X_\tg}{\hat c_0}) r^{-T+k} -  \E[X_k] \bigg]^2 r^2 \hat a_{k+1}^2 \hat\bmu^\top \hat\D_k^{-1} \bsigma \hat\D_k^{-1} \hat\bmu\\
        &=  r^2 \bigg[(\frac{\hat b_0 -X_\tg}{\hat c_0}) r^{-T+k} - \hat\eta^k \frac{X_\tg - r^T}{\hat C^T -r^T} - r^k \frac{\hat C^T - X_\tg}{\hat C^T - r^T}\bigg]^2 \frac{\hat\bmu^\top (\hat{\bsigma} + \rho \Q)^{-1} \bsigma (\hat{\bsigma} + \rho \Q)^{-1} \hat\bmu}{(1+\hat\bmu^\top (\hat{\bsigma} + \rho \Q)^{-1} \hat\bmu)^2} \\
        &= r^2 \frac{\hat\bmu^\top (\hat{\bsigma} + \rho \Q)^{-1} \bsigma (\hat{\bsigma} + \rho \Q)^{-1} \hat\bmu}{(1+\hat\bmu^\top (\hat{\bsigma} + \rho \Q)^{-1} \hat\bmu)^2} \bigg[\frac{X_\tg - r^T}{\hat C^T -r^T}\bigg]^2 \hat\eta^{2k} 
        = \hat\Gamma \bigg[\frac{X_\tg - r^T}{\hat C^T -r^T}\bigg]^2 \hat\eta^{2k}\\
        \hat g_k &:=r^2 [(1-\hat a_{k+1} \bmu^\top \hat\D_k^{-1} \hat\bmu)^2 + \hat a_{k+1}^2 \hat\bmu^\top \hat\D_k^{-1} \bsigma \hat\D_k^{-1} \hat\bmu] = \hat\Gamma + \hat\eta^2 := \hat g
\end{aligned} \end{equation*}
where $\hat\Gamma = r^2 \frac{\hat\bmu^\top (\hat{\bsigma} + \rho \Q)^{-1} \bsigma (\hat{\bsigma} + \rho \Q)^{-1} \hat\bmu}{(1+\hat\bmu^\top (\hat{\bsigma} + \rho \Q)^{-1} \hat\bmu)^2}$. 
The variance can be written as: 
\begin{equation*}     \begin{aligned}
    \Var(X_k) 
    &= \sum_{i=0}^{k-1} \hat\nu_i \prod_{j=i+1}^{k-1} \hat g_j 
    = \sum_{i=0}^{k-1} \hat\Gamma \bigg[\frac{X_\tg - r^T}{\hat C^T -r^T}\bigg]^2 \hat\eta^{2i}(\hat\Gamma + \hat\eta^2)^{k-i-1}\\
    &= \hat\Gamma (\hat\Gamma + \hat\eta^2)^{k-1} \bigg[\frac{X_\tg - r^T}{\hat C^T -r^T}\bigg]^2 \sum_{i=0}^{k-1}  (\frac{\hat\eta^2}{\hat\Gamma + \hat\eta^2})^i\\
    \Var(X_T) &= \hat\Gamma (\hat\Gamma + \hat\eta^2)^{T-1} \bigg[\frac{X_\tg - r^T}{\hat C^T -r^T}\bigg]^2 \sum_{i=0}^{T-1}  (\frac{\hat\eta^2}{\hat\Gamma + \hat\eta^2})^i 
    = \bigg[\frac{X_\tg - r^T}{\hat C^T -r^T}\bigg]^2 ((\hat\Gamma + \hat\eta^2)^T - \hat\eta^{2T})
\end{aligned} \end{equation*}
Since we have explicit constraint on the terminal expected return, we can easily get the standardized Sharpe ratio over the whole $T$ periods:
\begin{align*}
    \SR^2(\hat{X}_T) 
    &= \frac{1}{T} \frac{(\E[X_T] - r^T)^2}{\Var(X_T)} 
    = \frac{1}{T} \frac{(\hat \eta^T -r^T)^2}{((\hat\Gamma + \hat\eta^2)^T - \hat\eta^{2T})}
\end{align*}
\end{proof}  

\subsection{An Important Lemma}
\begin{lemma}\label{lemma:keylemma}
Given $\X \in{\mathbb R}^{n\times p}$ with $\X = \Z \bsigma^{\frac12}$ where the elements of $\Z$ are iid with zero mean, variance 1 and finite eighth order moment, $\bsigma$ is a  positive semi-definite matrix. For any deterministic $\a, \b \in{\mathbb R}^p$ with bounded norm and positive definite matrix $\S \in{\mathbb R}^{p\times p}$, define 
\begin{align*}
    m_n(z)=\a^\top \bigg(\frac{\X^\top \X}{n}+\S+z\I_p\bigg)^{-1}\b,
\end{align*}
and assume that $p,n$ tends to infinity with proportion $p/n= c>0$. Then for $z > 0$, 
\begin{align}
\label{eq:convergence_mz}
    m_n(z)-\a^\top \bigg( \frac{\bsigma}{1+s_n(z)}+\S+z\I_p \bigg)^{-1}\b ~\to~ 0 \quad \text{almost surely}.
\end{align}
Here, $s_n(z)$ uniquely solves the following equation:
\begin{align*}
    s_n(z)=\frac{c}{p}~\tr~\bsigma\bigg(\frac{\bsigma}{1+s_n(z)}+\S+z\I_p\bigg)^{-1}.
\end{align*}
The almost sure convergence also holds after we take derivative. 
\begin{align}
\label{eq:convergence_mz_deriv}
    m_n'(z)-\frac{ \rm{d} \bigg[ \a^\top \Big(\frac{\bsigma}{1+s_n(z)}+\S+z\I_p \Big)^{-1} \b \bigg]}{\rm{d} z}~\to~ 0 \quad \text{almost surely}.
\end{align}
\end{lemma}

\begin{proof}
The first almost sure convergence has been proved in Theorem 1 of \cite{rubio2011spectral}. Corollary 1.11 in \cite{widder1938stieltjes} showed that the convergence is indeed internally closed uniform convergence. Since the functions are all analytic, the second almost sure convergence directly follow. 
\end{proof}

\subsection{Proofs for Section 3.1}

\subsubsection{Proof of Theorem~\ref{theorem:single_mu_SR}}
\label{prof:theorem:single_mu_SR}
\begin{proof}
We consider general regularization $\Q$, the optimal portfolio using $\bsigma$ and $\hat \bmu$ is 
\begin{equation*}     
\begin{aligned}
    \u^\ast &= (\bsigma + \Q)^{-1} (\frac{1}{2\omega} \hat\bmu + \alpha \Q \bar\w) \,, 
    \\
    \frac{1}{2\omega} &= \frac{X_\tg - r - \hat\bmu^\top(\bsigma + \Q)^{-1} (\alpha \Q \bar\w)}{\hat\bmu^\top(\bsigma + \Q)^{-1} \hat\bmu} = \frac{X_\tg - r - \alpha (\varepsilon + \bmu)^\top (\bsigma + \Q)^{-1} \Q \bar\w}{(\varepsilon + \bmu)^\top(\bsigma + \Q)^{-1} (\varepsilon + \bmu)} \,.
\end{aligned} 
\end{equation*}
We would like to show the almost surely convergence of the terms that involve $\varepsilon$. 
By normality of $\varepsilon$, we can see that $\E[\varepsilon^\top (\bsigma + \Q)^{-1} \Q \bar\w] = 0$ and $\E[\varepsilon^\top (\bsigma + \Q)^{-1} \varepsilon] = \tr((\bsigma + \Q)^{-1} \bsigma/n)$. 
We start with the first order terms. 
Since $\varepsilon^\top (\bsigma + \Q)^{-1} \Q \bar\w$ is normal distribution, for any $\epsilon >0$, 
\begin{equation*}     
\begin{aligned}
    \mathbb{P}\left(|\varepsilon^\top (\bsigma + \Q)^{-1} \Q \bar\w | \geq \epsilon \right) 
    \leq 
    2 \exp \left( -\frac{2 n \epsilon^2}{\bar\w^\top \Q (\bsigma + \Q)^{-1} \bsigma (\bsigma + \Q)^{-1} \Q \bar\w} \right)
\end{aligned} 
\end{equation*}
By assumption, $\Q$ and $\bsigma$ are well-conditioned with bounded eigenvalues, and $\| \w \|_2 \leq \infty$. 
Then we can show that for some constant $C$, 
\begin{equation*}     
\begin{aligned}
    \sum_n^\infty \mathbb{P}\left(|\varepsilon^\top (\bsigma + \Q)^{-1} \Q \bar\w| \geq \epsilon \right) 
    \leq \sum_n^\infty 2 \exp \left( - C n\right) < \infty
\end{aligned} 
\end{equation*}
Therefore, by Borel-Cantelli Lemma we have the following almost surely convergence,  
\begin{equation*}     \begin{aligned}
    \varepsilon^\top (\bsigma + \Q)^{-1} \Q \bar\w \xrightarrow{a.s} 0, ~~~~ 
    \varepsilon^\top (\bsigma + \Q)^{-1} \bmu \xrightarrow{a.s} 0
\end{aligned} \end{equation*}
Now we focus on the second order term of $\varepsilon$, we can rewrite it as a weighted sum of chi-squared distribution
\begin{equation*}     
\begin{aligned}
    \varepsilon^\top (\bsigma + \Q)^{-1} \varepsilon = \frac{1}{n} \z^\top \bsigma^{1/2} (\bsigma + \Q)^{-1} \bsigma^{1/2} \z := \frac{1}{n} \z^\top \bm{B} \z \,,
\end{aligned} 
\end{equation*}
where $z_i$ follows the standard normal distribution and the eigenvalues of matrix $\bm{B} = \bsigma^{1/2} (\bsigma + \Q)^{-1} \bsigma^{1/2}$ are in constant order by assumption. 
Then by Hanson-Wright Inequalities, we have for some constant $C>0$,
\begin{equation*}     
\begin{aligned}
    \mathbb{P}\left(|\frac{1}{n} \z^\top \bm{B} \z - \E[\frac{1}{n} \z^\top \bm{B} \z]| \geq \epsilon \right) \leq 2 \exp \left( - C \min \left( \frac{n^2 \epsilon^2}{\| \bm{B} \|_F^2}, \frac{n \epsilon}{\| \bm{B} \|_2} \right) \right) \,,
\end{aligned} 
\end{equation*}
where $\| \bm{B} \|_F^2 \leq p \| \bm{B} \|_2^2$ and $\| \bm{B} \|_2 = O(1)$. 
Recall that we assume $p/n = c$ is a constant, so we can show that 
\begin{equation*} 
\begin{aligned}
    \sum_n^\infty \mathbb{P}\left(|\frac{1}{n} \z^\top \bm{B} \z - \E[\frac{1}{n} \z^\top \bm{B} \z]| \geq \epsilon \right) 
    \leq 2 \sum_n^\infty \exp \left( - C \min \left( \frac{n^2 \epsilon^2}{p \| \bm{B} \|_2^2}, \frac{n \epsilon}{\| \bm{B} \|_2} \right) \right) < \infty \,.
\end{aligned} 
\end{equation*}
Then by Borel-Cantelli Lemma, we can show the almost surely convergence,
\begin{equation*}
\begin{aligned}
    \varepsilon^\top (\bsigma + \Q)^{-1} \varepsilon \xrightarrow{a.s} \tr((\bsigma + \Q)^{-1} \bsigma/n) \,.
\end{aligned} 
\end{equation*}
Therefore, the following convergence holds
\begin{equation*}
\begin{aligned}
    \frac{1}{2\omega} \xrightarrow{a.s} \frac{1}{2\omega^\ast}
    &= \frac{X_\tg - r - \alpha \bmu^\top(\bsigma + \Q)^{-1} \Q \bar\w}{\bmu^\top(\bsigma + \Q)^{-1} \bmu + \tr((\bsigma + \Q)^{-1} \bsigma/n)} \,. 
\end{aligned} 
\end{equation*}
Then, we can show that the Sharpe ratio of terminal wealth almost surely converges. 
And its limiting value is 
\begin{align*}
    \SR_\infty(\Q,\alpha\bar\w)
    = 
    \frac{ \bmu^\top (\bsigma + \Q)^{-1} (\frac{\bmu}{2\omega^\ast}  + \alpha \Q \bar\w)}{\sqrt{(\frac{\bmu}{2\omega^\ast}  + \alpha \Q \bar\w)^T (\bsigma + \Q)^{-1} \bsigma (\bsigma + \Q)^{-1} (\frac{\bmu}{2\omega^\ast}  + \alpha \Q \bar\w) + \frac{1} {(2\omega^\ast)^2} \tr((\bsigma + \Q)^{-1}  \bsigma (\bsigma + \Q)^{-1} \bsigma / n)}}\,.
\end{align*}
By letting $\A = \bsigma + \Q$, $\d_\alpha = \frac{\bmu}{2\omega^\ast}  + \alpha \Q \bar\w$, $e_\bmu = \tr( \A^{-1}  \bsigma \A^{-1} \bsigma / n)$, we have 
\begin{align*}
    \SR_\infty(\Q,\alpha\bar\w)
    = 
    \frac{ \bmu^\top \A^{-1} \d_\alpha }{\sqrt{ \| \A^{-1} \d_\alpha \|_\bsigma + \frac{e_\bmu} {(2\omega^\ast)^2} }} \,,
\end{align*}
which completes the proof. 
\end{proof}

\subsubsection{Proof of Corollary \ref{coro:3.2}}
\begin{proof}
Consider the case where $\Q = \rho \bsigma$, for any $\rho \geq 0$, we denote the limiting values of expected terminal wealth and its variance as  
\begin{equation*}
\begin{aligned}
    \E[X_T]_\infty
    & = 
    \frac{1}{2\omega^\ast} (\rho + 1)^{-1} \bmu^\top \bsigma^{-1} \bmu + \alpha (\rho + 1)^{-1} \bmu^\top \bsigma^{-1} \Q \bar\w
    \\
    & = 
    \frac{1}{\rho + 1} \Big[ \frac{\bmu}{2\omega^\ast} ^\top \bsigma^{-1} \bmu + \alpha \bmu^\top \bsigma^{-1} \Q \bar\w \Big] \,,
    \\
    \Var{X_T}_\infty
    & = 
    (\frac{\bmu}{2\omega^\ast}  + \alpha \Q \bar\w)^\top (\bsigma + \Q)^{-1}  \bsigma (\bsigma + \Q)^{-1} (\frac{\bmu}{2\omega^\ast}  + \alpha \Q \bar\w) + (\frac{1}{2\omega^\ast})^2 \tr((\bsigma + \Q)^{-1} \bsigma (\bsigma + \Q)^{-1} \bsigma/n) \\
    & = 
    \frac{1}{(\rho + 1)^2} \Big[ (\frac{1}{2\omega^\ast})^2 ( \bmu^\top \bsigma^{-1} \bmu + p/n) + \alpha^2 \bar\w^\top \Q \bsigma^{-1} \Q \bar\w + 2 \alpha \frac{\bmu}{2\omega^\ast} ^\top \bsigma^{-1} \Q \bar\w \Big] \,,
\end{aligned}
\end{equation*}
where $\frac{1}{2\omega^\ast} = \frac{(\rho + 1) (X_\tg - r) - \alpha \bmu^\top \bsigma^{-1} \Q \bar\w}{ \bmu^\top \bsigma^{-1} \bmu + p/n }$. 
For any given $\alpha$ and $\rho$, the derivatives respect to $\w$ of both $\E[X_T]_\infty$ and $\Var{X_T}_\infty$ can be written as linear combinations of $\bsigma^{-1}\bmu$ and $\bsigma^{-1} \Q \bar\w$. 
It implies that the optimal reference weight should be $\Q \bar\w^\ast \propto \bmu$. 
Since we can adjust the scale of reference weight $\Q \bar\w$, we only need to find the optimal $\alpha$ when $\Q \bar\w^\ast = (\rho + 1) \bmu$. 
Notice that the limiting Sharpe ratio depends on $\alpha$, where 
\begin{align*}
    \SR_\infty^2(\Q, \alpha \bar\w^\ast)
    = 
    \frac{(\frac{1}{2\omega^\ast} + \alpha)^2 (\bmu^\top \bsigma^{-1} \bmu)^2}{(\frac{1}{2\omega^\ast})^2 ( \bmu^\top \bsigma^{-1} \bmu + p/n) + \alpha^2 \bmu^\top \bsigma^{-1} \bmu + 2 \alpha \frac{\bmu}{2\omega^\ast} ^\top \bsigma^{-1} \bmu}
\end{align*}
We can see that for any $\rho>0$, 
\begin{align*}
\SR_\infty^2(\rho \bsigma, \mathbf{0})
= 
\frac{(\bmu^\top \bsigma^{-1} \bmu)^2}{ \bmu^\top \bsigma^{-1} \bmu + p/n} \,.
\end{align*}
Then we check the FOC of squared limiting Sharpe ratio by letting $\frac{\partial \SR_\infty(\rho \bsigma, \alpha \bmu)}{\partial \alpha} = 0$. 
It implies that the optimal scale $\alpha^\ast = \frac{ X_\tg - r }{\bmu^\top \bsigma^{-1} \bmu}$. 
Plugging in this optimal scale, we have 
\begin{align*}
\SR_\infty^2(\Q, \alpha^\ast \bar\w^\ast) = \bmu^\top \bsigma^{-1} \bmu \,.
\end{align*} 
\end{proof}

\subsubsection{Proof of Theorem~\ref{theorem:SR_converge_sig_single}}
\begin{proof}
Recall that we assume full knowledge of the expected return, i.e. $\hat\bmu = \bmu$ and set the regularization matrix as $\Q = \rho \bar\Q$ in this section. 

Let $\a = \b = \bsigma^{-1/2}\bmu$ and $\S = \rho \bsigma^{-1/2} \bar\Q \bsigma^{-1/2} - z \I$, we have 
\begin{align*}
    m(z) 
    = \a^\top \bigg( \frac{\X^\top \X}{n}+\S+z\I \bigg)^{-1}\b 
    = \bmu^\top \bigg(\frac{ \bsigma^{1/2}\X^\top \X \bsigma^{1/2}}{n}+\rho \bar\Q \bigg)^{-1} \bmu \,,
\end{align*}
where the elements of $\X$ are iid with zero mean, variance 1 and finite eighth order moment. 
By Lemma \ref{lemma:keylemma}, we can show that 
\begin{align*}
    \bmu^\top (\hat\bsigma + \rho \bar\Q)^{-1} \bmu 
    \xrightarrow{a.s} 
    \bmu^\top \bigg( \frac{\bsigma}{1+s(\rho)} + \rho \bar\Q \bigg)^{-1} \bmu \,.
\end{align*}
Similarly we can show that 
\begin{align*}
    \frac{1}{2 \omega} \xrightarrow{a.s}
    \frac{X_\tg - r - \rho \bmu^\top \bigg( \frac{\bsigma}{1+s(\rho)} + \rho \bar\Q \bigg)^{-1} \bar\Q \w}{\bmu^\top \bigg( \frac{\bsigma}{1+s(\rho)} + \rho \bar\Q \bigg)^{-1} \bmu}
    := \frac{1}{2 \omega^\ast} \,.
\end{align*}

Let $\a = \bsigma^{-1/2} \bmu, \b = \bsigma^{-1/2} (\frac{1}{2 \omega} \bmu + \rho \w)$ and $\S = \rho \bsigma^{-1/2} \bar\Q \bsigma^{-1/2} - z \I$, we have 
\begin{equation*}
    \begin{aligned}
        m(z) 
        &= \a^\top \bigg(\frac{\X^\top \X}{n}+\S+z\I\bigg)^{-1}\b
        = \bmu^\top \bsigma^{-1/2} \bigg(\frac{\X^\top \X}{n}+\S+z\I\bigg)^{-1} \bsigma^{-1/2} (\frac{1}{2 \omega} \bmu + \rho \w) 
        \\
        &= \bmu^\top \bigg(\frac{ \bsigma^{1/2}\X^\top \X  \bsigma^{1/2}}{n}+\rho \bar\Q \bigg)^{-1} (\frac{1}{2 \omega} \bmu + \rho \w) \,,
    \end{aligned}
\end{equation*}
where the elements of $\X$ are iid with zero mean, variance 1 and finite eighth order moment, then we have the convergence for any value of $\frac{1}{2 \omega}$.
\begin{align*}
    m(z) = \langle \bsigma^{-1}\bmu, \u^* \rangle_{\bsigma}
    \xrightarrow{a.s}  \bmu^\top \bsigma^{-1/2} \bigg( \frac{\I}{1+s(\rho)} + \rho\bsigma^{-1/2} \bar\Q \bsigma^{-1/2} \bigg)^{-1} \bsigma^{-1/2} (\frac{1}{2 \omega} \bmu + \rho \w) \,.
\end{align*}
And by continuous mapping theorem, we have 
\begin{align*}
    \langle \bsigma^{-1}\bmu, \u^* \rangle_{\bsigma} 
    \xrightarrow{a.s} 
    \bmu^\top \bigg( \frac{\bsigma}{1+s(\rho)}+\rho\I \bigg)^{-1} (\frac{1}{2 \omega^\ast} \bmu + \rho \w)\,.
\end{align*}

Next, by letting $\a = \b = \bsigma^{-1/2} (\frac{1}{2 \omega} \bmu + \rho \w)$, we have 
\begin{equation*}
    \begin{aligned}
        m_n'(z) &= -\a^\top (\frac{\X^\top \X}{n}+\S+z\I)^{-2}\b 
        \\
        &= - (\frac{1}{2 \omega} \bmu + \rho \w)^\top \bsigma^{-1/2}(\frac{\X^\top \X}{n}+\S+z\I)^{-2} \bsigma^{-1/2} (\frac{1}{2 \omega} \bmu + \rho \w) 
        \\
        &= - \|(\frac{1}{2 \omega} \bmu + \rho \w)^\top \bsigma^{-1/2}(\frac{\X^\top \X}{n}+\S+z\I)^{-1}\|_2^2 
        \\
        &= - \|(\frac{1}{2 \omega} \bmu + \rho \w)^\top (\frac{\X^\top \X \bsigma^{1/2}}{n} + \rho \bsigma^{-1/2} \bar\Q)^{-1}\|_2^2 
        \\
        &= - \|(\frac{1}{2 \omega} \bmu + \rho \w)^\top (\frac{\bsigma^{-1/2}\bsigma^{1/2}\X^\top \X \bsigma^{1/2}}{n} + \rho \bsigma^{-1/2} \bar\Q)^{-1}\|_2^2 
        \\
        &= - \|(\frac{1}{2 \omega} \bmu + \rho \w)^\top (\frac{\bsigma^{1/2}\X^\top \X \bsigma^{1/2}}{n}+\rho \bar\Q )^{-1}\|_{\bsigma}^2 \,,
    \end{aligned}
\end{equation*}
and we have 
\begin{equation*}
    \begin{aligned}
        & \frac{\rm{d} [\a^\top (\frac{\I}{1+s_n(z)}+\S+z\I )^{-1}\b]}{\rm{d} z} 
        = 
        -\a^\top (\frac{\I}{1+s_n(z)}+\S+z\I )^{-1} (\I - \frac{s'(z)}{(1+s(z))^2}\I) (\frac{\I}{1+s_n(z)}+\S+z\I )^{-1}\b \,.
    \end{aligned}
\end{equation*}

Here $\tilde s(\rho)$ is given as the solution of $\tilde s(\rho) = (\frac{\tilde s(\rho)}{(1+s(\rho))^2} - 1) \frac{c}{p} \tr [\bsigma^2 (\frac{\bsigma}{1+s(\rho)}+\rho \bar\Q)^{-2}]$. 
In order for the above convergence to hold, we need to assume $\|\bsigma\|$ is bounded. 
Once we have the almost sure convergence of the numerator and denominator, we conclude
\begin{align*}
    \SR_\infty(\rho\bar\Q, \w) = \frac{ \Big(1 - \frac{\tilde s(\rho)}{(1+s(\rho))^2} \Big)^{-1/2} \bmu^\top \Big(\frac{\bsigma}{1+s(\rho)} + \rho \bar\Q \Big)^{-1} (\frac{\bmu}{2\omega^\ast}  + \rho \bar\Q \w)}{\sqrt{(\frac{\bmu}{2\omega^\ast}  + \rho \bar\Q \w)^\top \Big(\frac{\bsigma}{1+s(\rho)}+\rho\bar\Q \Big)^{-1} \bsigma \Big(\frac{\bsigma}{1+s(\rho)} + \rho \bar\Q \Big)^{-1} (\frac{\bmu}{2\omega^\ast}  + \rho \bar\Q \w)}} \,.
\end{align*}
Let $\kappa_\rho = 1 - \frac{\tilde s(\rho)}{(1+s(\rho))^2}$, $\A_\rho = \frac{\bsigma}{1+s(\rho)} + \rho \bar\Q$, $\d_\rho = \frac{\bmu}{2\omega^\ast}  + \rho \bar\Q \w$, we have 
\begin{align*}
    \SR_\infty(\rho\bar\Q, \w) = \kappa_\rho^{-1/2} \cdot \frac{\bmu^\top \A_\rho^{-1} \d_\rho}{\sqrt{\d_\rho^\top \A_\rho^{-1} \bsigma \A_\rho^{-1} \d_\rho}}
    \,,
\end{align*}
where $s(\rho)$ is the solution of $s(\rho) = \frac{c}{p} \tr(\bsigma \A_\rho^{-1})$, 
$\tilde s(\rho)$ is the solution of $\tilde s(\rho) = - \kappa_\rho \frac{c}{p} \tr(\bsigma^2 \A_\rho^{-2})$ and $\frac{1}{2\omega^\ast} = \frac{X_\tg - r - \rho \bmu^\top \A_\rho)^{-1} \bar\Q \w}{\bmu^\top \A_\rho^{-1} \bmu}$. 
\end{proof}

\subsection{Proofs for Section 3.2}
\subsubsection{Proof of Lemma~\ref{lemma:param_mu}}
\begin{proof}
For $k=T-1$, we have 
\begin{equation*}     \begin{aligned}
    \hat{{\D}}_{T-1}^{-1} &= (\bsigma + \Q_{T-1} )^{-1} - \frac{(\bsigma + \Q_{T-1} )^{-1} \hat\bmu \hat\bmu^\top (\bsigma + \Q_{T-1} )^{-1}}{1+\hat\bmu^\top (\bsigma + \Q_{T-1} )^{-1} \hat\bmu}\\
    \hat{{a}}_{T-1} &= \frac{r^2 + (\hat\bmu^\top (\bsigma + \Q_{T-1} )^{-1} \Q_{T-1} \w_{T-1} )^2}{1+\hat\bmu^\top (\bsigma + \Q_{T-1} )^{-1} \hat\bmu} + \w_{T-1}^\top \Q_{T-1} \w_{T-1} - \w_{T-1}^\top \Q_{T-1} (\bsigma + \Q_{T-1} )^{-1} \Q_{T-1} \w_{T-1}\\
    \hat{{b}}_{T-1} &= \frac{r + \hat\bmu^\top (\bsigma + \Q_{T-1} )^{-1} \Q_{T-1} \w_{T-1}}{1+\hat\bmu^\top (\bsigma + \Q_{T-1} )^{-1} \hat\bmu}\\ 
    \hat{{c}}_{T-1} &= - \frac{\hat\bmu^\top (\bsigma + \Q_{T-1} )^{-1} \hat\bmu}{1+\hat\bmu^\top (\bsigma + \Q_{T-1} )^{-1} \hat\bmu}
\end{aligned} \end{equation*}
Similar to the proof in Appendix \ref{prof:theorem:single_mu_SR}, by using Borel-Cantelli Lemma and continuous mapping theorem, we can show that the following almost surely convergence holds for the parameters at $T-1$.  
\begin{equation*}     \begin{aligned}
    \hat{{a}}_{T-1} 
    \xrightarrow{a.s} a_{T-1}^\ast = & \frac{( r+ \bmu^\top (\bsigma + \Q_{T-1} )^{-1} \Q_{T-1} \w_{T-1} )^2}{ 1 + \bmu^\top (\bsigma + \Q_{T-1})^{-1} \bmu + \tr((\bsigma + \Q_{T-1})^{-1} \bsigma/n)} + \w_{T-1}^\top \Q_{T-1} \w_{T-1} 
    \\
    &~ - \w_{T-1}^\top \Q_{T-1}  (\bsigma + \Q_{T-1})^{-1} \Q_{T-1} \w_{T-1}\\
    \hat{{b}}_{T-1} 
    \xrightarrow{a.s} b_{T-1}^\ast = & \frac{r + \bmu^\top (\bsigma + \Q_{T-1})^{-1} \Q_{T-1}\w_{T-1}}{1 + \bmu^\top (\bsigma + \Q_{T-1})^{-1} \bmu + \tr((\bsigma + \Q_{T-1})^{-1} \bsigma/n)}\\ 
    \hat{{c}}_{T-1} 
    \xrightarrow{a.s} c_{T-1}^\ast = & \frac{1}{1 + \bmu^\top (\bsigma + \Q_{T-1})^{-1} \bmu + \tr((\bsigma + \Q_{T-1})^{-1} \bsigma/n)} - 1
\end{aligned} \end{equation*} 
Following the induction, we assume that $\hat{a}_{k+1}$ will converge to $a_{k+1}^\ast$, then for $\hat{a}_k$ we have 
\begin{equation*}     \begin{aligned}
    \hat{{a}}_k 
    &= \hat{{a}}_{k+1} r^2 - \hat{{a}}_{k+1}^2 r^2 \hat\bmu^{\top} \hat{{\D}}_k^{-1} \hat\bmu + \w_k^\top \Q_k \w_k - \w_k^\top \Q_k \hat{{\D}}_k^{-1} \Q_k \w_k + 2 r_k \hat a_{k+1} \w_k^\top \Q_k \hat\D_k^{-1} \hat\bmu_k \\
    &= \hat a_{k+1}^{-1} \frac{(\hat a_{k+1} r + \w_k^\top \Q_k \left( \bsigma + \hat a_{k+1}^{-1} \Q_k \right)^{-1} \hat\bmu )^2 }{1+\hat\bmu_k^\top \left( \bsigma + \hat a_{k+1}^{-1} \Q_k  \right)^{-1} \hat\bmu} + \w_k^\top \Q_k \w_k - \hat a_{k+1}^{-1} \w_k^\top \Q_k \left( \bsigma + \hat a_{k+1}^{-1} \Q_k \right)^{-1} \Q_k \w_k 
    \\
    & \xrightarrow{a.s} a_k^\ast \,,   
\end{aligned} \end{equation*}
where 
\begin{equation*}
\begin{aligned}
    a_k^\ast 
    & := (a_{k+1}^\ast)^{-1} \frac{(a_{k+1}^\ast r + \w_k^\top \Q_k \left( \bsigma + (a_{k+1}^\ast)^{-1} \Q_k \right)^{-1} \bmu )^2 }{1+\bmu_k^\top \left( \bsigma + (a_{k+1}^\ast)^{-1} \Q_k  \right)^{-1} \bmu + \tr((\bsigma + (a_{k+1}^\ast)^{-1} \Q_k )^{-1} \bsigma/n)} 
    \\
    & ~ + \w_k^\top \Q_k \w_k - (a_{k+1}^\ast)^{-1} \w_k^\top \Q_k \left( \bsigma + (a_{k+1}^\ast)^{-1} \Q_k \right)^{-1} \Q_k \w_k \,.
\end{aligned}
\end{equation*}

The convergence holds because for any given value of $\hat a_{k+1}$, the term $\hat\bmu^\top (\bsigma + \hat{a}_{k+1}^{-1} \Q_k)^{-1} \hat\bmu$ will converge according to the proof in Appendix \ref{prof:theorem:single_mu_SR}. 
Then we can apply continuous mapping theorem since we have $\hat a_{k+1} \xrightarrow{a.s} a_{k+1}^\ast$. 
Similarly, the following almost surely convergence holds at $k$.  
\begin{equation*}
\begin{aligned}
    \hat{b}_k \xrightarrow{a.s} b_k^\ast 
    &= b_{k+1}^\ast \frac{r + (a_{k+1}^\ast)^{-1} \bmu^\top (\bsigma + (a_{k+1}^\ast)^{-1} \Q_k )^{-1} \Q_k \w_k}{1 + \bmu^\top (\bsigma + (a_{k+1}^\ast)^{-1} \Q_k )^{-1} \bmu + \tr((\bsigma + (a_{k+1}^\ast)^{-1} \Q_k )^{-1} \bsigma/n)} \,,
    \\
    \hat{c}_k \xrightarrow{a.s} c_k^\ast 
    &= c_{k+1}^\ast + \frac{(b_{k+1}^\ast)^2}{a_{k+1}^{\ast}}   (\frac{1}{1 + \bmu^\top (\bsigma + (a_{k+1}^\ast)^{-1} \Q_k )^{-1} \bmu + \tr((\bsigma + (a_{k+1}^\ast)^{-1} \Q_k )^{-1} \bsigma/n)} - 1) \,.
\end{aligned} 
\end{equation*}
We can define $\bdelta_k^{\ast} = \bsigma + \Q_k /a_{k+1}^{\ast} $, and show that 
\begin{align*}
        \hat{a}_k \xrightarrow{a.s} 
        & a_k^\ast
        = \frac{  (a_{k+1}^{\ast})^{-1}   (a_{k+1}^\ast r + \w_k^\top \Q_k ( \bdelta_k^{\ast} )^{-1} \bmu )^2 }{1+\bmu^\top ( \bdelta_k^{\ast} )^{-1} \bmu + \tr(( \bdelta_k^{\ast} )^{-1} \bsigma/n)} + ( a_{k+1}^{\ast} )^{-1} \w_k^\top \Q_k \Big(a_{k+1}^{\ast} \Q_k^{-1} - ( \bdelta_k^{\ast} )^{-1} \Big) \Q_k \w_k \,,
        \\
        \hat{b}_k \xrightarrow{a.s} 
        & b_k^\ast 
        = {{b}}_{k+1}^\ast \frac{r + (a_{k+1}^{\ast})^{-1} \bmu^\top ( \bdelta_k^{\ast} )^{-1} \Q_k\w_k}{1 + \bmu^\top ( \bdelta_k^{\ast} )^{-1} \bmu + \tr( (\bdelta_k^{\ast}  )^{-1}  \bsigma/n)} \,,
        \\
        \hat{c}_k \xrightarrow{a.s} 
        & c_k^\ast 
        = c_{k+1}^\ast + \frac{(b_{k+1}^{\ast})^2}{a_{k+1}^{\ast}} \frac{\bmu^\top ( \bdelta_k^{\ast} )^{-1} \bmu + \tr( (\bdelta_k^{\ast}  )^{-1}  \bsigma/n)}{1 + \bmu^\top ( \bdelta_k^{\ast} )^{-1} \bmu + \tr( (\bdelta_k^{\ast}  )^{-1}  \bsigma/n)} 
    \end{align*}
\end{proof}

\subsubsection{Proof of Theorem~\ref{theorem:multi_SR_mu}}
\begin{proof}
Recall the parameters we defined in Proposition \ref{prop:eff_frontier}, we can show that 
\begin{equation*}
    \begin{aligned}
    \bmu^\top \hat{\D}_k^{-1} \hat\bmu 
    & \xrightarrow{a.s} (a_{k+1}^\ast)^{-1} \frac{\bmu^\top (\bsigma + (a_{k+1}^\ast)^{-1} \Q_k)^{-1} \bmu}{1 + \bmu^\top (\bsigma + (a_{k+1}^\ast)^{-1} \Q_k)^{-1} \bmu+ \tr((\bsigma + (a_{k+1}^\ast)^{-1} \Q_k )^{-1} \bsigma/n)}\\
    \bmu^\top \hat{\D}_k^{-1} \Q_k \w_k
    & \xrightarrow{a.s} (a_{k+1}^\ast)^{-1} \frac{\bmu^\top (\bsigma + (a_{k+1}^\ast)^{-1} \Q_k)^{-1} \Q_k \w_k}{1 + \bmu^\top (\bsigma + (a_{k+1}^\ast)^{-1} \Q_k)^{-1} \bmu+ \tr((\bsigma + (a_{k+1}^\ast)^{-1} \Q_k )^{-1} \bsigma/n)}\\
    \hat\bmu^\top \hat{\D}_k^{-1} \bsigma \hat{\D}_k^{-1} \hat\bmu
    & = \hat a_{k+1}^{-2} \hat\bmu^\top (\bsigma + \hat a_{k+1}^{-1} \Q_k  + \hat\bmu \hat\bmu^\top)^{-1} \bsigma (\bsigma + \hat a_{k+1}^{-1} \Q_k  + \hat\bmu \hat\bmu^\top)^{-1} \hat\bmu \\
    &= \hat a_{k+1}^{-2} \frac{\hat\bmu^\top (\bsigma + \hat a_{k+1}^{-1} \Q_k  )^{-1} \bsigma (\bsigma + \hat a_{k+1}^{-1} \Q_k  )^{-1} \hat\bmu}{(1 + \hat\bmu^\top (\bsigma + \hat a_{k+1}^{-1} \Q_k  )^{-1} \hat\bmu)^2} \\
    & \xrightarrow{a.s} (a_{k+1}^\ast)^{-2} \frac{\bmu^\top (\bsigma + (a_{k+1}^\ast)^{-1} \Q_k  )^{-1} \bsigma (\bsigma + (a_{k+1}^\ast)^{-1} \Q_k  )^{-1} \bmu + \tr((\bsigma + (a_{k+1}^\ast)^{-1} \Q_k  )^{-2} \bsigma^2/n)}{(1 + \bmu^\top (\bsigma + (a_{k+1}^\ast)^{-1} \Q_k  )^{-1} \bmu + \tr((\bsigma + (a_{k+1}^\ast)^{-1} \Q_k  )^{-1} \bsigma/n))^2}
\end{aligned}
\end{equation*}
where $C$ is a constant depending on $\Q_k$, $\bsigma$ and $\bmu$. 
Then according to the definition in Proposition \ref{prop:eff_frontier}, we have the following almost surely convergence and the limiting values can be achieved by plugging in each component. 
\begin{equation*}
\begin{aligned}
    \hat\lambda_k 
    &= \frac{\hat b_0 - X_\tg}{\hat c_0} \hat b_{k+1} \bmu^\top \hat{\D}_k^{-1} \hat\bmu \xrightarrow{a.s} \lambda_k^\ast\\
    \hat\eta_k 
    &= r- r \hat a_{k+1} \bmu^\top \hat{\D}_k^{-1}  \hat\bmu + \bmu^\top \hat{\D}_k^{-1} \Q_k\w_k \xrightarrow{a.s} \eta_k^\ast\\
    \hat\nu_k 
    &= \E[X_k]^2 \bigg[ r^2 \hat{a}_{k+1}^2  \hat\bmu^\top \hat{\D}_k^{-1} \bsigma \hat{\D}_k^{-1}  \hat\bmu + \w_k^\top \Q_k \hat{\D}_k^{-1} \bsigma \hat{\D}_k^{-1} \Q_k \w_k - 2 r \hat{a}_{k+1}  \hat\bmu^\top \hat{\D}_k^{-1} \bsigma \hat{\D}_k^{-1} \Q_k \w_k \bigg]\\
    & ~~~~~ + \E[X_k] \bigg[ 2 \frac{\hat{b}_0 - X_\tg}{\hat{c}_0} \hat{b}_{k+1}  \hat\bmu^\top \hat{\D}_k^{-1} \bsigma \hat{\D}_k^{-1} \Q_k \w_k - 2 \frac{\hat{b}_0 - X_\tg}{\hat{c}_0}\hat{b}_{k+1} r \hat{a}_{k+1} \hat\bmu^\top \hat{\D}_k^{-1} \bsigma \hat{\D}_k^{-1}  \hat\bmu \bigg]\\
    & ~~~~~ + (\frac{\hat{b}_0 - X_\tg}{\hat{c}_0} \hat{b}_{k+1})^2  \hat\bmu^\top \hat{\D}_k^{-1} \bsigma \hat{\D}_k^{-1}  \hat\bmu 
    \xrightarrow{a.s} \nu_k^\ast\\
    \hat{g}_k 
    &= r^2 \hat{a}_{k+1}^2  \hat\bmu^\top \hat{\D}_k^{-1} \bsigma \hat{\D}_k^{-1}  \hat\bmu + \w_k^\top \Q_k \hat{\D}_k^{-1} \bsigma \hat{\D}_k^{-1} \Q_k \w_k - 2 r \hat{a}_{k+1}  \hat\bmu^\top \hat{\D}_k^{-1} \bsigma \hat{\D}_k^{-1} \Q_k \w_k\\
    & ~~~~ + (r- r \hat{a}_{k+1} \bmu^\top \hat{\D}_k^{-1}  \hat\bmu + \bmu^\top \hat{\D}_k^{-1} \Q_k \w_k)^2 
    \xrightarrow{a.s} g_k^\ast,
\end{aligned} 
\end{equation*}
which leads to the convergence of average Sharpe ratio. 
\begin{align*}
        \SR_\infty(T,\Q_k,\w_k) 
        = \frac{\E[\hat X_T]_\infty - r^T}{\sqrt{T \cdot \Var(\hat X_T)_\infty}} \,,
    \end{align*}
where $\E[\hat X_T]_\infty = \prod_{i=0}^{T-1} \eta^\ast_i + \sum_{i=0}^{T-1} (\lambda^\ast_{i} \prod_{j=i+1}^{T-1} \eta^\ast_{j})$ and $\Var(\hat X_T)_\infty = \sum_{i=0}^{T-1} \nu_i^\ast \prod_{j=i+1}^{T-1}{g}_j^\ast$. 
\end{proof}

\subsubsection{Proof of Corollary~\ref{coro:multi_SR_mu_w_0}}
\label{proof:multi_SR_mu}
\begin{proof}
Recall that in Corollary \ref{coro:os_SR_w=0}, we have derived the out-of-sample Sharpe ratio when we set $\Q_k = \rho \hat a_{k+1} \Q$ and $\w_k = 0$, which is 
\begin{align*}
    \SR^2(\hat{X}_T) = \frac{1}{T} \frac{(\hat\eta^T -r^T)^2}{((\hat{\Gamma} + \hat\eta^2)^T - \hat\eta^{2T})}, 
\end{align*}
$\hat\eta = r\Big(1-\frac{\bmu^\top(\hat\bsigma + \rho \bar\Q)^{-1}\hat\bmu}{1+\hat\bmu^\top (\hat\bsigma + \rho \bar\Q)^{-1} \hat\bmu} \Big)$ is the same $\hat\eta_k$ in Proposition~\ref{prop:eff_frontier} and $\hat{\Gamma} = r^2 \frac{\hat\bmu^\top (\hat{\bsigma} + \rho \bar\Q)^{-1} \bsigma (\hat{\bsigma} + \rho \bar\Q)^{-1} \hat\bmu}{(1+\hat\bmu^\top (\hat{\bsigma} + \rho \bar\Q)^{-1} \hat\bmu)^2}$.
In this Corollary \ref{coro:multi_SR_mu_w_0}, we assume that we know the exact value of $\bsigma$ and use sample estimator $\hat\bmu$. 
Then we can write the out-of-sample Sharpe ratio as  
\begin{equation*}     \begin{aligned}
    \SR^2(\hat{X}_T)
    &= \frac{1}{T} \frac{ \left( (r (1-\frac{\bmu^\top(\hat\bsigma + \rho \Q)^{-1}\hat\bmu}{1+\hat\bmu^\top (\hat\bsigma + \rho \Q)^{-1} \hat\bmu}))^T - r^T \right)^2}{\left( r^2 \frac{\hat\bmu^\top ({\bsigma} + \rho \Q)^{-1} \bsigma ({\bsigma} + \rho \Q)^{-1} \hat\bmu}{(1+\hat\bmu^\top ({\bsigma} + \rho \Q)^{-1} \hat\bmu)^2} + r^2 (1-\frac{\bmu^\top(\bsigma + \rho \Q)^{-1}\hat\bmu}{1+\hat\bmu^\top (\bsigma + \rho \Q)^{-1} \hat\bmu})^2 \right)^T - r^{2T} (1-\frac{\bmu^\top(\bsigma + \rho \Q)^{-1}\hat\bmu}{1+\hat\bmu^\top (\bsigma + \rho \Q)^{-1} \hat\bmu})^{2T}} \\
    &= \frac{1}{T} \frac{((1+\hat\bmu^\top (\hat\bsigma + \rho \Q)^{-1} \hat\bmu - \bmu^\top (\hat\bsigma + \rho \Q)^{-1} \hat\bmu)^T - (1+\hat\bmu^\top (\hat\bsigma + \rho \Q)^{-1} \hat\bmu)^T )^2}{(\hat\bmu^\top (\hat{\bsigma} + \rho \Q)^{-1} \bsigma (\hat{\bsigma} + \rho \Q)^{-1} \hat\bmu + (1+\varepsilon^\top(\hat\bsigma + \rho \Q)^{-1}\hat\bmu)^2)^T - (1+\varepsilon^\top(\hat\bsigma + \rho \Q)^{-1}\hat\bmu)^{2T}} 
\end{aligned} \end{equation*}
And we set $\Q$ as the true covariance matrix $\bsigma$, we can simplify the expression of Sharpe ratio. 
\begin{equation*}     \begin{aligned}
    \SR^2(\hat{X}_T) 
    &= \frac{1}{T} \frac{[(\rho + 1 + \hat\bmu^\top \bsigma^{-1} \hat\bmu)^T - (\rho + 1 + \varepsilon^\top \bsigma^{-1} \hat\bmu)^T)]^2}{( \hat\bmu^\top \bsigma^{-1} \hat\bmu + (\rho + 1 + \varepsilon^\top \bsigma^{-1} \hat\bmu)^2)^T - (\rho + 1 + \varepsilon^\top \bsigma^{-1} \hat\bmu)^{2T}} 
\end{aligned} \end{equation*}
Use the same techniques in Appendix \ref{prof:theorem:single_mu_SR}, we can show that $\hat\bmu^\top \bsigma^{-1} \hat\bmu \xrightarrow{a.s} \bmu^\top \bsigma^{-1} \bmu + p/n$ and $\varepsilon^\top \bsigma^{-1} \hat\bmu \xrightarrow{a.s} p/n$. 
Then we can derive the limiting value of average Sharpe ratio, 
\begin{equation*}     \begin{aligned}
    \SR^2(\hat{X}_T) 
    \xrightarrow{a.s} \frac{1}{T} \frac{[(\bmu^\top \bsigma^{-1} \bmu + \rho + 1 + p/n)^T - (\rho + 1 + p/n)^T)]^2}{( \bmu^\top \bsigma^{-1} \bmu + p/n + (\rho + 1 + p/n)^2)^T - (\rho + 1 + p/n)^{2T}} 
    := \SR_\infty^2(T, \rho, \BFzero)
\end{aligned} \end{equation*}
Following Theorem~\ref{theorem:single_mu_SR}, let $e_{\rho,\bmu} = e_\bmu + \rho + 1$ and $e_\bmu = p/n$, we have 
\begin{align*}
    \SR_\infty(T, \rho, \BFzero) = \frac{T^{-1/2} ((\bmu^\top \bsigma^{-1} \bmu + e_{\rho,\bmu})^T - e_{\rho,\bmu}^T)}{\sqrt{ ( \bmu^\top \bsigma^{-1} \bmu + p/n + e_{\rho,\bmu}^2)^T - e_{\rho,\bmu}^{2T} }}\,.
\end{align*}
Now, we would like to check how the Sharpe ratio change as the parameter $\rho$ differs. 
We take the derivative of Sharpe ratio respect to $\rho$, which has the same sign of the following term
\begin{equation*}     \begin{aligned} 
    & [(\rho + 1 + \bmu^\top \bsigma^{-1} \bmu + p/n)^{T-1} - (\rho + 1 + p/n)^{T-1})] 
    \\& \times [( \bmu^\top \bsigma^{-1} \bmu + p/n + (\rho + 1 + p/n)^2)^T - (\rho + 1 + p/n)^{2T}] 
    \\
    - & [(\rho + 1 + \bmu^\top \bsigma^{-1} \bmu + p/n)^T - (\rho + 1 + p/n)^T)] 
    \\
    & \times [( \bmu^\top \bsigma^{-1} \bmu + p/n + (\rho + 1 + p/n)^2)^{T-1} - ((\rho + 1 + p/n)^2)^{T-1}] (\rho + 1 + p/n)
\end{aligned} \end{equation*}
If $p/n=0$ and $\rho =  0$, this term will be equals to $0$. 
Then, $\frac{\partial (\text{SR}_T^\ast)^2}{\partial \rho}|_{\rho=0}  =0$, which implies for low-dimensional portfolio, we can achieve optimal asymptotic Sharpe ratio without regularization. 
Next, we would like to check the sign of $\frac{\partial \text{SR}_T^\ast}{\partial \rho}|_{\rho=0}$ when $p/n >0$. 
It is equivalent to check the sign of the following term 
\begin{align*}
& (1+p/n)^{3T-1)} [(\frac{\bmu^\top \bsigma^{-1} \bmu}{1 + p/n} + 1)^{T-1}-1] [(\frac{\bmu^\top \bsigma^{-1} \bmu + p/n}{(1 + p/n)^2} + 1)^{T-1}-1] 
\\
& \times \left[ \frac{(\frac{\bmu^\top \bsigma^{-1} \bmu + p/n}{(1 + p/n)^2} + 1 )^{T}-1}{(\frac{\bmu^\top \bsigma^{-1} \bmu + p/n}{(1 + p/n)^2} + 1)^{T-1}-1} - \frac{(\frac{\bmu^\top \bsigma^{-1} \bmu}{1 + p/n} + 1)^{T}-1}{(\frac{\bmu^\top \bsigma^{-1} \bmu}{1 + p/n} + 1)^{T-1}-1} \right] \,.
\end{align*}
We can show that, for $p/n > 0$, the derivative $\frac{\partial \text{SR}_T^\ast}{\partial \rho}\big|_{\rho=0}$ is positive when $\bmu^\top \bsigma^{-1} \bmu < 1$, and non-positive otherwise. 
\end{proof}

\subsubsection{Special Case in Section \ref{sec3.2.1}}
\label{ssse:est_mu_Q_0}
~

If we assume $\Q \to \mathbf{0}$ and $\Q\w \to \tilde\w$, the limiting values of the parameters will be 
\begin{align*}
    a_k^\ast
    &= 
    (a_{k+1}^\ast)^{-1} \left[ \frac{(a_{k+1}^\ast r + \tilde\w^\top \bsigma^{-1} \bmu )^2 }{1+\bmu_k^\top \bsigma^{-1} \bmu + p/n} - \tilde\w^\top \bsigma^{-1} \tilde\w \right]
    \\
    b_k^\ast 
    &= b_{k+1}^\ast \frac{r + (a_{k+1}^\ast)^{-1} \bmu^\top \bsigma^{-1} \tilde\w}{1 + \bmu^\top \bsigma^{-1} \bmu + p/n } 
    \\
    c_k^\ast 
    &= c_{k+1}^\ast + \frac{(b_{k+1}^\ast)^2}{a_{k+1}^{\ast}}   (\frac{1}{1 + \bmu^\top \bsigma^{-1} \bmu + p/n} - 1) 
\end{align*}
Recall the parameters we defined in Proposition \ref{prop:eff_frontier}, we can show that 
\begin{equation*}
    \begin{aligned}
    \bmu^\top \hat{\D}_k^{-1} \hat\bmu 
    & \xrightarrow{a.s} (a_{k+1}^\ast)^{-1} \frac{\bmu^\top \bsigma^{-1} \bmu}{1 + \bmu^\top \bsigma^{-1} \bmu + p/n}
    \\
    \bmu^\top \hat{\D}_k^{-1} \tilde\w
    & \xrightarrow{a.s} (a_{k+1}^\ast)^{-1} \frac{\bmu^\top \bsigma^{-1} \tilde\w}{1 + \bmu^\top \bsigma^{-1} \bmu + p/n}
    \\
    \hat\bmu^\top \hat{\D}_k^{-1} \bsigma \hat{\D}_k^{-1} \hat\bmu
    & \xrightarrow{a.s} (a_{k+1}^\ast)^{-2} \frac{\bmu^\top \bsigma^{-1} \bmu + p/n}{(1 + \bmu^\top \bsigma^{-1} \bmu + p/n)^2}
    \\
    \bmu^\top \hat{\D}_k^{-1} \bsigma \hat{\D}_k^{-1} \hat\bmu
    &= \hat a_{k+1}^{-2} \bmu^\top (\bsigma^{-1} - \frac{\bsigma^{-1} \hat\bmu \hat\bmu^\top \bsigma^{-1}}{1 + \hat\bmu^\top \bsigma^{-1} \hat\bmu}) \bsigma (\bsigma^{-1} - \frac{\bsigma^{-1} \hat\bmu \hat\bmu^\top \bsigma^{-1}}{1 + \hat\bmu^\top \bsigma^{-1} \hat\bmu }) \hat\bmu
    \xrightarrow{a.s} (a_{k+1}^\ast)^{-2} \frac{\bmu^\top \bsigma^{-1} \bmu }{(1 + \bmu^\top \bsigma^{-1} \bmu + p/n)^2}
    \\
    \bmu^\top \hat{\D}_k^{-1} \bsigma \hat{\D}_k^{-1} \bmu
    &= \hat a_{k+1}^{-2} \bmu^\top (\bsigma^{-1} - \frac{\bsigma^{-1} \hat\bmu \hat\bmu^\top \bsigma^{-1}}{1 + \hat\bmu^\top \bsigma^{-1} \hat\bmu}) \bsigma (\bsigma^{-1} - \frac{\bsigma^{-1} \hat\bmu \hat\bmu^\top \bsigma^{-1}}{1 + \hat\bmu^\top \bsigma^{-1} \hat\bmu }) \bmu
    \\
    &\xrightarrow{a.s} (a_{k+1}^\ast)^{-2} 
    \left[ \bmu^\top \bsigma^{-1} \bmu - 2 \frac{(\bmu^\top \bsigma^{-1} \bmu)^2}{1 + \bmu^\top \bsigma^{-1} \bmu + p/n}  + \frac{(\bmu^\top \bsigma^{-1} \bmu)^2 (\bmu^\top \bsigma^{-1} \bmu + p/n)}{(1 + \bmu^\top \bsigma^{-1} \bmu + p/n)^2}\right]
\end{aligned}
\end{equation*}

Then according to the definition in Proposition \ref{prop:eff_frontier}, we have the following almost surely convergence and the limiting values can be achieved by plugging in the components defined in \ref{prop:eff_frontier}. 
By checking the FOC respective to $\tilde\w$, we can show that $\tilde\w^\ast \propto \bmu$. 
So, we let $\tilde\w = \alpha \bmu$, and check how $\alpha$ affect the Sharpe ratio. 

\subsubsection{Proof of Lemma~\ref{lemma:sig_converge}} \label{proof:lemma3.2}
\begin{proof}
    Note that We assume no estimation error in the expected return return, i.e. $\hat\bmu = \bmu$.
    Then, under the same assumptions in Proposition \ref{prop:eff_frontier}, for deterministic $\a, \b \in \mathbb{R}^p$ with bounded norm, we want to show the convergence of $\x^\top \hat\D_k^{-1} \y$ and $\x^\top \hat\D_k^{-1} \bsigma \hat\D_k^{-1} \y$. 
    Since these convergence will depend on $\hat a_k$, we start with $\hat a_T=1$, and verify the almost sure convergence of $\hat a_{T-1}$. 
    By definition, we have 
    \begin{equation*}     
        \begin{aligned}
            \hat a_{T-1} = r^2  + \w_{T-1}^\top \Q_{T-1} \w_{T-1} - r^2 \bmu^\top \hat\D_{T-1}^{-1} \bmu - \w_{T-1}^\top \Q_{T-1} \hat\D_{T-1}^{-1} \Q_{T-1} \w_{T-1} + 2 r \w_{T-1}^\top \Q_{T-1} \hat\D_{T-1}^{-1} \bmu 
        \end{aligned} 
    \end{equation*}
    Let $\S_{T-1} = \bsigma^{-1/2} (\bmu \bmu^\top + \Q_{T-1}) \bsigma^{-1/2} - z_{T-1}\I$, $a=\bsigma^{-1/2} \x$ and $a=\bsigma^{-1/2} \y$, where $z_{T-1}$ should be small enough to guarantee the positive definiteness of $\S_{T-1}$. 
    According to lemma \ref{lemma:keylemma}, we have 
    \begin{equation*}    
        \begin{aligned}
        \x^\top \hat\D_{T-1}^{-1} \y
        & = \x^\top (\hat{\bsigma} + \S_{T-1} + z_{T-1} \I)^{-1} \y 
        \xrightarrow{a.s} 
        \x^\top \big( \frac{\bsigma}{1+s_{T-1}}+\bmu\bmu^\top + \Q_{T-1} \big)^{-1} \y 
        := C_{\x,\y,T-1}^\ast
        \end{aligned} 
    \end{equation*}
    where $s_{T-1}$ is the solution of $s_{T-1}=\frac{c}{p}\tr \bsigma  [( \frac{\bsigma}{1+s_{T-1}} +\bmu \bmu^\top + \Q_{T-1} )^{-1} ]$. 
    By taking values of $\x$ and $\y$ from $\bmu$ and $\Q_{T-1}\w_{T-1}$, we can show that 
    \[
    \hat{a}_{T-1} \xrightarrow{a.s} a_{T-1}^\ast 
    = r^2 + \w_{T-1}^\top \Q_{T-1} \w_{T-1} - r^2 C_{\bmu,\bmu,T-1}^\ast + 2 r C_{\bmu,\Q_{T-1}\w_{T-1},T-1}^\ast -  C_{\Q_{T-1}\w_{T-1},\Q_{T-1}\w_{T-1},T-1}^\ast
    \]
    Following the induction, we suppose that $\hat{a}_t$ will almost surely converge to $a_t^\ast$ for $t=k+1, \ldots, T$. 
    For the remaining problems, we start with the convergence of $\x^\top \hat\D_k^{-1} \y$. 
    First, we care about the difference
    \begin{equation*}     \begin{aligned}
          & (\hat{\bsigma}+\bmu\bmu^\top + \frac{1}{{a}_{k+1}^\ast} \Q_k)^{-1} - (\hat{\bsigma}+\bmu\bmu^\top + \frac{1}{\hat{a}_{k+1}} \Q_k)^{-1} \\
        = & (\hat{\bsigma}+\bmu\bmu^\top + \frac{1}{{a}_{k+1}^\ast} \Q_k)^{-1} (\frac{1}{a_{k+1}} \Q_k - \frac{1}{a_{k+1}^\ast} \Q_k)  (\hat{\bsigma}+\bmu\bmu^\top + \frac{1}{\hat{a}_{k+1}} \Q_k)^{-1} \\
        = & (\frac{1}{\hat{a}_{k+1}}- \frac{1}{a_{k+1}^\ast}) (\hat{\bsigma}+\bmu\bmu^\top + \frac{1}{a_{k+1}^\ast} \Q_k)^{-1} \Q_k (\hat{\bsigma}+\bmu\bmu^\top + \frac{1}{\hat{a}_{k+1}} \Q_k)^{-1}
    \end{aligned} \end{equation*}
    Then,
    \begin{equation*}     \begin{aligned}
    	 & \Big| \x^\top [(\hat{\bsigma}+\bmu\bmu^\top + \frac{1}{{a}_{k+1}^\ast} \Q_k)^{-1} - (\hat{\bsigma}+\bmu\bmu^\top + \frac{1}{\hat{a}_{k+1}} \Q_k)^{-1}] \y \Big| 
         \\
    	 = & \Big| (\frac{1}{\hat{a}_{k+1}} - \frac{1}{a_{k+1}^\ast}) \x^\top (\hat{\bsigma}+\bmu\bmu^\top + \frac{1}{a_{k+1}^\ast} \Q_k)^{-1} \Q_k (\hat{\bsigma}+\bmu\bmu^\top + \frac{1}{\hat{a}_{k+1}} \Q_k)^{-1} \y \Big| 
         \\
    	 \leq & \Big| \frac{1}{\hat{a}_{k+1}}- \frac{1}{a_{k+1}^\ast} \Big| \cdot \Big\| \x^\top (\hat{\bsigma}+\bmu\bmu^\top + \frac{1}{a_{k+1}^\ast} \Q_k)^{-1} \Q_k (\hat{\bsigma} + \bmu\bmu^\top + \frac{1}{\hat{a}_{k+1}} \Q_k)^{-1} \y \Big\| 
         \\
    	 \leq & \lambda_{\max}(\Q_k) \cdot \Big| \frac{1}{\hat{a}_{k+1}}- \frac{1}{a_{k+1}^\ast} \Big| \cdot \Big| \bmu^\top (\hat{\bsigma}+\bmu\bmu^\top + \frac{1}{a_{k+1}^\ast} \Q_k)^{-1} (\hat{\bsigma}+\bmu\bmu^\top + \frac{1}{\hat{a}_{k+1}} \Q_k)^{-1} \bmu \Big|, 
    \end{aligned} \end{equation*}
    where $| \frac{1}{\hat{a}_{k+1}}- \frac{1}{a_{k+1}^\ast} | \xrightarrow{a.s} 0$. 
    By Lemma \ref{lemma:param_a_nonneg}, we have $\hat a_k, a_k^\ast >0$ and $\Big| \x^\top (\hat{\bsigma} + \bmu\bmu^\top + \frac{1}{a_{k+1}^\ast} \Q_k)^{-1} (\hat{\bsigma} + \bmu\bmu^\top + \frac{1}{\hat{a}_{k+1}} \Q_k)^{-1} \y \Big|$ is uniformly bounded. 
    Therefore, $\Big| \x^\top [(\hat{\bsigma}+\bmu\bmu^\top + \frac{1}{{a}_{k+1}^\ast} \Q_k)^{-1} - (\hat{\bsigma}+\bmu\bmu^\top + \frac{1}{\hat{a}_{k+1}} \Q_k)^{-1}] \y \Big|$ will also converge to $0$ almost surely.
    
    Next, we focus on the convergence of $\x^\top (\hat{\bsigma} + \bmu\bmu^\top + \frac{\Q_k}{a_{k+1}^\ast} )^{-1} \y $, where we can still find a small enough $z_{k}$ such that $ \S_{k} = \bsigma^{-1/2} (\bmu\bmu^\top + \frac{1}{a_{k+1}^\ast}\Q_k) \bsigma^{-1/2} - z_{k} \I$ is positive definite.
    \begin{equation*}     \begin{aligned}
        \x^\top (\hat{\bsigma}+\bmu\bmu^\top + \frac{1}{a_{k+1}^\ast}\Q_k)^{-1} \y
        \xrightarrow{a.s} 
        \x^\top \big( \frac{\bsigma}{1+s_{k}} + \S_k + z_k \I \big)^{-1} \y \,,
    \end{aligned} \end{equation*}
    where $s_{k}$ is the solution of $s_{k} = \frac{c}{p} \tr \bsigma [(\frac{\bsigma}{1+s_{k} } + \frac{1}{a_{k+1}^\ast}\Q_k + \bmu \bmu^\top)^{-1}]$.
    
    Finally, we have the convergence of $\x^\top \hat\D_k^{-1} \y$:
    \begin{equation*}     \begin{aligned}
        \x^\top \hat\D_{k}^{-1} \y
        & = \hat{a}_{k+1}^{-1} \x^\top (\hat{\bsigma}+\bmu\bmu^\top + \frac{1}{\hat{a}_{k+1}} \Q_k)^{-1} \y\\
        & = \hat{a}_{k+1}^{-1} \x^\top [(\hat{\bsigma}+\bmu\bmu^\top + \frac{1}{\hat{a}_{k+1}} \Q_k)^{-1} - (\hat{\bsigma}+\bmu\bmu^\top + \frac{1}{a_{k+1}^\ast} \Q_k)^{-1} + (\hat{\bsigma}+\bmu\bmu^\top + \frac{1}{a_{k+1}^\ast} \Q_k)^{-1}] \y \\
        & = \hat{a}_{k+1}^{-1} \x^\top [(\hat{\bsigma}+\bmu\bmu^\top + \frac{1}{\hat{a}_{k+1}} \Q_k)^{-1} - (\hat{\bsigma}+\bmu\bmu^\top + \frac{1}{a_{k+1}^\ast} \Q_k)^{-1}] \y 
        + \hat{a}_{k+1}^{-1} \x^\top [(\hat{\bsigma}+\bmu\bmu^\top + \frac{1}{a_{k+1}^\ast} \Q
        _k)^{-1}] \y \\
        & \xrightarrow{a.s} (a_{k+1}^\ast)^{-1} \x^\top \big( \frac{\bsigma}{1+s_{k}}+\bmu\bmu^\top + \frac{1}{a_{k+1}^\ast} \Q_k \big)^{-1} \y := C_{\x,\y,k}^\ast \,.
    \end{aligned} \end{equation*}
    Then, by induction, we can show that for every $k = 0, 1, \ldots, T-1$, 
    \begin{equation*}
    \hat{a}_k 
    \xrightarrow{a.s} a_k^\ast 
    = a_{k+1}^{\ast} r^2 - (a_{k+1}^\ast)^2 r^2 C_{\bmu,\bmu,k}^\ast + 2 r a_{k+1}^\ast C_{\bmu,\Q_k\w_k,k}^\ast + \w_k^\top \Q_k \w_k -  C_{\Q_k\w_k,\Q_k\w_k,k}^\ast \,.
    \end{equation*}
    Therefore, the following convergence holds for every $k$.
    \begin{equation*}
    \x^\top \hat\D_k^{-1} \y
    \xrightarrow{a.s} (a_{k+1}^\ast)^{-1} \x^\top \Big( \frac{\bsigma}{1+s_k}+\bmu\bmu^\top + \frac{\Q_k}{a_{k+1}^\ast} \Big)^{-1} \y 
    := C_{\x,\y,k}^\ast \,.
    \end{equation*}

    The remaining part is the convergence of $\x^\top \hat\D_k^{-1} \bsigma \hat\D_k^{-1} \y$. 
    According to Lemma \ref{lemma:keylemma}, we can set $\a = \bsigma^{-1/2} \x$ and $\b = \bsigma^{-1/2} \y$ and show that 
    \begin{align*}
        m_n'(z) 
        & = - \a^\top (\frac{\X^\top \X}{n}+\S+z\I_p)^{-2}\b
        = - \x^\top \bsigma^{-1/2}(\frac{\X^\top \X}{n}+\S+z\I_p)^{-2} \bsigma^{-1/2} \y 
        = - (a_{k+1}^\ast)^2 \x^\top \hat\D_k^{-1} \bsigma \hat\D_k^{-1} \y \,,
    \end{align*}
    where the elements of $\X$ are iid with zero mean, variance 1 and finite eighth order moment. 
    Similarly to previous proofs, we have 
    \begin{align*}
        & \frac{\rm{d} [\a^\top (\frac{\I}{1+s_n(z)}+\S+z\I_p )^{-1}\b]}{\rm{d} z}
        = - \a^\top (\frac{\I}{1+s_n(z)}+\S+z\I_p )^{-1} (\I_p - \frac{s'(z)}{(1+s(z))^2}\I) (\frac{\I}{1+s_n(z)}+\S+z\I_p )^{-1}\b \\
        = & (\frac{s'(z)}{(1+s(z))^2}-1)  \Big(1-\frac{\tilde s_k}{(1+s_k)^2}\Big) \x^\top \Big(\frac{\bsigma}{1+s_k}+ \bmu \bmu^\top + \frac{\Q_k}{a_{k+1}^\ast}\Big)^{-1} \bsigma \Big(\frac{\bsigma}{1+s_k}+ \bmu \bmu^\top + \frac{\Q_k}{a_{k+1}^\ast} \Big)^{-1} \y \,,
    \end{align*}
    which completes the proof of 
    \begin{equation*}
    \x^\top \hat\D_k^{-1} \bsigma \hat\D_k^{-1} \y
    \xrightarrow{a.s} 
    (a_{k+1}^\ast)^{-2} \Big(1-\frac{\tilde s_k}{(1+s_k)^2}\Big) \x^\top \Big(\frac{\bsigma}{1+s_k}+ \bmu \bmu^\top + \frac{\Q_k}{a_{k+1}^\ast}\Big)^{-1} \bsigma \Big(\frac{\bsigma}{1+s_k}+ \bmu \bmu^\top + \frac{\Q_k}{a_{k+1}^\ast} \Big)^{-1} \y := B_{\x,\y,k}^\ast \,.
    \end{equation*}
\end{proof}

\subsubsection{Proof of Theorem~\ref{theorem:pram_converge_sig}}
\begin{proof}
Theorem~\ref{theorem:pram_converge_sig} contains two part: the convergence of the parameters and the limiting Sharpe ratio.  

We first prove the asymptotic properties of parameters. 
Since the convergence of $\hat a_k$s has been proved in \ref{proof:lemma3.2}, it is straightforward to apply Lemma \ref{lemma:sig_converge} to the parameters $\hat b_k$ and $\hat c_k$, 
which completes the proof of 
\begin{equation*}
\begin{aligned}
    \hat{a}_k & \xrightarrow{a.s} a_k^\ast 
    = a_{k+1}^{\ast} r^2 - (a_{k+1}^\ast)^2 r^2 C_{\bmu,\bmu,k}^\ast + 2 r a_{k+1}^\ast C_{\bmu,\Q_k\w_k,k}^\ast + \w_k^\top \Q_k \w_k -  C_{\Q_k\w_k,\Q_k\w_k,k}^\ast \,,
    \\
    \hat{b}_k & \xrightarrow{a.s} b_k^\ast 
    = b_{k+1}^\ast r - a_{k+1}^\ast b_{k+1}^\ast r C_{\bmu,\bmu,k}^\ast + b_{k+1}^\ast C_{\bmu,\Q_k\w_k,k}^\ast \,,
    \\
    \hat{{c}}_k & \xrightarrow{a.s} {c}_k^\ast 
    = {c}_{k+1}^\ast - (b_{k+1}^\ast)^2 C_{\bmu,\bmu,k}^\ast \,.
\end{aligned}
\end{equation*}

Now we focus on the convergence of Sharpe ratio. 
We start with the parameters in Proposition \ref{prop:eff_frontier}
\begin{equation*}
\begin{aligned}
    \hat{\lambda}_k 
    &= \frac{\hat{b}_0 - X_\tg}{\hat{c}_0} \hat{b}_{k+1} \bmu^\top \hat{\D}_k^{-1} \bmu \,,
    \\
    \hat\eta_k
    &= r- r \hat{a}_{k+1} \bmu^\top \hat{\D}_k^{-1}  \bmu + \bmu^\top \hat{\D}_k^{-1} \Q_k\w_k \,,
    \\
    \hat\nu_k 
    &=\Big\| \E[\hat X_k]\big( r \hat{a}_{k+1} \bmu -\Q_k \w_k \big) 
    -\frac{\hat{b}_0 - X_{\tg} }{ 2\hat{c}_0 } \hat{b}_k \bmu  \Big \|^2_{ \hat\D_k^{-1} \bsigma \hat\D_k^{-1} } \,,
    \\ 
    \hat{g}_k 
    & = \Big\|  ra_{k+1} \bmu - \Q_k \w_k  \Big\|^2_{ \hat\D_k^{-1} \bsigma \hat\D_k^{-1} } + \hat{\eta}_k^2 \,,
\end{aligned}
\end{equation*}
where $\E[\hat{X}_{k+1}] = \prod_{i=0}^{k} \hat{\eta}_i + \sum_{i=0}^k(\hat{\lambda}_{i}\prod_{j=i+1}^{k}\hat{\eta}_{j})$ and we assume $\hat\bmu = \bmu$ in this Proposition. 
By Lemma \ref{lemma:sig_converge}, we can show that
\begin{equation*}
\begin{aligned}
    \hat\lambda_k
    & \xrightarrow{a.s} \lambda_k^\ast 
    := \frac{b_0^\ast - X_\tg}{c_0^\ast} b_{k+1}^\ast C_{\bmu,\bmu,k}^\ast  \,,
    \\
    \hat\eta_k
    & \xrightarrow{a.s} \eta_k^\ast 
    :=r - r a_{k+1}^\ast C_{\bmu,\bmu,k}^\ast + C_{\bmu,\Q_k\w_k,k}^\ast  \,,
    \\
    \E[\hat{X}_{k+1}]
    & \xrightarrow{a.s} \E[\hat{X}_{k+1}]_\infty
    = \prod_{i=0}^{k} \eta_i^\ast + \sum_{i=0}^k( \lambda_i^\ast \prod_{j=i+1}^{k} \eta_j^\ast) \,.
\end{aligned}
\end{equation*}
Let $\phi_k^\ast = r a_{k+1}^\ast E[\hat X_k]_\infty + \frac{b_0^\ast - X_\tg}{{c}_0^\ast} b_{k+1}^\ast$, following Lemma \ref{lemma:sig_converge}, we have 
\begin{equation*}
\begin{aligned}
    \hat\nu_k \xrightarrow{a.s} \nu_k^\ast 
    = & \E[\hat X_{k}]_\infty^2 [ r^2 (a_{k+1}^\ast)^2 B_{\bmu,\bmu,k}^\ast + B_{\Q_k\w_k,\Q_k\w_k,k}^\ast - 2 r a_{k+1}^\ast B_{\bmu,\Q_k\w_k,k}^\ast ] \\
    & ~~~~~~ + 2 E[\hat X_k]_\infty \frac{b_0^\ast - X_\tg}{{c}_0^\ast} b_{k+1}^\ast [r a_{k+1}^\ast B_{\bmu,\bmu,k}^\ast - B_{\bmu,\Q_k\w_k,k}^\ast ] 
    + (\frac{b_0^\ast - X_\tg}{{c}_0^\ast} b_{k+1}^\ast)^2 B_{\bmu,\bmu,k}^\ast
    \\
    = & ( \phi_k^\ast B_{\bmu,\bmu,k}^\ast )^2 + \E[\hat X_{k}]_\infty^2 B_{\Q_k\w_k,\Q_k\w_k,k}^\ast - 2 \phi_k^\ast E[\hat X_k]_\infty B_{\bmu,\Q_k\w_k,k}^\ast \,,
    \\
    \hat g_k \xrightarrow{a.s} g_k^\ast 
    = &r^2 (a_{k+1}^\ast)^2 B_{\bmu,\bmu,k}^\ast + B_{\Q_k\w_k,\Q_k\w_k,k}^\ast - 2 r a_{k+1}^\ast B_{\bmu,\Q_k\w_k,k}^\ast + (\eta_k^\ast)^2 \,.
\end{aligned}
\end{equation*}
The asymptotic Sharpe ratio of the $\RRMV$ portfolio, given $\Q_k$ and $\w_k$, will almost surely converge to
\begin{align*} 
    \SR_\infty(T, \Q_k, \w_k)
    = 
    \frac{X_\tg - r^T}{\sqrt{ T \sum_{i=0}^{T-1} \nu_i^\ast \prod_{j=i+1}^{T-1}{g}_j^\ast} } \,,
\end{align*}
where $\E[\hat X_T] = X_\tg$ is from our assumption of $\hat\bmu = \bmu$.
\end{proof}

\subsubsection{Special Case in Section \ref{sec3.2.2}}
\label{ssse:est_sig_Q_0}
~\\
Similar to Section~\ref{sec3.2.1}, we consider $\Q_k \to \mathbf{0}$ and $\w_k = \mathbf{0}$ separately. 
The limiting values in Lemma \ref{lemma:sig_converge} will be 
\begin{equation*}
\begin{aligned}
    C_{\x,\y,k}^\ast 
    &= (a_{k+1}^\ast)^{-1} \x^\top \big( \frac{\bsigma}{1+s}+\bmu\bmu^\top \big)^{-1} \y \,,
    \\
    B_{\x,\y,k}^\ast 
    &= (a_{k+1}^\ast)^{-2} (1-\frac{\tilde s}{(1+s)^2}) \x^\top (\frac{\bsigma}{1+s}+ \bmu \bmu^\top)^{-1} \bsigma (\frac{\bsigma}{1+s}+ \bmu \bmu^\top)^{-1} \y \,,
\end{aligned}    
\end{equation*}
where $s$ is the solution of $s = \frac{c}{p} \tr \bsigma [(\frac{\bsigma}{1+s} + \bmu \bmu^\top)^{-1}]$ 
and 
$\tilde s$ is the solution to $\tilde s = (\frac{\tilde s}{(1+s)^2} - 1) \frac{c}{p} \tr [\bsigma^2 (\frac{\bsigma}{1+s} + \bmu \bmu^\top)^{-2}]$.
And the parameter will be 
\begin{equation*}
\begin{aligned}
    a_k^\ast 
    &= r^2 a_{k+1}^{\ast} - r^2 (a_{k+1}^\ast)^2 C_{\bmu,\bmu,k}^\ast + 2 r a_{k+1}^\ast C_{\bmu,\Q_k\w}^\ast -  C_{\Q_k\w,\Q_k\w,k}^\ast \,,
    \\
    b_k^\ast 
    &= b_{k+1}^\ast r - r a_{k+1}^\ast b_{k+1}^\ast C_{\bmu,\bmu,k}^\ast + b_{k+1}^\ast C_{\bmu,\Q_k\w,k}^\ast \,,
    \\ 
    c_k^\ast 
    &= c_{k+1}^\ast -(b_{k+1}^\ast)^2 C_{\bmu,\bmu,k}^\ast \,.
\end{aligned}    
\end{equation*}
By plugging in the limiting values, we have the asymptotic Sharpe ratio. 
To illustrate the results numerically, we use the same simulation as in previous section. 
Similar to the single-period case, a necessary condition for optimal reference weight is $\w \propto \bmu$. 
So, we let $\w = \alpha \bmu$ and check how $\alpha$ will affect the limiting Sharpe ratio. 

\subsubsection{Proof of Corollary \ref{corr:multi_SR_sig}}
\begin{proof}
In this Corollary, we set $\w = \mathbf{0}$ and  $\Q_k = a_{k+1} \rho \bar\Q$. 

By Corollary \ref{coro:os_SR_w=0}, we have
\begin{align*}
    \SR^2(\hat{X}_T) 
    &= \frac{1}{T} \frac{(\E[X_T] - r^T)^2}{\Var(X_T)} 
    = \frac{1}{T} \frac{(\hat \eta^T -r^T)^2}{((\hat\Gamma + \hat\eta^2)^T - \hat\eta^{2T})}
\end{align*}
where we have the following parameters when $\hat\bmu = \bmu$,
\begin{equation*}
\begin{aligned}
    \hat\eta 
    &= r(1-\frac{\bmu^\top(\hat\bsigma + \rho \Q)^{-1} \bmu}{1 + \bmu^\top (\hat\bsigma + \rho \Q)^{-1} \bmu}) \,,
    \\
    \hat\Gamma 
    &= r^2 \frac{\bmu^\top (\hat{\bsigma} + \rho \Q)^{-1} \bsigma (\hat{\bsigma} + \rho \Q)^{-1} \bmu}{(1+\bmu^\top (\hat{\bsigma} + \rho \Q)^{-1} \bmu)^2} \,.
\end{aligned}    
\end{equation*}
We can show the following almost surely convergence
\begin{equation*}
\begin{aligned}
    \bmu^\top (\hat\bsigma + \rho \Q)^{-1} \bmu 
    & \xrightarrow{a.s} \bmu^\top \big(\frac{\bsigma}{1+s(\rho)} + \rho\Q \big)^{-1} \bmu,\\
    \|\bsigma^{1/2} (\hat{\bsigma} + \rho \Q)^{-1} \bmu \|_2^2 \,,
    & \xrightarrow{a.s} \Big(1 - \frac{\tilde s(\rho)}{(1+s(\rho))^2} \Big) \bmu^\top \Big(\frac{\bsigma}{1+s(\rho)} + \rho\Q \Big)^{-1} \bsigma \Big(\frac{\bsigma}{1+s(\rho)} + \rho\Q \Big)^{-1}\bmu \,.
\end{aligned} 
\end{equation*}
Therefore, the Sharpe ratio will almost surely converge to $\SR_\infty(T, \rho, \mathbf{0})$,
\begin{align*}
   \SR_\infty(T, \rho, \mathbf{0}) 
   = 
   \Big(1 - \frac{\tilde s(\rho)}{(1+s(\rho))^2} \Big)^{-T/2} \left( \frac{1}{T} \frac{((\bmu^\top \big(\frac{\bsigma}{1+s(\rho)} + \rho\bar\Q \big)^{-1} \bmu + 1)^T - 1)^2}{ (\bmu^\top \Big(\frac{\bsigma}{1+s(\rho)} + \rho\bar\Q \Big)^{-1} \bsigma \Big(\frac{\bsigma}{1+s(\rho)} + \rho\bar\Q \Big)^{-1}\bmu + 1)^T - 1} \right)^{1/2} \,,
\end{align*}
where $s(\rho)$ is the solution of $s(\rho) = \frac{c}{p} \tr \bsigma [(\frac{\bsigma}{1+s(\rho)} + \rho \bar\Q + \bmu \bmu^\top)^{-1}]$, and $\tilde s(\rho)$ is given as the solution of $\tilde s(\rho) = (\frac{\tilde s(\rho)}{(1+s(\rho))^2} - 1) \frac{c}{p} \tr [\bsigma^2 (\frac{\bsigma}{1+s(\rho)} + \rho \bar\Q + \bmu \bmu^\top)^{-2}]$. 

Let $\kappa_\rho = 1 - \frac{\tilde s(\rho)}{(1+s(\rho))^2}$, $\A_\rho = \frac{\bsigma}{1+s(\rho)} + \rho \bar\Q$, we complete the proof of 
\begin{align*}
   \SR_\infty(T, \rho, \BFzero) 
   = \kappa_\rho^{-T/2} \cdot 
   \frac{ T^{-1/2}  ((\bmu^\top \A_\rho^{-1} \bmu + 1)^T - 1) }{\sqrt{(\bmu^\top \A_\rho^{-1} \bsigma \A_\rho^{-1}\bmu + 1)^T - 1}}\,,
\end{align*}
where $s(\rho)$ is the solution of $s(\rho) = \frac{c}{p} \tr \bsigma ((\A_\rho + \bmu \bmu^\top)^{-1})$, 
and $\tilde s(\rho)$ is given as the solution of $\tilde s(\rho) = -\kappa_\rho \frac{c}{p} \tr(\bsigma^2 (\A_\rho + \bmu \bmu^\top)^{-2})$. 
\end{proof}

\section{Time-Consistent Policy for RRMV Problem}\label{se:tc_policy} 
In this section, we develop the time-consistent policy for the problems $\cP_{\RRMV}$ and $\cP^{\dag}_{\RRMV}(\omega)$. 
To distinguish the TC policies with the PC policies, we denote the optimal TC policies of these two problems as  $\pi_{tc}^*(\cP^{\dagger}_{\RRMV}(\omega))$ and $\pi_{tc}^*(\cP_{\RRMV})$. 

When the portfolio includes the risk-free asset, we introduce the following parameters. 
Recall that for the excessive return $\P_k$, we have $\bsigma_k = \textrm{Cov}[\P_k]$ and $\bmu_k = \E[\P_k]$,
\begin{equation} \label{def_tc_rf_VMG}
    \begin{aligned}
        & {\V}_{k} = \Var\bigg(\prod_{i={k+1}}^{T-1}(r_i + \P_i^\top {\h}_i)\bigg) \\
        & {\M}_k = \bigg(\prod_{i={k+1}}^{T-1}(r_i + \bmu_i^\top {\h}_i)^2\bigg) \bsigma_{k} + {\V}_k (\bsigma_{k}+ \bmu_k \bmu_k^\top) + \Q_k \\
        & {\G}_k = {\M}_{k}^{-1} \bmu_{k} [\prod_{i={k+1}}^{T-1}(r_i + \bmu_i^\top {h}_i)] \\
        &{\h}_k = {\M}_{k}^{-1} \Q_k  \w_k - r_k {V}_k {\M}_k^{-1} \bmu_k
    \end{aligned}
\end{equation}
for $k=T-2,\ldots,1,0$, and ${\V}_{T-1} = 1$, ${\M}_{T-1} = \bsigma_{T-1} +  \Q_{T-1} $ and ${\G}_{T-1} = {\M}_{T-1}^{-1} {\mu}_{T-1}$.

\begin{theorem}\label{theorem:sol_RRMV_dag_tc_rf}
The time-consistent policy for the problem ${\cP}^{\dag}_{\RRMV}(\omega)$ is
\begin{align}
{\u}_k^* = \frac{1}{2\omega} {\G}_k + {\h}_k X_k \label{def_RRMV_dag_tc_rf_uk}
\end{align}
for $k=0,\ldots,T-1$, where the parameters $\G_k$ and $\h_k$ are defined in (\ref{def_tc_rf_VMG}). Furthermore, under the optimal policy ${\u}_k^\ast$, the expectation of the terminal wealth $X_T$ can be expressed as 
\begin{align}
\E[X_T] = \frac{1}{2\omega} \sum_{i=0}^{T-1} \Big((\bmu_i^\top {\G}_i)\prod_{j=i+1}^{T-1} (r_j + \bmu_j^\top {\h}_j)\Big)  + \prod_{i=0}^{T-1} ( r_i + \bmu_i^\top {\h}_i). \label{def_tc_rf_ExT}
\end{align}
When $\omega$ is set at
\begin{align}
\omega = \omega^* = \frac{\sum_{i=0}^{T-1} [(\prod_{j=i+1}^{T-1}(r_j + \bmu_j^\top {\h}_j))(\bmu_i^\top {\G}_i)]}{2(X_\tg - \prod_{i=0}^{T-1} ( r_i + \bmu_i^\top {\h}_i)},\label{def_omega_start_tc_rf}
\end{align}
the policy (\ref{def_RRMV_dag_tc_rf_uk}) is also the time-consistent policy for ${\cP}_{\RRMV}$.
\end{theorem}

\subsection{Proof of Theorem \ref{theorem:sol_RRMV_dag_tc_rf}}
\begin{proof}
    We first derive the time-consistent policy at $t=T-1$. 
    We can easily verify that an investor follows the optimal policy ${\u}_t^\ast = \frac{1}{2\omega} {\G}_t + {\h}_t X_t$ when $t = T-1$.
    By induction, we assume investor follows the same optimal policy for all $t = k, \ldots, T-1$. 
    Then at $t=k$, we have that:
    \begin{equation*}
        \begin{aligned}
            {\u}_k^\ast &= \frac{1}{2\omega} {\G}_k + {\h}_k X_k\\
            X_{k+1} &= r_k X_k + \P_k^\top {\u}_k = \frac{1}{2\omega} \P_k^\top {\G}_k + (r_k + \P_k^\top {\h}_k) X_k\\
        \end{aligned}
    \end{equation*}
    And at time $t=k+1$, we have:
    \begin{equation*}
        \begin{aligned}
            {\u}_{k+1}^\ast &= \frac{1}{2\omega} {\G}_{k+1} + {\h}_{k+1} X_{k+1}\\
            X_{k+2} &= r_{k+1} X_{k+1} + \P_{k+1}^\top (\frac{1}{2\omega} {\G}_{k+1} + {\h}_{k+1} X_{k+1}) \\
            &= \frac{1}{2\omega} \P_{l+1}^\top {\G}_{k+1} + \frac{1}{2\omega} \P_k^\top {\G}_k (r_{k+1} + \P_{k+1}^\top {\h}_{k+1}) + (r_k +\P_k^\top {\h}_k) (r_{k+1} + \P_{k+1}^\top {\h}_{k+1}) X_k\\
        \end{aligned}
    \end{equation*}
    Similarly, the terminal wealth can be expressed by $X_k$ in:
    \begin{equation*}
        \begin{aligned}
            X_T &= \frac{1}{2\omega} {K}_k + \prod_{i=k}^{T-1}(r_i + \P_i^\top {\h}_i)X_k\\
        \end{aligned}
    \end{equation*}
    where ${K}_k$ includes the terms which are irrelevant to $X_k$ 
    Therefore, the variance can be computed by conditioning on the $X_k$:
    \begin{equation*}
        \begin{aligned}
            \E[X_T | X_k] 
            &= \frac{1}{2\omega} E[{K}_k | X_k] + [\prod_{i=k}^{T-1}(r_i + \bmu_i^\top {h}_i)]X_k \\
            \E[X_T] 
            &= \frac{1}{2\omega} E[{K}_k | X_k] + [\prod_{i=k}^{T-1}(r_i + \bmu_i^\top {h}_i)] (r_{k-1}X_{k-1} + \bmu_{k-1}^\top {\u}_{k-1}) \\
            \Var(X_T | X_k) 
            &= \Var(\prod_{i=k}^{T-1}(r_i + {\P}_i^\top {h}_i)) X_k^2 = {V}_{k-1} X_k^2 \\
            \Var(X_T) 
            &= [\prod_{i=k}^{T-1}(r_i + \bmu_i^\top {h}_i)]^2 \Var(X_k) + {V}_{k-1} \E[X_k^2]\\
            &= [\prod_{i=k}^{T-1}(r_i + \bmu_i^\top {h}_i)^2] {\u}_{k-1}^\top \hat{\bsigma} {\u}_{k-1} + {V}_{k-1} (r_{k-1}^2 X_{k-1}^2 + 2 r_{k-1} X_{k-1} \bmu^\top_{k-1} {\u}_{k-1} + {\u}_{k-1}^\top \E[\P_{k-1} \P_{k-1}^\top] {\u}_{k-1})\\
        \end{aligned}
    \end{equation*}
    The optimization problem at time $k-1$ is:
    \begin{equation*}
        \begin{aligned}
            \min ~ & \omega \Var(X_T) - \E[X_T] + \omega \| {\u}_{k-1} - X_{k-1} \w_{k-1}  \|_{\Q_{k-1}}^2\\
            s.t.
            ~&~     X_{t+1} = r_t X_t + \P_t ^\top {\u}_t, ~~ t =k-1,\dots,T\\
        \end{aligned}
    \end{equation*}
    After substituting the optimal policy $\u_i^\ast$ at time $i=k-1,\dots,T$, we can derive the first order optimal condition: 
    \begin{equation*}
        \begin{aligned}
            2 \omega {\M}_{k-1} {\u}_{k-1} + 2\omega r_{k-1} {V}_{k-1} X_{k-1} \bmu_{k-1} - [\prod_{i=k}^{T-1}(r_i + \bmu_i^\top {h}_i)] \bmu_{k-1} - 2 \omega X_{k-1} \Q_{k-1} \w_{k-1}   = 0
        \end{aligned}
    \end{equation*}
    Then the optimal policy at time $k-1$ is:
    \begin{equation*}
        \begin{aligned}
            {\u}_{k-1}^\ast 
            &= X_{k-1} {\M}_{k-1}^{-1} \Q_{k-1}  \w_{k-1} - X_{k-1} r_{k-1} {V}_{k-1} {\M}_{k-1}^{-1} \bmu_{k-1} + \frac{1}{2\omega} {\M}_{k-1}^{-1} \bmu_{k-1} [\prod_{i=k}^{T-1}(r_i + \bmu_i^\top {h}_i)] \\
            &= \frac{1}{2 \omega} {\G}_{k-1} + {\h}_{k-1} X_{k-1}\\
        \end{aligned}
    \end{equation*}
    
    Now we prove the results of the expected return $\E[X_T]$.
    Since the optimal policy is that ${\u}_t^\ast = \frac{1}{2\omega} {\G}_t + {\h}_t X_t$, we can easily find out that 
    $\E[X_1] = \frac{1}{2\omega}\bmu_0^\top {\G}_0 + r_0 + \bmu_0^\top {\h}_0$ 
    and $\E[X_2] = \frac{1}{2\omega}\bmu_1^\top {\G}_1 + \frac{1}{2\omega} (r_1 + \bmu_1^\top {\h}_1) \bmu_0^\top {\G}_0 + (r_1 + \bmu_1^\top {\h}_1) (r_0 + \bmu_0^\top {\h}_0)$
    , so the expression is valid when $t=2$. 
    
    Then by induction, we assume that it is true when $t=k$, which is:
    \begin{equation*}
        \begin{aligned}
            \E[X_k] &= \frac{1}{2\omega} \sum_{i=0}^{k-1} [(\bmu_i^\top {\G}_i)\prod_{j=i+1}^{k-1} (r_j+\bmu_j^\top {\h}_j)] + \prod_{i=0}^{k-1} (r_i + \bmu_i^\top {\h}_i)\\
        \end{aligned}
    \end{equation*}
    Therefore, at time $t=k+1$:
    \begin{equation*}
        \begin{aligned}
            \E[X_{k+1}] 
            &=\frac{1}{2\omega} \bmu_k^\top {\G}_k + (r_k + \bmu_k^\top {\h}_k) \E[X_k]\\
            &=\frac{1}{2\omega} \bmu_k^\top {\G}_k + (r_k + \bmu_k^\top {\h}_k) \{ \frac{1}{2\omega} \sum_{i=0}^{k-1} [(\bmu_i^\top {\G}_i)\prod_{j=i+1}^{k-1} (r_j+\bmu_j^\top {\h}_j)] + \prod_{i=0}^{k-1} (r_i + \bmu_i^\top {\h}_i)\}\\
            &=\frac{1}{2\omega} \sum_{i=0}^{k} [(\bmu_i^\top {\G}_i)\prod_{j=i+1}^{k} (r_j+\bmu_j^\top {\h}_j)]  + \prod_{i=0}^{k} (r_i + \bmu_i^\top {\h}_i)
        \end{aligned}
    \end{equation*}
\end{proof}  

\end{document}